\documentclass[11pt]{article} 

\usepackage{amsmath,amsthm,latexsym,amssymb,amsfonts,epsfig}

\newtheorem{Definition}{Definition}[section]
\newtheorem{Lemma}{Lemma}[section]
\newtheorem{Corollary}{Corollary}[section]
\newtheorem{Proposition}{Proposition}[section]

\newtheorem{Example}{Example}[section]

\newcommand{\be}{\begin{equation}}
\newcommand{\ee}{\end{equation}}
\newcommand{\ba}{\begin{eqnarray}}
\newcommand{\ea}{\end{eqnarray}}

\title{{\sf Non-perturbative Quantum Gravity}\\
 {\sf in Fock representations}} 
\author{
{\sf T. Thiemann}$^1$\thanks{{\sf 
thomas.thiemann@gravity.fau.de}}\\
\\
{\sf $^1$ Inst. for Quantum Gravity, FAU Erlangen -- N\"urnberg,}\\
{\sf Staudtstr. 7, 91058 Erlangen, Germany}\\
}
\date{{\small\sf \today}}

\makeatletter
\@addtoreset{equation}{section}
\makeatother

\begin{document} 

\maketitle

{\sf

\begin{abstract}
Perturbative quantum gravity starts from prescribing a background 
metric. That background metric is then used in order to carry out 
two separate steps: 1. One splits the non-perturbative metric into 
background and deviation from it (graviton) and expands the action 
in terms of the graviton which results in an ifinite series of 
unknown radius of convergence. 2. One constructs a Fock representation for 
the graviton and performs perturbative graviton quantum field theory on the 
fixed background as dictated by the perturbative action. 
The result is a non-renormalisable theory without predictive power.   

It is therefore widely believed that a 
non-perturbative approach is mandatory in order to construct a 
fundamental, not only effective, predictive quantum field theory of the 
gravitational interaction. Since perturbation theory is by definition 
background dependent, the notions of background dependence (BD) and 
perturbation theory (PT) are often considered as symbiotic, as if they 
imply each other.

In the present work we point out that there is no such symbiosis, 
these two notions are in fact logically independent. 
In particular, one {\it can} use BD structures
while while not 
using PT at all. Specifically, we construct BD Fock representations (step 2
above) 
for the full, {\it non-perturbative} metric rather than the graviton 
({\it not} step 1 above) 
and therefore
never perform a perturbative expansion. Despite the fact that the 
gravitational Lagrangean is a non-polynomial, not even analytic, 
function of the metric we show that e.g. the Hamiltonian constraint
with any density weight 
can be defined as a quadratic form with dense form domain in such a 
representation.     
\end{abstract}

\section{Introduction}
\label{s1}

Today we have not yet constructed an interacting  
matter quantum field theory (QFT) in four 
dimensional Minkowski space that obeys the Wightman axioms
\cite{1}. Still, elementary particle physics has been spectacularly 
successful in using perturbation theory based on Fock representations 
for the free (i.e. non-interacting) QFT which can be constructed 
in any dimension. The chosen Fock representation is subordinate
to a split of the non-perturbative Hamiltonian $H=H_0+V$, which is 
a polynomial function of the fields, into a free piece $H_0$ which is 
a quadratic polynomial and the interaction $V$ which is a polynomial 
of higher degree. While by construction $H_0$ is an operator on the 
free Fock space ${\cal H}_0$ with dense invariant domain ${\cal D}_0$
given by the span of Fock states,
by Haag's theorem \cite{2} the interaction term $V$ is typically 
not an operator but only a quadratic form with dense form domain
${\cal D}_0$. That is to say that $H$ is not an operator on ${\cal H}_0$
but merely a qudratic form which is at the heart of the problem of 
giving meaning to powers of $H$ as they appear in the would be unitary 
time evolution operator $e^{itH}$ or the would be Moeller operator
$e^{it H}\; e^{-itH_0}$ which is employed  
in the associated scattering matrix $S$. To define it as an operator 
would lead to a Wightman QFT and would require a representation 
subordinate to $H$ and not to $H_0$ on a Hilbert space $\cal H$.  

If one tries to construct the S-matrix nevertheless 
by Taylor expanding with respect to $V$ as a quadratic 
form, one meets divergences that come from products of quadratic forms
which is an ill-defined mathematical operation. These divergences can 
be tamed using regulariation and removed by absorbing them into 
the monomial coefficients of the Hamiltonian. If this can be done 
order by order for a {\it finite} number of monomials one calls the theory 
perturbatively renormalisable. If this requires an infinite number of 
monomials one calls the theory perturbatively non-renormalisable. 
In the non-renormalisable case one calls the theory truncated to a 
finite number of monomials effective if one can argue that the neglected 
terms do not play any role up to a given energy resolution.
Even in the renormalisable case there is no guarantee that the perturbation
series converges in any sense. Importantly, one can do this also in curved 
spacetime (CST) i.e. for a general Lorentzian signature metric, not necessarily
the Minkowski metric \cite{3}.

In quantum gravity one has tried to apply this programme to the gravitational 
field. A first difference with the matter field sector is that while 
the natural ``background matter'' field is the ``zero'' field, there is no
such natural background for the metric. Therefore in a first step 
one splits the full metric $g$ into a fixed background $g_0$ and 
a graviton deviation $h=g-g_0$. Then one expands the gravitational
Lagrangean in powers of $h$ at $g=g_0$. Since the gravitational Lagrangean 
is not polynomial, not even analytic, in $g$, the expansion is an infinite 
series with an unknown radius of convergence even classically. 
In a second step one collects the terms quadratic in $h$ into a free 
part (the analog of $H_0$) and the interacting rest (the analog of $V$). 
This must be accompanied by fixing the coordinate gauge freedom. 
One can then use the free part to construct a Fock representation for 
$h$ adapted to it and $g_0$. 
The rest of the programme then proceeds formally the same way as for 
matter QFT in the CST encoded by $g_0$. Unfortunately the resulting 
theory is perturbatively non-renormalisable \cite{4} in the 
following sense: the perturbative expansion of a Lagrangian in terms of 
$h$ delivers n-th order monomial interaction terms $L_n$. 
One can now ask whether the counter terms triggered by $L_n$ are found 
in the list $\{L_m\}_{m=0}^\infty$. If this is true for all $n\le N$ 
and if this leaves no more than a finite number of free parameters 
then the theory would be renormalisable. However for a given $L$ 
(say the Einstein-Hilbert Lagrangian) one finds even for relatively low 
orders $n$ counter terms $L'_m$ which belong to a different Lagrangian 
$L'$ i.e. are not in the allowed list. This means that one needs to 
generalise the Lagrangian introducing more couplings 
and it is believed that this process never ends and results in a theory with 
infinitely many free parameters.   

It is therefore widely believed that a fundamental QFT of gravity must 
be constructed using a non-perturbative approach (NPA). Examples for 
such programmes are, in alphabetical order, 
asymptotically safe quantum gravity (ASQG) \cite{6}, 
causal dynamical triangulations (CDT) \cite{7}, causal set theory 
(CST) \cite{8}, loop quantum gravity (LQG) \cite{9} and 
string theory (ST) in the AdS/CFT incarnation \cite{10}. See also
\cite{11} for a confrontation of these approaches in a single volume.

Since the perturbative approach to QG as sketched above uses a background 
metric $g_0$ as a starting point and since it has failed to construct 
a fundamental theory, many quantum geometers take background independence
(BI) as a profound, almost holistic, 
guiding principle to construct a fundamental theory of QG.
This would indeed be a necessary step, if the notions of background dependence 
(BD) and perturbative approach (PA) would be truly symbiotic, i.e. 
if indeed BD would imply a PA and vice versa. However, we would like 
to point our that {\it there is no such symbiosis.} Indeed, while one cannot
A. perform perturbation theory without B. using a background metric, 
one can B. use a background metric without A. performing perturbation theory.
Thus in terms of logical implication we have 
$A\;\Rightarrow B$ but  $B\;\not\Rightarrow A$. That is, we may use 
a backround metric $g_0$ to perform various QFT constructions but 
without using $g_0$ to split the non-perturbative metric $g$ into 
$g=g_0+h$. Without such a split, the Lagrangean is not expanded 
in powers of of $h$ and there is {\it no perturbation series at 
the classical level.}.

In this paper we will use a background metric $g_0$ to construct 
a Fock representation, but not of the perturbation field $h=g-g_0$ but
of the non-perturbative field $g$. In fact, even in perturbative QG
one could have used a Fock representation for $h$ based on a background 
metric $g_1$ which could be different from $g_0=g-h$ although there it 
would be an unnecessary, even disadvantageous (because the quadratic 
part of the Hamiltonian is an operator rather than a quadratic 
form only if $g_1=g_0$) blow-up of background structure. As a consequence 
the algebraic state (vacuum expectation value) underlying this Fock space 
structure is in particular not diffeomorphism invariant. This means that 
not a single Fock vector state, not even the vacuum 
(i.e. the vector annihilated by 
the annihilation operators) is diffeomorphism invariant. This appears to be 
in rather strong conflict in representing the gauge transformations of 
theory on the corresponding Hilbert space in the operator constraint approach 
and various operators such 
the Weyl operators that define the Fock representation and constraint
operators that encode the quantum Einstein equations. However we remind 
of the phenomenon of symmetry breaking \cite{12}: It is perfectly possible 
that an operator $H$ is invariant under a unitary symmetry $U$ while 
none of its ground states is (in that case the ground state is 
necessarily degenerate). With respect to the diffeomorphism group      
and the generators of diffeomorphisms (diffeomorphism constraints) it 
is perfectly posible that the generators are diffeomorphism co-variant 
while no state is invariant.

Such a background dependent representation is of any practical use 
only if one succeeds to formulate the quantum Einstein equations in 
mathematically rigorous way. In LQG one has gone rather far with 
a manifestly {\it background independent} representation of (a version 
tailored to non-abelian gauge theories of) 
the Weyl operators \cite{13}: There is a unique 
such representation of Narnhofer-Thirring type \cite{14}
based on an inavariant state $\omega$ 
on the corresponding Weyl algebra
\cite{15}. The spatial diffeomorphism group consequently acts by 
unitary operators. However, one must pay a high price for this: The 
state $\omega$ is irregular, one cannot construct the spatial 
metric and its conjugate momentum simultaneously as operators. Similarly
one cannot construct the generators of the spatial diffeomorphism
group as operators. One can get away with this because one can 
at least obtain a representation of the spatial diffeomorphism group.
However, the generators of the classical Einstein equations do not 
form a Poisson Lie algebra and thus Lie group techniques cannot be used
to construct the exponentiated temporal generator, called the Hamiltonian
constraint. One has to construct it directly in non-exponentiated form 
which requires regularisation and operator ordering choices and 
removal of the regulator based on spatial diffeomorphism invariance 
\cite{16}. However, not all regularisation choices are washed out 
when the regulator is removed and leaves the Hamiltonian constraint 
with quantisation ambiguities. Furthermore, while the algebra of constraint
operators closes with structure operators ordered to the right, these 
structure operators are not quantisations of the classical structure 
functions. Thus while there is no mathematcial anomaly, there is 
a physical one. These complications can be overcome in the 
Abelian truncation of vacuum gravity \cite{16a} but for a different 
Weyl algebra and with a different dense invariant domain of the 
constraint operators.  

To avoid this anomaly in LQG, one can replace the Hamiltonian 
constraint by the spatially diffeomorphism invariant master constraint 
\cite{17} or one can perform a reduced phase space quantisation \cite{18}
or one can try to define the constraints on a linear 
dual space of distributions 
without Hilbert space structure \cite{16b}.
However, another annoying consequence of the 
irregularity of the state which penetrates to all three of these 
proposals  
is the fact the resulting Hilbert space is not separable. This again 
leads to quantisation ambiguities. It is these ambiguities and their removal
that have motivated Hamiltonian renormalisation methods \cite{19} and also
the present work.

Fock representations are based on regular pure states on the Weyl algebra 
and thus we can define the spatial metric and its conjugate momentum 
as operator valued distributions. Moreover, the associated representation
of the Weyl algebra is irreducible, the underlying vacuum expectation 
value state is pure, every vector is cyclic for it.
As we will show, in a non-perturbative Fock representation one cannot
easily define the constraints as {\it operators}, but one can define them 
as quadratic forms on a suitable dense form domain. There are 
no quantisation ambiguities, one {\it must} choose normal ordering. 
Moreover, although quadratic forms cannot be multiplied, 
in a recent contribution \cite{29} we showed  
that for the generators of the spatial diffeomorphisms (spatial 
diffeomorphism constraints) one can define a {\it commutator} 
and check that it is free of physical anomalies.  
In this work we will show that given any representation, not 
necessaily in the Fock class, quadratic forms are 
sufficient to rigorously formulate and solve the Lorentzian or 
Euclidian quantum Einstein 
equations. Moreover, we will show 
that the commutator algebra can be meaningfully defined 
at least for Euclidian signature in background dependent 
Fock representations when the constraints are 
written in polynomial form. For the Lorentzian signature case 
this might also be possible by exactly the same technique 
but the computations, while straightforward, are much more complex
and have not been finished yet. 

As this works for a huge class of Fock representations one 
may ask which one to choose and whether this does not 
correspond to yet another source of ambiguity. However, note first of all 
that this choice of background structure $B_0$ consisting 
of a choice of background metric $g_0$ and a Fock state 
$\omega_0$ based on it must also be 
made in the perturbative approach, this ``background ambiguity'' is in addition
to the ambiguities that one encounters in the renormalisation 
process which requires to add an infinite number of counter terms 
corresponding to all possible curvature invariants. In the non-perturbative 
approach, if it can be accomplished, the 
counter term ambiguity is missing, which is an improvement. Next note 
that the background structure ambiguity $B_0$
is of the same quality as the representation ambiguity that 
one encounters in QFT in CST. The difference with QFT in CST is 
that there one neglects quantum gravity altogether and the background 
$g=g_0$ is considered as fixed.
But even then there are infinitely 
many choices and one must downsize them by imposing additional physical 
requirements such as the Hadamard condition on the 2-point function 
\cite{20} which states that the short distance behaviour should be 
the same as that of the usual Fock representations of Minkowski space 
singled out by Poincar\'e invariance. 
In quantum gravity one can use either the operator 
constraint or the reduced phase space approach. 
In the operator constraint approach  
the choice of representation 
is likely not very important because 
what one is interested in is not the states of that representation but rather
the solutions of the constraint equations which are distributions over 
a dense subspace in the representation space and that space of distributions 
could be independent of the choice of $g_0$. Of course 
the space of solutions must be equipped with a Hilbert space structure by 
itself which opens a new representation problem, namely that of the 
gauge invariant observables.  In a reduced phase space approach where one 
considers only representations of the observables from the outset, 
the choice of representation constructed from some $g_0$ enters 
only at that stage which is comparable to the problem 
of choice of representation in QFT in CST with the difference that 
now also quantum geometry fields have to be considered. One can now 
in principle again use the Hadamard condition to select suitable $B_0$ 
but with an additional input: The domain of the operators in question 
should be chosen in such a way that the quantum geometry fluctuations 
around $g_0$ are small. Such semiclassical, namely coherent, 
states do exist within Fock representations. This at least ties $B_0$
to the domain of the operators. Apart from that one needs further 
physical input in order to downsize the ambiguity $B_0$, for instance 
one could ask that a maximal number of quadratic forms that enter 
the dynamics in fact turn into operators which 
is similar in nature to the Hadamard condition that ensures that the 
stress energy tensor is a well defined quadratic form.
Finally we remind of Fell's theorem \cite{21} 
which says that 
in principle one can approximate all states from a given one in the sense 
of matrix elements of an arbitrary but finite number of observables 
to arbitrary precision. Unfortunately the proof of that theorem is 
not constructive and therefore of little practical use.

We note that this background dependent but 
non-perturbative approach is much closer to the formalism of QFT
in CST, which is believed to be framework to choose in the semiclassical limit 
of quantum gravity (vanishing geometry fluctuations), 
than when dealing with the background independent 
representation employed in LQG. For suitable matter the physical 
Hamiltonian of the reduced phase space formulation is spatially 
diffeomorphism invariant and one can promote the uniquness theorem 
\cite{15} to the physical Hilbert space level. However, that physical 
Hilbert space is again non-separable and the representation is not regular 
with respect to either the metric or conjugate momentum, the physical 
Hamiltonian operator thus again suffers from quantisation ambiguities. 
In the Fock representation the physical Hilbert space is separable, 
the quantisation
ambiguities are removed to a large extent but the physical 
Hamiltonian is now only a quadratic form, similar as in interacting 
QFT, and no longer an operator on the Fock space. 
Hence there are complementary advantages and disadvantages of Fock versus 
LQG representations. What makes the Fock representation attractive 
is that it is computationally much simpler and equips us with a huge aresenal 
or techniques. 

Thus we see that the viability of Fock representations 
in quantum gravity depends on whether one is able to formulate either 
of i. the quantum constraints, ii. the master constraint, iii. the 
physical Hamiltonian as a densely defined quadratic form on the Fock 
space. Now for the spatial diffeomorphism constraints (SDC) this 
is easy to check they are quadratic polynomials in the fields. For the
Hamiltonian constraint (HC) this is easy to check if one uses it in polynomial 
form which is possible by multiplying it by a sufficiently large 
power of the square root of the spatial metric upon which it becomes 
a polynomial as well. One can show that the SDC commutator can be defined 
in terms of quadratic forms and closes without anomaly
\cite{29}. For the HC 
this is less clear because the SDC and HC do not form a Lie algebra.
If we replace the HC by a master constraint which is invariant under spatial 
diffeomorphisms then we do obtain a Lie algebra and the corresponding 
commutator algebra of quadratic forms has a chance to close. However,
in order to be spatially diffeomorphism invariant, the master constraint
must not be a polynomial! Rather it is a polynomial multiplied 
by inverse powers of the determinant of the spatial metric. Nevertheless
it turns out that one can construct a densely defined quadratic form
from the master constraint using a new technique which we present,
to the best knowledge of the author, for the first time in this 
work. Finally, the physical Hamiltonian is obtained by solving 
the constraints for a canonical momentum conjugate to a scalar 
field which serves as a physical clcok and thus is naturally a density 
of weight one and thus also not a polynomial. Therefore the same 
remarks as for the master constraint apply with additional complications 
originating from the fact that solving for the momenta involves choosing 
a square root. We present a partial solution to this problem which 
again defines a quadartic form.\\
\\
The architecture of this work is as follows:\\
\\

In section \ref{s2} we 
introduce the classical framework and 
define a class of Fock representations for quantum 
gravity both in the constraint quantisation and 
reduced phase space approach. We focus mostly on the geometry sector.

In section \ref{s3} we define the SDC and HC (in polynomial form)
as quadratic forms on the 
Fock space. To define these, one makes use of mode structures
on the one particle Hilbert space, mode cut-offs 
and limiting patterns for the cut-off removal. 
One exploits the freedom in the choice of such structures to tame 
the commutator of quadratic forms and take limits. We show that the 
hypersurface deformation algebra closes in this quadratic form sense 
without anomaly for Euclidean signature vacuum GR written in real 
valued self-dual 
connection variable in four spacetime dimensions. The only reason
to restrict to Euclidian signature and four dimensions is that the 
real valued self-dual 
variables require that restriction but have the advantage of being 
polynomials of degree four only. The same technique can be applied in 
other dimensions and for Lorentzian signature but the computations 
are much harder as we encounter polynomials of degree ten or higher
and have not been 
completed yet.

In section \ref{s4} we present the new technique to define any 
real power of the determinant of the spatial metric {\it at a point} 
as a densely defined quadratic form. The form domain must take into 
account that the quantum metric be non-degenerate if negative powers 
are involved. The same technique can be used to define 
geometric observables such as lengths, areas and volumes  
and the master constraint as quadratic forms in Fock representations.
As for the physical Hamiltonian, which involves a square root, 
one can show that classically it has upper and lower bounds that do 
not involve a square root. One can then quantise these upper and lower 
bounds as quadratic forms on the Fock space which may serve as an 
approximation of the square root expression. We also mention some 
ideas of how to turn the quadratic forms into actual operators without 
using perturbative counter term methods.
 
In section \ref{s5} we summarise and conclude.

\section{Classical Hamiltonian formulation}
\label{s2}

Our description will be brief, more details can be found in 
\cite{9,11}.

\subsection{Constraint formulation}
\label{s2.1}

A foliation of the spacetime $D+1$ manifold $M\cong \mathbb{R}\times \sigma$
with a fixed $D$ manifold $\sigma$ into is encoded by lapse $N^0:=N$ 
and shift functions $N^a,\;a=1..D$ which define the timelike unit 
normal $n^\mu=N^{-1}[\delta^\mu_0-N^a \delta^\mu_n],\;\mu=0,..,D$ of 
the $t=const.$ hypersurfaces in the coordinates 
$x^0=t\in \mathbb{R}$ and $x^a$ on $\sigma$. One has the 
relation $g_{00}=-N^2+q_{ab} N^a N^b, g_{0a}=q_{ab}\;N^b,\; g_{ab}=q_{ab}$ 
between 
$N^\mu$ and the components of the spacetime metric tensor $g_{\mu\nu}$.
Let 
$\omega_\mu^A,\; A=1,..,d$ be a 
connection in a principal fibre bundle (we will shortly 
specialise it to the frame rotation bundle)
with $d-$dimensional gauge group.
The canonical Hamiltonian density is 
a linear combination of constraints $C_\mu,\; \mu=0,..,D$ and $G_A,\; 
A=1..d$ 
\be \label{2.1}
H_{{\rm can}}=v^\mu\; P_\mu+v^A P_A+N^\mu\; C_\mu-\omega_0^A\; G_A
\ee
called secondary Hamiltonian constraint (HC) $C_0$,  
spatial diffeomorphism constraint (SDC) $C_a$ and Gauss constraint $G_A$
respectively and of the primary constraints $P_\mu,P^A$ of momenta 
conjugate to $N^\mu,\omega_0^A$ which are multiplied by the velocities 
$v^\mu=\dot{N}^\mu,\; v^A=\dot{\omega}^A$ that one cannot solve for 
in the Legendre tranform of the Lagrangian. All constraints are first 
class.

The secondary constraints contain 
contributions from both geometry and matter. As nature contains fermions 
one is forced to consider $D+1$-Bein fields $e_\mu^I,\; I=0,..,D$ which 
constitute the spacetime metric $g_{\mu\nu}=\eta_{IJ} e^I_\mu e^J_\nu$ 
where $\eta$ is the Minkowski metric. This implies in particular Gauss 
constraints corresponding to the $D+1$ frame Lorentz group SO(1,D). 
One must remove  
the $D$ boost generators among those Gauss constraints 
by imposing the time gauge 
$e^\mu_0=n^\mu$ leaving $D(D+1)/2-D=D(D-1)/2$ constraints
for the frame rotation group SO(D) as the boost generators generate 
further secondary constraints on $e_\mu^I$ with respect to which they 
are second class.
We will therefore understand that only those rotation 
generators are included in (\ref{2.1}).

Then the gravitational sector is entirely described by $N^\mu$ and 
the D-Bein $e_a^j,\; j=1,..,D$ corresponding to $D+1+D^2=(D+1)^2-D$ fields,
together with the momenta $P_\mu,P^a_j$ conjugate to them
and we have $q_{ab}=\delta_{jk} e^j_a e^k_b$. Since the 
constraints $C_\mu,G_A$ do not depend on $N^\mu,\omega_0^A$ the primary 
constraints in fact have vanishing Poisson brackets with the secondary 
constraints. It follows that $N^\mu,\omega_0^A$ are pure gauge and 
we can reduce the phase space by fixing $N^\mu=f^\mu,-\omega_0^A=f^A$ 
to be arbitrary Lagrange multiplier functions. Then the stability of that 
gauge enforces $v^\mu=v^A=0$. Thus (\ref{2.1}) simplifies to a linear 
combination of secondary constraints
\be \label{2.2}
H_{{\rm can}}=f^\mu\; C_\mu+f^A\; G_A
\ee
The geometry contributions to the constraints reads explicitly
\ba \label{2.3}     
C_0 &=& [\det(q)]^{-1/2}\;[q_{ac}\;q_{bd}-[D-1]^{-1}q_{ab} q_{cd}]\;
p^{ab}\; p^{bc}-[\det(q)]^{1/2}\;(R[q]+\Lambda) 
\nonumber\\
C_a &=& P^b_j\;[e_b^j]_{,a}-[P^b_j e_a^j]_{,b}
\nonumber\\
G_{jk} &=& 2\; P^b_{[j} \delta_{k]l};e^l_b
\ea
where $\Lambda$ is a cosmological constant and
\be \label{2.4}
p^{ab}:=\frac{1}{2}\; P^{(a}_j e^{b)}_k \delta^{jk},\;e^a_j \; 
e_a^k:=\delta_j^k
\ee
$R[q]$ is the Ricci scalar of $q$. As $G_{jk}$ generates 
frame rotations, the quantities $q_{ab}, P^{ab}$ are rotation invariant.
Thus $G_{jk}$ has vanishing Poisson brackets with $C_0$ but not with 
$C_a$. It is therefore convenient to decompose (indices $j,k,l..$ are moved 
with $\delta_{jk}, \delta^{jk}$)
\be \label{2.5}
P^a_j=G_{jk}\; e^{ak}+2\;p^{ab} e_{bj}
\ee
and to rewrite $C_a$ in manifestly Gauss invariant form 
\be \label{2.6}
\hat{C}_a=C_a-[G_{jk} e^{bk} e^j_{b,a}-(G_{jk} e^{bk} e^j_a)_{,b}],\;
=P^{bc}\; q_{bc,a}-2(P^{bc}\; q_{ca})_{,b}
\ee
The non-vanishing Poisson brackets are (we set Newton's constant equal 
to unity) 
\be \label{2.7a}
\{P^a_j(x),e_b^k(y)\}=\delta(x,y)\;\delta^a_b\;\delta_j^k
\ee
In $D=3$ it is possible to reformulate the theory in terms of SU(2) 
connections \cite{23} by realising that 
\be \label{2.8a}
E^a_j:=[\det(q)]^{1/2}\; e^a_j,\;
A_a^j:=\Gamma_a^j-[\det(q)]^{-1/2}\;q_{ab} P^{bj}
\ee
defines a canonical transformation $(P,e)\mapsto (E,A)$ where 
\be \label{2.9a}
e^j_{a,b}-\Gamma^c_{ba} \; e_c^j+\epsilon_{jlk} \Gamma[e]^k_b\; e^l_a:=0
\ee
defines the spin connection $\Gamma_a^j$ of $e$ from the Levi-Civita 
connection $\Gamma^c_{ab}$ of $q$. In higher dimensions one can also 
arrive at a connection formulation \cite{24} but one has to work
with the hybrid connection of SO(D+1) and impose additional simplicity 
constraints. The non-vanishing Poisson brackets are then
\be \label{2.7}
\{E^a_j(x),A_b^k(y)\}=\delta(x,y)\;\delta^a_b\;\delta_j^k
\ee
In these variables the constraints are given in terms of the 
curvature $F=2(dA+A\wedge A)$ of $A$ by
\ba \label{2.8}
\hat{C}_0 &=& |\det(E)|^{-1/2} F_{ab}^j\;\epsilon_{jkl} E^a_k E^b_l
-2\;|\det(E)|^{1/2}\;R[q] 
\nonumber\\
\hat{C}_a &=& F_{ab}^j E^b_j
\nonumber\\
G_j &=& E^a_{j,a}+\epsilon_{jkl}\; A_a^k \; E^{bl}
\ea
where $C_0-\hat{C}_0$ differ by terms proportional to $G_j$ and 
$q_{ab}=|\det(E)| E^j_a E_{bj},\; E^a_j E_a^k=\delta_j^k$.

The HC $C_0$ or $\hat{C}_0$ as produced by the Legendre transform of 
the Lagrangian carries 
a natural density weight of unity however then it is evidently not polynomial
in the fields $P^a_j,e_a^j$ or $P^{ab}, q_{ab}$ or $E^a_j, A_a^j$. For 
purposes of Fock quantisation it is convenient to consider its polynomial 
version. This is obtained as follows: In the $P^a_j, e_a^j$ formulation
we have $[\det(q)]^{1/2}=|\det(e)|$. The Ricci scalar is given explicitly
by 
\ba \label{2.9}
R&=& q^{ac}\; \delta^c_d \; R_{abc}\;^d
=q^{ac}\; \delta^c_d \; 
(-2\partial_{[a} \Gamma^d_{b]c}+
2\Gamma^e_{c[a}\; \Gamma^d_{b]e}),\;
\nonumber\\
\Gamma^d_{ab}&=&q^{dc}\; \Gamma_{cab},\;
2\Gamma_{cab}=2q_{c(a,b)}-q_{ab,c}
\nonumber\\
q^{ab}&=&\det(q)^{-1}\; [(D-1)!]^{-1}\; 
\epsilon^{a c_1..c_{D-1}}
\epsilon^{b d_1..d_{D-1}} 
q_{c_1 d_1}..q_{c_{D-1} d_{D-1}}
\nonumber\\
\Gamma^d_{ab,e}&=&
q^{dc}_{,e} \Gamma_{cab}+q^{dc} \Gamma_{cab,e}
=-q^{dg}\; q_{gh,e}\; q^{hc} \Gamma_{cab}+q^{dc} \Gamma_{cab,e}
\ea
It follows that $\det(q)^2\; \Gamma^e_{ca}\Gamma^d_{be}$ 
is a monomial of degree $2D$
in $2D-2$ factors of $q$ and two factors of $\partial q$ 
while $\det(q)^2\; \Gamma^d_{bc,a}$
is a monomial of degree $2D$ consisting of two terms, 
the first term consisting of $2D-2$ factors of $q$ and two
factors of $\partial q$ while the second term consists of (2D-1) factors of 
$q$ and one factor of $\partial^2 q$. Due to the contraction 
of the Ricci tensor with the inverse 
metric we require one more factor of $\det(q)$. It follows that 
\be \label{2.10}
\hat{C}_0:=\det(q)^{5/2}\; C_0
\ee 
is a polynomial in the $(q,p)$ polarisation where the ``kinetic term'' 
containing two facors of $p^{ab}$ is of degree $2(D+1)$ in $q_{ab}$ 
while the Ricci term has degree $2D+D-1=3D-1$ in $q_{ab}$. For $D=3$ 
its top degree is ten. 

For the $(e,P)$ polarisation we make use of the fact that the 
Ricci tensor can be written as $R_{abjk}\; e^{aj} e^{bk}$ and 
$R_{jk}=2\;[d\Gamma+\Gamma\wedge \Gamma]_{jk}$ where 
$\det(e)\;\Gamma_{jk}$ is a homogeneous 
polynomial of degree $D$ with one factor
of $\partial e$ and $D-1$ factors of $e$. Since 
$|\det(e)|=\sqrt{\det(q)}$, by the same 
reasoning it is now sufficient to set 
\be \label{2.11}
\hat{C}_0:=\det(q)^{3/2}\; C_0
\ee 
whose kinetic term containing two factors of $P^a_j$ has $2(D+1)$ factors 
of $e$ while the Ricci term has $2D+2(D-1)=4D-2$ factors of $e$. Its 
top degree is also ten in $D=3$ (in fact it is homogeneous).

Finally in the $(A,E)$ polarisation which is only available in $D=3$
we have $R=R_{abjk} E^{aj} E^{bk} \det(q)^{-1}$ and again 
note that $\det(E)\Gamma_{jk}$ is a homogeneous polynomial 
of degree 3 with one factor of $\partial E$ and two factors of $E$. 
Since $\det(E)=\det(q)$ we must use (\ref{2.10}) again 
which contains eight facors of $E$ in the $F$ term and 
is a polynomial of degree $6+2=8$ in the curvature term. Its top degree 
is also ten.

To summarise in all three polariasations we can generate a polynomial 
form of the Hamiltonian constraint with minimal top degree ten in $D=3$ but 
it has density weight six in the $(q,p)$ and $(A,E)$ polarisation while 
it has density weight four in the $(e,P)$ polarisation.

\subsection{Master constraint formulation}
\label{s2.2}

Consider the smeared constraints 
\be \label{2.15}
K(f,u,m):=\int_\sigma\; d^Dx\; [f\; C_0+u^a\;C_a+m^A\; G_A]
\ee
They obey the following Poisson bracket relations
\be \label{2.16}
\{K(f,u,m),K(g,v,n)\}=-K(u[g]-v[f],[u,v]+q^{-1}(f\;dg-g\;df),[m,n]+u[n]-v[m])
\ee
where $[u,v]$ is the commutator of vector fields, $[m,n]$ the 
commutator of matrices, $df$ the exterior derivative
of a function and $u[f]=df[u]$ the contraction of a vector field with the 
exterior derivative. The r.h.s. of (\ref{2.16}) is a linear combination 
of smeared constraints, but the smearing functions are no longer constant 
functions on the phase space because of the appearance of $q^{ab}$. 
Thus the algebra of smeared constraints is no Lie algebra, also not 
in the polynomial version. 

This is often argued to cause 
trouble in the constraint quantisation using Fock spaces
due to the following reason: As we will see, the polynomial constraints 
can be quantised as densely defined quadratic forms in a Fock
representation and this is 
sufficient to state the quantum constraint equations and to solve them.
This requires to normal order the polynomial constraints. 
Furthermore, we show that there exists a mathematically 
meaningful procedure to define the commutator 
of quadratic forms even though the product of quadratic forms is ill-defined.
The result of the commutator calculation, if it exists, is then again a 
quadratic form which is necessarily normal ordered. However, the 
normal ordering of an expression of the form, say in $D=3$,  
\be \label{2.16a}
\{\hat{C}_0(f),\hat{C}_0(g)\}=
-
\int\;d^Dx\; \det(q)^3 q^{ab}(f g_{,b}-f_{,b} g)\; C_a  
\ee
is not be given by 
\be \label{2.17}
-\int\;d^Dx\; :C_a:\;\;: \det(q)^3 q^{ab}:\;\;(f g_{,b}-f_{,b} g)\;
\ee
Thus, if $l$ is a linear functional on the form domain $\cal D$
of the constraints 
such that $l[C_a(u^a)\psi)=0$
%:=\sum_b l(b)\;<b,C_a(u^a)\psi>=0$ 
for all vector fields $u$ and all 
$\psi\in {\cal D}$ 
then $l([C(u),C(v)]\psi)\not=0$ even if 
$[C(u),C(v)]$ is given by the normal ordering of (\ref{2.16a}), i.e. 
even if there is no anomaly in the constraint algebra. This,
in contrast to the situation in which the constraints are 
promoted to operators with dense invariant domain $\cal D$ 
is, however, not a contradiction
because $C(v)\psi$ is no longer in $\cal D$ so that $l(C(u)C(v)\psi)\not=0$
(in fact this expression diverges). This will be explained in more 
detail in the next section. 

We can avoid these complications using the master constraint method.
It comes in a minimal and a maximal version. In the minimal version
it is only applied to $C_0$
\be \label{2.18}
M:=\int\; d^Dx\; C_0\;\det(q)^{-1/2} C_0
\ee
To see that the single $M$ classically encodes the same information as 
the infinite number of $C_0(f)$ we note that $M=0$ iff $C_0(f)=0$ for all
$f$ thanks to the positivity of the integrand of (\ref{2.18}). Next
for a function $F$ on the phase space we have $\{F,\{F,M\}\}_{M=0}=0$ iff
$\{C(f),F\}_{M=0}=0$ for all $f$ i.e. $M$ defines the same weak Dirac
observables.  Since $C_0$ is a density of weight one it follows that 
$\{G(m),M\}=\{C_a(u^a),M\}=0$ and in the quantum theory one is reduced to 
check the algebra of the $G(M),C_a(u^a)$ for anomaly freeness
and that they commute with $M$, the commutator $[M,M]=0$ is now trivial. 
However,
the price to pay is that $M$ in contrast to $G(m),C_a(u^a)$ is not 
polynomial and thus it is unclear how to define $M$ as a quadratic form 
on the Fock space. Surprisingly, this is indeed possible as we will see
in later sections. 

Now instead of defining a solution by 
$l(\hat{C}_0(f)\psi=0$ for all $f,\psi$ one instead imposes 
$l(M\psi)=0$ for all $\psi$ where $M$ is normal ordered. Normal 
ordering keeps symmetry of a quadratic form but destroys positivity.
Note however that in an ordering such that $C_0$ is a symmetric 
operator we would formally 
have 
\be \label{2.19}
<\psi,M\psi>=\int\; d^Dx\; ||([\det(q)]^{-1/4} C_0)(x)\psi||^2
\ee
which is formally solved by those $\psi$ with $\hat{C}_0(f)\psi=0$ for all
$f$. Thus ``modulo ordering corrections'' the condition $l(M\psi)=0$ appears 
to be reasonable.\\
\\
In the maximal version we combine {\it all constraints} into a single 
master constraint. As then we do not need to check any algebra, we can use 
the constraints in polynomial form. That is one would set e.g.
\be \label{2.20}
M=\int\; d^Dx\;[\det(q)\;q^{ab}\; C_a\; C_b+\det(q)^5\;C_0^2+
\delta^{AB}\; G_A G_B]
\ee
and normal order.

\subsection{Reduced phase space formulation}
\label{s2.3}

In this case one removes the constraints classically using 
gauge fixing. This has the usual global reservations such as Gribov copies.
It is in general true that constraining and quantisation do not commute 
thus we consider this as an alternative route which is only classically 
locally equivalent to the constraint method. 

As for the non-gravitational Gauss constraints 
we can use one of the standard gauges such as axial gauge or Coulomb gauge.
For the gravitational Gauss constraint we can use the upper triangular 
gauge on the D-bein \cite{25}. That is, one can solve the relation
\be \label{2.21}
q_{ab}=e_a^j\delta_{jk} e^k_b
\ee
uniquely and algebraically for $e^j_a=u^j_a,\; 
u^j_a=0,\;a>j; [u^a_j]_{j=a}>0$. Here $j=a$ refers to a choice of coordinate 
system in which we label coordinates by $a=1,..,D$ as well as the internal 
indices $j=1,..,D$ and one can state sufficient conditions under which this 
is compatible with the atlas of $\sigma$. 
We consider $q_{ab}=u^j_a\;u_{bj}$ as a function of of $u_a^j$ and solve 
the Gauss constraints $G_{jk}=0$ algebraically for $P^a_j,\; a>j$. The
reduced phase is then coordinatised by $Q^a_j:=P^a_j, u_a^j,\; a\le j$. 
Equivalently
we invert $q_{ab}=q_{ab}[u]$ for $u_a^j=u_a^j[q]$ and pull 
back the symplectic structure
\be \label{2.22}
\Theta=\int\;d^Dx\; \sum_{a\le j} Q^a_j\;[\delta u_a^j]
=\int\; P^{ab} \;[\delta q_{ab}] 
\ee
where $P^{ab}=\frac{1}{2}\; P^{(a}_j e^{b)j}$ at $e_a^j=u_a^j,\; G_{jk}=0$
(if the fermionic current is not vanishing one must add the fermionic 
contribution to $\Theta$ to see this \cite{25}).

Thus one may choose to work either with $q_{ab}, P^{ab}$ or $u_a^j, Q^a_j$.
Next one reduces the constraints $C_\mu$. This may be done in various ways 
depending on the matter content. To have a concrete example in mind we 
consider the case of $D+1$ minimally coupled, massless scalar fields 
$\phi^I,\; I=0,..,D$ with 
conjugate momentum $\pi_I$ which contribute
\be \label{2.23}
\pi_I \phi^I_{,a},\;\;\;\; 
\frac{1}{2}\;[\delta^{IJ} \pi_I\pi_J[\det(q)]^{-1/2}
+\delta_{IJ} \phi^I_{,a}\phi^J_{,b}[\det(q)]^{1/2}]
\ee
to $C_a,C_0$ respectively. We fix the coordinate gauge freedom 
\be \label{2.24}
G^I=\phi^I-L^{-1}\delta^I_\mu x^\mu=0,\;\phi^I_\ast=L^{-1} x^I
\ee
where we have introduced a length scale $L$ for dimensional reasons ($\phi^I$
is dimension free while coordinates have length dimension). We then 
solve the constraints for $\pi_I$ which yields
\be \label{2.23a}
\pi^\ast_a=-L\;C'_a,\;
\pi_0^\ast=\sqrt{-2[\det(q)]^{1/2}\tilde{C}},\;
\tilde{C}=C'+\frac{1}{2}(L^2\;\delta^{ab} C'_a C'_b/\sqrt{\det(q)}
\;+L^{-2} q^{ab} \delta_{ab}\sqrt{\det(q)})  
\ee
where $C',C'_a$ are HC and SDC without the scalar contribution and 
we have chosen the positive root for $\pi_0^\ast$. The 
stability of the gauge condition under gauge transformations
\be \label{2.24a}
G^I_{,0}+\{C_\mu(f^\mu),G^I\}=-L^{-1} \delta^I_0+L^{-1}\;f^a \delta^I_a
+f^0\;[\det(q)]^{-1/2} \pi_I=0
\ee
gives at $\phi=\phi_\ast, \pi=\pi^\ast$ with $Q=\sqrt{\det(q)}$
\be \label{2.25}
f^0_\ast=L^{-1}\; Q \; [\pi_0^\ast]^{-1},\;
f^a_\ast=-L\;f^0_\ast\; \pi_a\ast\; Q^{-1}=-\pi_a^\ast/\pi_0^\ast
\ee
The physical Hamiltonian is defined for a phase space function $F$ independent
of  the gauge degrees of freedom $\phi,\pi$ by 
\be \label{2.26}
\{H,F\}=\{C_\mu(f^\mu),F\}_{\phi-\phi_\ast=\pi-\pi^\ast=f-f_\ast=0}
\ee
It drives the dynamics of the ``true degrees of 
freedom'' i.e. the non-scalar fields and working out 
(\ref{2.26}) we find that it is explicitly given by
\be \label{2.27}
H=-L^{-1}\;\int\; d^Dx\; \pi_0^\ast
\ee
~\\
In quantising the reduced phase space we only quantise the non-gauge fixed 
matter fields and the gravitational degrees of freedom $q_{ab}, P^{ab}$ 
and we are no longer interested in the constraints but seek to implement 
the Hamiltonian $H$ at least as a quadratic form. To do this, the square 
root in (\ref{2.24}) is problematic because we can quantise only integrals 
over $\tilde{C}^n\;Q^m$ 
with $Q:=\det(q)^{1/2}$ and $n\in \mathbb{N}_0,\;m\in \mathbb{Z}_0$
as quadratic forms as we will see. 
To deal with this problem we offer 
the following preliminary solution: We note that the expression 
$-\tilde{C}$ 
under the square root is constrained to be positive in the classical 
theory. Consider any function $\Lambda_0>0$ and write 
\be \label{2.28}
-L^{-1}\sqrt{-2Q\tilde{C}}
=-\sqrt{2}\; L^{-1}Q\sqrt{\Lambda_0}\;
[1-(1+\frac{\tilde{C}}{\Lambda_0 Q})]^{1/2}
\ee
We have 
$y:=1+\frac{\tilde{C}}{Q\Lambda_0}\le 1$ and for such $y$ 
we have the estimate 
\be \label{2.29}
1-\frac{y}{2}-\frac{y^2}{2}\le \sqrt{1-y}\le 1-\frac{y}{2}
\ee
To see this note that for $y\le 1$ the claimed upper bound is positive 
so the claimed upper inequality is equivalent to $1-y\le (1-y/2)^2$ which is 
identically satisfied even for all $y$. 
The claimed lower bound can be written 
as $\frac{1}{2}(1-y)(2+y)$ which for $y\le 1$ is positive for $y\ge -2$.
Thus for $-2\le y\le 1$ the claimed lower inequality is equivalent to 
$(1-y)^2\;(2+y)^2\le 4\;(1-y)$ which holds for $y=1$ and for $y<1$ 
is eqivalent to $(1-y)\;(2+y)^2-4=-y^2(3+y)\le 0$ which indeed holds for 
$-2\le y\le 1$. Thus the claimed lower bound holds 
when the lower bound is 
positive and thus also for all other $y<1$ as $\sqrt{1-y}$ is not negative.
Of course, sharper bounds can be obtained by taking the 
higher order Taylor expansion of $\sqrt{1-y}$ in terms of 
$y$ into account.  

Accordingly, we can approximate the square root by an interpolating value 
between upper and lower bound
\be \label{2.30}
\overline{\sqrt{1-y}}=1-y/2-k\; y^2/2,\;\;k\in [0,1]
\ee
Accordingly we could approximate the Hamiltonian by
\be \label{2.31}
H_k=\sqrt{2\Lambda_0}\; L^{-1}\;\int\;d^Dx\;Q\;\sqrt{2\Lambda_0}\; L^{-1}\;
[y/2+k\;y^2/2-1]
\ee
and for any $0\le k\le $ 
the approximant $H_k$ is in between the classically allowed extrema.
Alternatively we could truncate the Taylor expansion 
of $-\sqrt{1-y}$ at some order $N$.

While upper and lower polynomial bounds are good to know an exact expression 
that avoids the square root would be benefitial. Such an expression is 
given by
\be \label{2.32}
\sqrt{z}=z\;\int_\mathbb{R}\; \frac{ds}{\sqrt{2\pi}}\;e^{-\frac{s^2}{2}\;z}
\ee
which however is no longer polynomial. It would require to be able 
to quantise $Q\; z\; e^{-s^2 z/2},\;z=-\frac{\tilde{C}}{Q}$ as a quadratic 
form which is not known to be possible. 

In lack of a better option we treat $k$ or the truncation 
order $N$ as a phenomenological 
parameter that serves to approximate the exact expression. We therefore 
have to make sure that the $k$ dependent or higher than order $N$ terms 
indeed are subleading 
by providing a choice for $\Lambda_0$. To that extent, 
we will see in section \ref{s4}
that we may actually 
define $\frac{\tilde{C}}{Q}$ as a quadratic form in Fock representations 
with dense form domain 
defined by the span of the excitations of a coherent state $\Omega_Z$ where 
$Z$ denotes a point in the classical phase space such that the 
corresponding D-metric is positive definite. We thus pick
\be \label{2.36a}
\Lambda_0:=-<\Omega_Z, \frac{\tilde{C}}{Q}\Omega_Z>
\ee
and quantise (\ref{2.31}) by normal ordering. 
If $\Omega_Z$
is a coherent state peaked on a spatially very homogeneous universe 
then $\Lambda_0$ is expected to be almost a constant on sufficiently 
large scales. It is expected that the fluctutation factors 
$\Delta:=1+\tilde{C}/(\Lambda_0 Q)=1-
\frac{\tilde{C}}{Q}\;[<\frac{\tilde{C}}{Q}>]^{-1}$
contribute subdominantly in $:Q\; \Delta^n:$ where $:.:$ denotes 
normal ordering. See the end of section \ref{s4.3} for more details.\\
\\
In any case, the point of this discussion is that the objects 
\be \label{2.36b}
\tilde{C}^n\; Q^{n-1}
\ee
can be quantised non-perturbatively as quadratic forms on the Fock space.
The expansion in powers of $n$ is not to be considered as a perturbative 
expansion but rather is similar to the relativistic expansion
$p_0=-\sqrt{m^2+\vec{p}^2}=m+\frac{\vec{p}^2}{2m}+..$ of the mass shell
condition $-p_0^2+\vec{p}^2+m^2=0$ with the difference 
that the square root is not expanded at $m$ but rather $0$. This is because 
we have not solved for the analog of $p_0$ which would be the negative 
gravitational mode but rather for the analog of 
$p_D=-\sqrt{-[-p_0^2+\vec{p}_\perp^2+m^2]}$ where 
$\vec{p}=(\vec{p}_\perp,p_D)$ and the analog of $m$ is the cosmological 
constant term. This cannot 
be avoided: It would have been more natural to solve the constraint 
$C=0$ for the negative conformal gravitational mode. But this is not possible 
algebraically as the gravitational momenta, including the 
conformal mode, appear in $C_a$ with derivatives.

\subsection{Fock representations}
\label{s2.4}

Consider an arbitary background metric $g$ of Euclidian signature 
on $\sigma$. From it we can construct the Laplacian 
$\Delta=g^{ab}\nabla_a\nabla_b$ where $\nabla_a$ is the torsion free 
covariant derivative compatible with $g$. The scalar one particle 
Hilbert space is defined by $L_2\equiv L_2(\sigma,\sqrt{\det(g)}\;d^Dx)$
and $\Delta$ is either essentially self-adjoint or at least has self-adjoint 
extensions of which we pick one.  
Let $\kappa$ be an invertible (in the sense of the spectral theorem) 
function of $\Delta$ and consider canonically conjugate tensor fields
$q^{a_1..a_A}_{b_1..b_B},\; p_{a_1..a_A}^{b_1..b_B}$ of dual type 
$(A,B,w=0)$ and $(B,A,w=1)$ respectively where $w$ is the density weight.
We consider the annihilators 
\be \label{2.37}
A_{c_1...c_{A+B}}:=2^{-1/2}\;
[\kappa\;
\;g_{c_1 a_1}..g_{c_A a_A}\; q^{a_1..a_A}_{c_{A+1}..c_{A+B}}
-i \kappa^{-1}\;\omega^{-1}
g_{c_{A+1} b_1}..g_{c_{A+B} b_B}\;p_{c_1..c_A}^{b_1..b_B}
\ee
where 
\be \label{2.38a}
\omega:=[\det(g)]^{1/2}
\ee
is invoked such that $A$ is a background tensor of type $(0,A+B,0)$.
 
The annihilators are smeared with density weight zero test functions 
$f^{a_1..a_{A+B}}$ of type $(A+B,0,0)$ relative to the volume form 
$\omega \; d^Dx$
\be \label{2.38}
A(f):=<f,A>_{L_2}:=\int\; d^Dx\;\omega\; [f^{c_1..c_{A+B}}]^\ast\;
A_{a_1..a_{A+B}}=2^{-1/2}\;[<\kappa f,q>-i<\omega^{-1}\;\kappa^{-1} f,p>]
\ee
where indices are pulled with $g$. Then using
\be \label{2.39}
[p_{a_1..a_A}^{b_1..b_B}(x),\;q^{a'_1..a'_A}_{b'_1..b'_B}(x')]
=i 
\delta_{a_1}^{a'_1}..\delta_{a_A}^{a'_A}
\delta^{b_1}_{b'_1}..\delta^{b_B}_{b'_B}\;
\delta(x,x')
\ee
where $\int\; d^Dx \delta(x,y) f(y):=f(x)$ we find 
\be \label{2.39a}
[A(f),A(g)^\ast]=\frac{1}{2}\;[
<\kappa f,\kappa^{-1} g>
+<\kappa^{-1} f,\kappa g>]
=<f,g>
\ee
where self-adjointness of $\kappa=\kappa(\Delta)$ was used: We have
with a multi-index $\mu=(a_1..a_C)$ and recalling that $\nabla_a$ acts 
on densities of weight $w$ by a correction 
$-w\;\Gamma^b_{ab}=-[\ln(\omega^w)]_{,a}$ so that $\nabla_a \tilde{v}^a=
\partial_a \tilde{v}^a$ for vector densities of weight one
\ba \label{2.40}
<f,\Delta g> &=&
\int\; d^Dx\; \omega\; f_\mu^\ast \nabla_a(\nabla^a g^\mu)
=\int\; d^Dx\; f_mu^\ast \nabla_a(\nabla^a \omega g^\mu)
\nonumber\\
&=& \int\; d^Dx\; 
[\nabla_a(f_\mu^\ast\nabla^a \omega g^\mu)
-(\nabla_a f_\mu^\ast)\; (\nabla^a \omega g^\mu)]
\nonumber\\
&=& -\int\; d^Dx\; 
(\nabla_a \omega f_\mu^\ast)\; (\nabla^a g^\mu)
\nonumber\\
&=& -\int\; d^Dx\; 
[\nabla_a( \nabla^a \omega f_\mu^\ast)\; g^\mu)
-(\nabla_a (\nabla_a \omega f_\mu^\ast))\;g^\mu]
\nonumber\\
&=& <\Delta f,g>
\ea
assuming that the integrations by parts result in no boundary terms.

Fermions are naturally densities of weight 1/2 and for each Weyl fermion 
component $\rho_I$ we have $\rho_I^\ast$ as its anti-conjugate 
\be \label{2.41} 
[\rho_I(x),\rho_{I'}^\ast(x')]_+=\delta_{II'}\;\delta(x,x')
\ee
We thus define the annihilator to be $B_I=\omega^{-1/2}\rho_I$.

The joint Fock vacuum is then defined by $A_\mu(x)\Omega=B_I(x)\Omega=0$.

\section{Polynomial SDC and HC as quadratic forms on Fock space} 
\label{s3}

We first define the constraints as quadratic forms on the Fock space.
Then we make some general remarks on the solution of constraints 
which are merely quadratic forms but not operators as well as 
the possibility of a commutator algebra of quadratic forms. We then 
give details for the case at hand when a Fock structure is available. 

\subsection{Polynomial constraint quadratic forms on Fock space}
\label{s3.1} 

Apart from algebraic work, there is not much to do. We have 
already shown that in section \ref{s2} that 
both the SDC and HC can be written in polynomial form by multiplying the 
HC by a sufficiently high power of $\sqrt{\det(q)}$. Now for any polynonomial 
\be \label{3.1}
P(x)=P(q(x),p(x),[\partial q](x),[\partial p](x),
\rho(x),\rho^\ast(x),[\partial \rho](x),[\partial \rho^\ast](x))
\ee
of the fields and its first spatial derivatives where $(q,p)$ 
and $(\rho,\rho^\ast)$ denotes the 
collection of all bosonic and fermionic conjugate pairs respectively, we 
decompose $q,p$ into annihilation and creation operators by inverting  
(\ref{2.37})
\be \label{3.2}
q=2^{-1/2}\;
[\otimes^A\; g^{-1}]\;\otimes\; [\otimes^B\; 1_D]\;\kappa^{-1}\cdot[A+A^\ast],
\;\;
p=i\; 2^{-1/2}\;
\omega[\otimes^A\; 1_D]\;\otimes\; 
[\otimes^B\; g]\;\kappa\;\cdot\;[A-A^\ast]
\ee
while $\rho=B, \rho^\ast=B^\ast$ are already annihilation and creation 
operators. Then in every monomial we order all occurrences of $A^\ast$ 
left of all occurrences of $A$ and 
all occurrences of $B^\ast$ 
left of all occurrences of $B$. For the bosonic and 
fermionic sector we multiply by $+1$ and $-1$
for any interchange of the  $A,A^\ast$ and of the 
$B,B^\ast$ respectively in this process.

The reordered polynomial is denoted by $:P(x):$ and given a definition 
of annihilation and creation operators the ordering of $:P(x):$ is unique.
For any $x\in \sigma$, $:P(x):$ is a densely defined quadratic form on the 
Fock space $\cal H$ where the dense form domain $\cal D$ is the finite 
linear span of the Fock states
\be \label{3.3}
A(f_1)^\ast\;..\;A(f_M)^\ast\;B(g_1)^\ast\;..\;B(g_N)^\ast\; \Omega
\ee
where $M,N\in \mathbb{N}_0$ are the bosonic and fermionic 
particle numbers respectively and $f_k,g_l$ are test functions (typically 
from the Schwartz function class ${\cal S}(\sigma)$ of smooth functions 
of rapid decrease at $\partial\sigma$ if $\sigma$ is not bounded or 
has a boundary). To see this we decompose $P$ into monomials $Q$ with 
say $K_B,K_F$ creators and $L_B,L_F$ annihilators (the subscripts 
denote bosonic and fermionic respectively) and consider 
\be \label{3.4}
<\psi,\;:Q(x):\; \psi'>_{{\cal H}}
\ee
where $\psi,\psi'$ are of the form (\ref{3.3}) with particle numbers 
$M,N$ and $M',N'$ respectively and let all annihilators and 
creators respectively act on $\psi',\psi$ respectively. The result is
zero if either of $K_B > M, K_F>N, L_B>M', L_F>N'$ or 
of the form 
\be \label{3.5}
<\psi_1,\psi_1'>\;h(x)
\ee
where the particle number of $\psi_1,\psi_1'$ respectively is given by 
$(M-K_B, N-K_F),\;(M'-L_B, N'-L_F)$ 
and $h$ is a test function again (a product of $K_B+K_F+L_B+L_F$ test 
functions involved in $\psi,\psi'$ or their derivatives convoluted 
with $g,g^{-1},\kappa,\kappa^{-1}$).  

Now the smeared, polynomial, normal ordered constraints are given 
by 
\be \label{3.6}
:C(f):=\int\; d^Dx\; f^\mu(x)\;:C_\mu(x):
\ee
and thus are symmetric (i.e. 
$<\psi,\;:C(f):\;\psi'>^\ast=<\psi',\;:C(f):\;\psi>$)
quadratic forms on the Fock space with dense form 
domain $\cal D$ because for any test function $h$ the integral
$\int\; d^Dx\; f^\mu\; h$ exists.

\subsection{Solution of the constraints as quadratic forms, constraint
algebra and solutions}
\label{s3.2}

The ususal approach to constraint quantisation is as follows \cite{26}:     
Suppose that some constraints $Z_I,\;I\in {\cal I}$ some index set 
are defined {\it as operators} on a dense invariant domain
${\cal D}$ of the Hilbert space ${\cal H}$, that is, $Z_I {\cal D}\subset 
{\cal D}$. A linear functional $l$ on $\cal D$ is called a solution 
iff $l[Z_I \psi]=0$ for all $I\in {\cal I},\;\psi\in {\cal D}$. Note 
that $Z_I\psi\in {\cal D}$ so that $l$ defined on ${\cal D}$ is defined
on $Z_I \psi$. 

Now let $b_n$ be an orthonormal basis of $\cal H$ 
constructed from elements of $\cal D$ (in the separable case we 
obtain this by using the Gram-Schmidt process, in the non-separable case 
this involves the axiom of choice. Fock spaces are separable). 
Then for an operator $Z_I$ densely 
defined on invariant $\cal D$ we have 
$Z_I \psi\in {\cal D}\subset {\cal H}$, thus 
\be \label{3.7a}
Z_I \psi=\sum_n b_n\; <b_n,Z_I\psi>,\;\; 
\ee
and
\be \label{3.7b}
||Z_I \psi||^2=\sum_n\;|<b_n,Z_I \psi>|^2
\ee
converges. A solution $l$ is now determined by the coefficients  
\be \label{3.8}
l_n:=l[b_n]
\ee
and can be written as the linear functional
\be \label{3.9}
l:=\sum_n\; l_n\; <b_n,.>_{{\cal H}}
\ee
It follows that $l$ is a solution iff
\be \label{3.10}
\sum_n\; l_n\; <b_n,Z_I\psi>=0
\ee
for all $I\in {\cal I}, \psi\in {\cal D}$ which is now an infinite system
of linear equations on the coefficients $l_n$ with no growth condition 
on $l_n$ as $n$ varies.\\
~\\
Now suppose that the constraints $Z_I$ are not operators on 
$\cal D$ but merely quadratic forms. Then $Z_I \psi$ is ill-defined
as a vector in the Hilbert space, we 
can only compute matrix elements $<\psi,Z_I\psi'>$. In particular, while 
we may try to use (\ref{3.7a}) as a formal definition, formula (\ref{3.7b})
now diverges. However, the condition (\ref{3.10}) is still meaningful!
That is, to state (\ref{3.10}) it is sufficient that $Z_I$ is a quadratic 
form on $\cal D$ and that we know the coefficients (\ref{3.8}) of $l$ with 
respect to an ONB $b_n\in {\cal D}$. We thus require $l$ to be a linear 
functional not only on $\cal D$ but also on the formal range of $Z_I$ 
on $\cal D$ defined by (\ref{3.7a}). We may consider $Z_I\psi$ as an 
anti-linear functional on $\cal D$ by 
\be \label{3.11}
\psi'\mapsto\; (Z\psi)[\psi']:=<\psi',Z_I\psi>
\ee
which is well-defined and then $l$ acts on $Z_I\psi$ by 
\be \label{3.12}
l[Z_I\psi]:=\sum_n\; l[b_n]\;(Z\psi)[b_n]
\ee
This can converge even absolutely if the $l[b_n]$ decay sufficiently much faster
than the $(Z_I\psi)[b_n]$ as $n$ varies. 

Assuming that the solution space consists not only of the trivial solution 
$l\equiv 0$, pick $l_0\not=0$. Then we may try to define an inner product,
at least in the separable case, on the space $L$ of solutions by 
\be \label{3.13}
<l,l'>_L:=\lim_M\;
\frac{\sum_{n\le M}\;l[b_n]^\ast\;l'[b_n]}{\sum_{n\le M}\;
|l_0[b_n]|^2} 
\ee
where we assumed w.l.g. that $b_n$ is labelled by $n\in \mathbb{N}_0$. 

Finally suppose that the constraints originate from functions on the 
phase space which we denote by the same symbol $Z_I$ and that form a 
Poisson Lie algebra
\be \label{3.14}
\{Z_I,Z_J\}=f_{IJ}\;^K\;\; Z_K
\ee
with structure constants $f_{IJ}\;^K$.
If the $Z_I$ are operators with dense invariant domain $\cal D$ then
for $\psi\in {\cal D}$ we have $Z_I\psi\in {\cal D}$ so that 
$[Z_I,Z_J]\psi\in {\cal D}$ and thus we define the quantum constraint 
algebra to be anomaly free iff
\be \label{3.15}
[Z_I,Z_J]\psi=i\;f_{IJ}\;^K\;\; Z_K\psi
\ee
for all $I,J\in {\cal I},\;\psi\in {\cal D}$. If the $Z_I$ are merely 
quadratic forms then $Z_I\psi$ is meaningless and while $<\psi,Z_I\psi'>$ 
is well defined, again $<\psi,Z_I Z_J\psi'>$ is meaningless. However,
let 
\be \label{3.16}
P_M:=\sum_{n\le M}\; b_n\;<b_n.>_{{\cal H}}
\ee
be the projection onto ${\cal H}_M$, the 
completion of the linear span of the $b_n,\;n\le M$, and suppose 
that for each $M<\infty$ the numbers
\be \label{3.17}
<\psi,Z_I\; P_M Z_J\psi'>
:=\sum_{n\le M}\; <\psi,Z_I\;b_n>\;
<b_n,Z_J\;\psi'>
\ee
exist. Then due to cancellation effects the weak limit 
\be \label{3.18}
<\psi,[Z_I,Z_J]\psi'>:=\lim_M\;  
[<\psi,Z_I\; P_M Z_J\psi'> - <\psi,Z_J\; P_M Z_I\psi'>]
\ee
may exist. In this case we have a {\it commutator algebra of quadratic 
forms} and we define it to be free of anomalies iff 
\be \label{3.19}
<\psi,[Z_I,Z_J]\psi'>=i\;f_{IJ}\;^K\;\; <\psi,Z_K\psi'>
\ee
for all $I,J\in {\cal I},\;\psi,\psi'\in {\cal D}$
i.e. in the sense of quadratic forms which generalises (\ref{3.15}).

In case that the classical constraints $Z_I$ close as in (\ref{3.14})
but with non-constant structure functions $f_{IJ}\;^K$ the condition 
(\ref{3.15}) (for the case that $Z_I$ are operators with dense invariant
domain $\cal D$) 
is ill-defined as it stands because now $f_{IJ}\;^K$ is 
itself operator valued so that ordering issues arise. For instance, 
if $Z_I,Z_J,f_{IJ}\;^K$ are symmetric operators then an ordering 
consistent with the adjointness relations would be 
\be \label{3.20}
[Z_I,Z_J]\psi=\frac{i}{2}\;(
f_{IJ}\;^K\;\; Z_K+Z_K\;f_{IJ}\;^K)\;\psi
\ee
Let now $l$ be a solution of the constraints in the sense 
of $l[Z_I\psi]=0$ for all $I\in {\cal I},\psi\in {\cal D}$ then as $Z_I\psi\in 
{\cal D}$ we would have the identity
\be \label{3.21}
l[[Z_I,Z_J]\psi]=0=\frac{i}{2}\;
l[(f_{IJ}\;^K\;\; Z_K+Z_K\;f_{IJ}\;^K)\;\psi]
\ee
for all $I,J,\psi$. If moreover $f_{IJ}\;^K\;\psi\in {\cal D}$ for all
$I,J,K,\psi$ then 
\be \label{3.21a}
0=\frac{i}{2}\;l[[f_{IJ}\;^K,Z_K]\;\psi]
\ee
for all $I,J,\psi$. Thus $l$ vanishes not only on all the $Z_I$ but also
on all of the $Z_{IJ}=i[f_{IJ}\;^K,Z_K]$ which correspond to the classical 
Poisson brackets $\{f_{IJ}\;^K,Z_K\}$ which are not required to vanish 
in the classical theory. This can be iterated and can result in a space 
of solutions which is too small \cite{27}. 
Accordingly, in case of structure functions 
it is not sufficient to have e.g. the property (\ref{3.20}), one must make sure 
that the space of solutions is large enough. This may may be achieved 
for non-symmetric constraint operators \cite{16} or non-symmetric 
structure functions or for the case 
that above assumption $f_{IJ}\;^K\psi\in {\cal D}$ fails.

Now consider the case 
that $Z_I$ are only quadratic forms. Then the above obstruction 
argument fails because 
we do not have the property $Z_I\psi\in {\cal D}$ so that we cannot conclude 
from $l[Z_I\psi]=0$ that also $l[Z_I\;Z_J\psi]=0$. Hence in this case 
there is no immediate conclusion about the size of the space of solutions 
to be drawn even if the $Z_I$ are symmetric quadratic forms. This 
is because of (\ref{3.18}) which says that while  
$<\psi,Z_I Z_J\psi'>, \;<\psi,Z_J Z_I\psi'>$ 
are ill defined (divergent) the difference between
$<\psi,Z_I P_M Z_J\psi'>$ and 
$<\psi,Z_J P_M Z_I\psi'>$ converges weakly as $M\to\infty$. Thus while we have
\be \label{3.21b} 
l[Z_I\psi]=\lim_{M'} \sum_{n'\le M'}\; l_{n'}\; <b_{n'},Z_I\psi>=0
\ee
we have 
\be \label{3.22} 
l[[Z_I,Z_J]\psi] 
:=\lim_{M'} \sum_{n'\le M'}\; l_{n'} \; <b_{n'},[Z_I,Z_J]\psi>
;=\lim_{M'} \sum_{n'\le M'}\; l_{n'} \; \lim_M\; 
\{<b_{n'}, Z_I P_M Z_J\psi> - <b_n, Z_J P_M Z_I\psi>\}
\ee
If $P_M Z_J \psi\in {\cal H}_M$ and 
$P_M Z_I \psi\in {\cal H}_M$ would converge strongly 
to some $\psi_J,\psi_I\in {\cal D}$ then
\be \label{3.23} 
l[[Z_I,Z_J]\psi] 
=\lim_{M'}\sum_{n'\le M'}\; l_{n'} \; 
\{<b_{n'}, Z_I \psi_J> - <b_{n'}, Z_J \psi_I>\}=l[Z_I\psi_J]-l[Z_J\psi_I]=0
\ee
However this is not the case, otherwise $Z_I,Z_J$ would be operators. 
Alternatively note that 
\be \label{3.24}
l[[Z_I,Z_J] \psi]=\lim_{M'}\;\lim_M \sum_{n'\in \le M'} \; l_{n'} 
\sum_{n\le M}\;
\{<b_{n'},Z_I\;b_n>\;<b_n, Z_J \psi> 
-<b_{n'},Z_J\;b_n>\;<b_n, Z_I \psi>\}
\ee
If we could interchange the limits $M,M'\to \infty$ then (\ref{3.24}) 
would indeed vanish. But such interchange of limits is typically 
only allowed when the double sum is absolutely 
convergent (Fubini type of theorem) which is not the case as 
the existence of $l[Z_I\psi]=0$ is crucially sensitive on the phases of 
the complex numbers $l_n$.

Note that if $[Z_I,Z_J]$ exists as a quadratic form and equals 
some ordered version of the classical expression 
$f_{IJ}\;^K\; Z_K$, that quantum version will 
be normal ordered in the case of the Fock quantisation. However 
then it will certainly not be ordered as $Z_K f_{IJ}\;^K$ and thus 
$l[[Z_I,Z_J]\psi]\not=0$ even if $f_{IJ}\;^K$ has $\cal D$ as invariant 
domain. 

This shows that the case of non-trivial structure functions 
comes with extra complications and motivates to recast the algebra of 
constraints into a form in which the structure functions are replaced 
by structure constants. Locally this is always possible because 
one can in principle solve the constraints for some of the momenta
and in this form the Poisson algebra of constraints is Abelian. Another 
possibility is the minimal or maximal 
master constraint approach discusssed in the 
previous section.
In Abelianised version and the minimal master constraint version
the constraints are no longer polynomial which therefore calls for new 
tools developed in section \ref{s4}. The maximal master constraint version
allows for a polynomial formulation and therefore  
fits into the present general discussion. The classical 
maximal master constraint is constructed from the $Z_I$ as 
\be \label{3.25}
M=\frac{1}{2}\sum_{I,J}\; Z_I\; Q^{IJ}\; Z_J
\ee
where $(Q^{IJ})_{I,J\in {\cal I}}$ is a positive definite matrix which 
may be a non-constant function on the classical phase space. One may 
certainly choose $Q^{IJ}, Z_I$ as polynomials so that $M$ is a polynomial
again and thus may be formulated as quadratic form in Fock representations 
by normal ordering.

One may criticise that the master constraint (\ref{3.25}) is ``blind'' for 
possible anomalies in the algebra or that it does not allow to check 
whether the classical algebra is represented in the quantum theory. However,
consider in the classical theory phase space functions $G^I$, so called gauge 
fixing conditions, such 
that $\Delta_I\;^J:=
\{Z_I,G^J\}$ is non-degenrate close to the constraint surface 
$Z_I=0,\; I\in {\cal I}$. Then 
\be \label{3.26}
\{M,G^J\}\; [\Delta^{-1}]_J\;^M\;
Q_{MI}=: \{M,G^J\}\;\hat{Q}_{JI}=Z_I +O(Z^2),\;\;Q_{IM}\;Q^{MJ}
=\delta_I\;^J
\ee
so that $Z_I$ maybe recovered from $M$ using $G^I$. Then 
\be \label{3.27}
\{Z_I,Z_J\}
=f_{IJ}\;^K\; Z_K
=f_{IJ}\;^M\; \hat{Q}_{KM} \{M,G^K\}+O(Z^2)
=\hat{Q}_{KI}\;\hat{Q}_{LJ}\;
\{\{M,G^K\},\{M,G^L\}\}+O(Z^2)
\ee
i.e.
\be \label{3.28}
\{\{M,G^I\},\{M,G^J\}\}=\hat{f}^{IJ}\;_K\;\{M,G^K\} +O(Z^2),\;
\hat{f}^{IJ}\;_K=f_{MN}\;^L \; Q^{IM}\;Q^{JN}\;Q_{KL},\;\;
\ee
so that the $Z^I:=\{M,G^I\}$ also close with structure functions, 
albeit in general with non polynomial ones. Thus 
classically $M$ is truly equivalent to the set $\{Z_I\}_{I\in {\cal I}}$
and in the quantum theory one could in principle
try to consider the constraint algebra 
defined by the $[M,G^I]$. 

Of course the whole point of the master constraint 
is to circumvent that constraint algebra with structure functions and one    
expects that unitarily inequivalent representations of the CCR and CAR 
supporting $M$ as a quadratic form differ in the ``gap number'' 
\be \label{3.29}
M_0:={\rm inf}_{0\not=\psi \in {\cal D}}\;\frac{<\psi,M\psi>}{||\psi||}
\ee
For instance, in Fock representations the classical positivity 
is destroyed in the normal ordering process while $M$ is a symmetric 
quadratic form. One therefore expects that in Fock representations
the expectation value of $M$ w.r.t. $\psi\in {\cal D}$ takes full range in 
$\mathbb{R}$, i.e. $M_0=-\infty$ so that solutions $l[M\psi]=0$ for all 
$\psi\in {\cal D}$ should exist. In that case a systematic procedure 
consists in solving the infinite system 
\be \label{3.30}
K_m:=\sum_n\; l_n M_{n,m}=0;\;\;M_{n,m}=<b_n,M b_m>
\ee
as follows: Let $(b_n)_{n=0}^\infty$ be a labelling of the ONB by 
$n\in \mathbb{N}_0$ and set $M^{(0)}_{n,m}:=M_{n,m},\;m,n\ge 0$.
Find all $m$ such that $M^{(0)}_{n,m}=0\; \forall\;
n\in \mathbb{N}_0$ and drop those $m$ from the set $\mathbb{N}_0$ resulting
in $R_0$. For the smallest $m_0\in R_0$   
find the smallest $n=n_0$ in $L_0:=\mathbb{N}_0$ 
such that $M_{n,m_0}\not=0$. 
Solve $K_{m_0}=$ for $l_{n_0}=\hat{l}_{n_0}$ and substitute that 
solution into the 
$K_m,\;m_0\not=m\in R_0$. Collecting coefficients one can write these 
$K_m$ as a sum over $l_n,\;n\in L_1=L_0-\{n_0\}$ with coefficients 
$M^{(1)}_{n,m}$. Drop all $m$ from $R_0-\{m_0\}$ for which 
$M^{(1)}_{n,m}=0\; \forall\; n\in L_1$ resulting in $R_1$. Now iterate,
thereby producing sets $L_0\supset L_1 \supset .. \supset L_N\; \supset ..,\; 
N\in \mathbb{N}_0,\; L_{N+1}=L_N-\{n_N\}$. 
The space of solutions $\{l_n\}_{n\in \mathbb{N}_0}$ will be parametrised 
by the free coefficients $l_n,\;n\not\in \{n_N\}_{N=0}^\infty$.

\subsection{Implementation in Fock representations}
\label{s3.3}

Parts of the exposition of the present section are also considered 
in \cite{29}. We include here a shorter version for reasons 
of self-containedness and refer to \cite{29} for more details.
Note that the index set $\cal I$ considered in this section 
labels a basis $e_I$ in the one particle Hilbert space and therefore 
can be used to smear constraints $Z$ as $Z_I:=<e_I,Z>$. 
In this case the meaning of $\cal I$ is the same as in the previous 
section. In general however the meaning of $\cal I$ in this and
the previous section is logically independent.\\
\\
The concrete implementation of the programme described in the previous 
section in Fock representations is sytematised by picking an  
orthonormal 
basis of real valued functions 
$e_I$ with {\it mode label} $I\in {\cal I}$ of the 
scalar one particle Hilbert space $L_2=L_2(\sigma,\; d^Dx\;\omega)$ and 
a background $D-$Bein $h^a_j,\;j=1..D$ with inverse co-$D-$Bein $h_a^j$ such 
that $g_{ab} h^a_j h^b_k=\delta_{jk}$. Then $e_{I;j_1..j_A}^{k_1..k_B}$ 
is an ONB of tensors of type $(A,B,0)$ with components 
\be \label{3.31}
[e_{I;j_1..j_A}^{k_1..k_B}]^{a_1..a_A}_{b_1..b_B}(x)
=e_I(x)\;\prod_{r=1}^A h^{a_r}_{j_r}(x)\;\;
\;\prod_{s=1}^B h_{b_s}^{k_s}(x)\;\;
\ee
with respect to the one particle inner product 
\be \label{3.32}
<t,\hat{t}>:=\int_\sigma\; d^Dx\;\omega\; 
\overline{t^{a_1..a_A}_{b_1..b_B}}\;
\prod_{r=1}^A\; g_{a_r a'_r}\;
\prod_{s=1}^B\; g^{b_s b'_s}\; \hat{t}^{a'_1..a'_A}_{b'_1..b'_B}
\ee
We define the background scalars 
of density weight zero (indices $j,k$ are moved with the Kronecker $\delta$)
\be \label{3.33a}
q_{j_1..j_{A+B}}:=
h_{a_1 j_1}..h_{a_A j_A}\; 
h^{b_1}_{j_{A+1}}..h^{b_B}_{j_{A+B}}\; q^{a_1..a_A}_{b_1..b_B},\;\;
p_{j_1..j_{A+B}}:=\omega^{-1}
h^{a_1}_{ j_1}..h^{a_A}_{j_A}\; 
h_{b_1 j_{A+1}}..h_{b_B j_{A+B}}\; p_{a_1..a_A}^{b_1..b_B},\;\;
\ee
which have Poisson brackets
\be \label{3.34}
\{p_{j_1..j_{A+B}}(x),q_{k_1..k_{A+B}}(y)\}=\delta(x,y)
\prod_{r=1}^{A+B}\; \delta_{j_r k_r}
\ee 
We define the bosonic smeared annihilators 
\be \label{3.33}
A_{I;j_1..j_{A+B}}:=<e_I,A_{j_1..j_{A+B}}>:=2^{-1/2}
<e_I,[\kappa\cdot q_{j_1..j_{A+B}}-i\kappa^{-1}\cdot p_{j_1..j_{A+B}}>_{L_2}
\ee
where $\kappa=\kappa(\Delta),\; 
\Delta=\omega^{-1}\partial_a \omega g^{ab} \partial_b$
and the fermionic smeared annihilators 
\be \label{3.34a}
B_{I\alpha}=<e_I,B_\alpha>:=<e_I,\omega^{-1/2} \rho_\alpha>_{L_2}
\ee
which annihilate the Fock vacuum $\Omega$.

The Fock orthonormal basis is then defined by 
\be \label{3.35}
b_n:=\prod_{I,\mu}\;\frac{[C_{I\mu}^\ast]^{n_{I\mu}}}{\sqrt{n_{I\mu}!}}\;
\Omega,\; ||n||:=\sum_{I,\mu} n_{I,\mu}<\infty
\ee
where $\mu$ stands for one of the multi-indices $j_1..j_{A+B}$ or 
$\alpha$ respectively in which $C_{I,\mu}$ means $A_{I,\mu}$ or 
$B_{I,\mu}$ respectively and the {\it occupation numbers} 
$n_{I\mu}\in \mathbb{N}_0$.\\
\\
Let $\cal D$ be the finite linear span of the $b_n$.
A {\it truncation structure} on $\cal I$ is a sequence 
of finite mode sets $S_M\subset {\cal I},\;
M\in \mathbb{N}_0$ with $S_0=\{I_0\}$, $S_M\subset S_{M+1},\; 
\cup_M S_M={\cal I}$ and $I_0$ is a single element in ${\cal I}$. 
It defines ${\cal D}_M\subset {\cal D}$ 
to be the span of the $b_n$ with $n_{I,\mu}=0$ for $I\not \in S_M$.
Note that the truncation 
structure does not impose any restriction on the occupation numbers 
$n_{I,\mu}$ for $I\in S_M$ except that by definition of a Fock state 
$||n||:=\sum_{I,\mu} n_{I\mu} <\infty$. For instance, 
for $\sigma=\mathbb{R}^D$ and $g_{ab}=\delta_{ab}$ 
the $e_I$ could be the Hermite 
functions $I=(I_1,..,I_D)\in \mathbb{N}_0^D$ and $I\in S_M$ iff
$|I|=I_1+..+I_D\le M$ and $I_0:=0$. Similarly, 
for $\sigma=T^D$ the $e_I$ could be the momentum eigenfunctions
with $I=(I_1,..,I_D)\in \mathbb{Z}^D$ and $I\in S_M$ iff
$|I|=|I_1|+..+|I_D|\le M$ and again $I_0:=0$.
   
The truncation structure can be used in order to define the truncations 
$:P_M(x):$ of a normal ordered polynomial $:P(x):$ as follows: We write 
$:P(x):$ in terms of the $A_{j_1..j_{A+B}}, B_\alpha$ and their adjoints 
and then use the completeness relation 
\be \label{3.36}
\delta(x,y)
=\sum_I\; e_I(x)\;[e_I(y)]^\ast\; \omega(y),\;\;
\int\; d^Dy \; \delta(x,y)\; f(y)=f(x)
\ee
to expand the annihililation and creation operators in terms of the 
smeared versions (\ref{3.33}), ((\ref{3.34}). For each monomial of order 
$N$ we obtain $N$ sums over the $I_1,..,I_N\in {\cal I}$. The truncation 
consists in the restriction of those sums to $I_1,..,I_N\in S_M$. 
In this way the infinite series is truncated to a normal ordered 
polynomial in the smeared creation and annihilation operators, for 
each monomial of order $N$ we obtain at most $|S_M|^N$ terms where 
$|S_M|$ is the finite number of elements of $S_M$. 

Note that since $S_M\to {\cal I},\; M\to \infty$, for each $\psi,\psi'$
we find $M\in \mathbb{N}_0$ such that $\psi,\psi'\in {\cal D}_M$. Therefore 
\be \label{3.37}
<\psi,\;:P(x):\;\psi'>=<\psi,\;:P_M(x):\;\psi'>
\ee
To see this consider $:P(x):$ expanded into smeared annihilation and 
creation operators without truncation of the infinite series involved. 
Then let all creation operators act to the left 
on $<\psi,.>$ and all annihilation operators to the right on $\psi'$. 
Whenever the mode label $I$ of such an operator is not in $S_M$ we can 
commute it trough the creation operators acting on $\Omega$ that are involved 
in $\psi,\psi'$ whereby we pick at most a sign and then $\Omega$ is 
annihilated. It follows 
\be \label{3.38}
:P(x):=w-\lim_{M\to \infty}\; :P_M(x):
\ee
i.e. every normal ordered polynomial is the weak operator topology limit 
of its truncations which is the appropriate operator topology to choose 
for quadratic forms.   

Now quadratic forms cannot be multiplied, i.e. their product is no 
longer a quadratic form which is the source of most 
complications in QFT and
necessiates regularisation and renormalisation procedures when 
considering the S-matrix (an infinite series of smeared interaction 
terms which itself is not an operator, only a quadratic form). However
the truncations $:P_M(x):$ are actually operators with dense invariant 
domain $\cal D$ because they are polynomials in smeared creation and 
annihilation operators with coefficients that take values in the 
Schwarz functions. Thus taking products of the truncations 
results in well defines operators densely defined on $\cal D$ whose 
limit in general does not exist, not even as a quadratic form. However,
due to cancellations involved, the limit of commutators of truncations 
may exist as a quadratic form. 

Thus, given smeared, normal ordered polynomial constraints and their 
truncations
\be \label{3.39}
C(u):=\int\;d^Dx\; u^\mu(x)\;:C_\mu(x):,\;\;
C_M(u):=\int\;d^Dx\; u^\mu(x)\;:C_{\mu,M}(x):,\;\;
\ee
we define, if it exists
\be \label{3.40}
[C(u),C(v)]:=w-\lim_{M\to \infty}\;[C_M(u),C_M(v)]
\ee
In \cite{29} we have carried out the analysis explicitly for the case 
of the spatial diffeomorphism constraints and found that 
\be \label{3.41}
[C[u],C[v]]=iC(-[u,v])
\ee
where $[u,v]$ is the commutator of vector fields. That is, {\it for any
choice of Fock representation of the above type}, no matter that it is 
heavily background dependent, the algebra of SDC quadratic forms not 
only exists but it is free of anomalies! This shows that background
dependence of elements of the quantum theory (here the representation of 
the CCR and CAR) does not spoil the background 
independence of the algebra of spatial diffeomorphisms.\\
\\
We now want to analyse whether this result can be extended to the 
full algebra of constraints for Lorentzian GR, written in polynomial form. We 
consider first the vacuum sector. We have seen that we can work with 
$(P^{ab},q_{ab})$ and in $D=3$ the HC becomes a polynmial of degree ten in 
density weight six while in the $(A,E)$ formulation 
it becomes a polynomial of degree twenty with density weight fourteen.
The higher the 
polynomial degree, the more terms will come out not normal ordered in 
the commutator calculation and the more potentially divergent and/or 
anomalous terms will appear upon reinstalling normal order. Thus, 
to get a first taste of the complications involved we consider Euclidian 
vacuum GR 
in 3+1 dimensions where we have available the connection formluation 
mentioned in section \ref{s2} involving GC, SDC and HC  
constraints $G_j,C_a,C_0$ which are only second or fourth order polynomials 
respectively in density weight two. 
We thus expect the complications involved to be less severe 
than for Lorentzian signature. The classical Poisson algebra
of the polynomial constraints 
\ba \label{3.42}
&& C(r,u,f) = \int\; d^Dx\; \{r^j\; G_j+u^a C_a+f\;C_0\},\;
G_j=E^a_{j,a}+\epsilon_{jkl} A_a^k E^a_m \delta^{lm},\;
\\
&& C_a=E^b_j A_{b,a}^j-(E^b_j A_a^j)_{,b},\;
C_0=F_{ab}^j\epsilon_{jkl} \delta^{km} \delta^{ln} E^a_m E^b_n,\;
F_{ab}^j=2A_{[b,a]}^j+\delta^{jm}\epsilon_{mkl} A_a^k A_b^l
\nonumber
\ea
is given by
\ba \label{3.43}
&& \{C(r,u,f),C(s,v,g)\} =-C(r\times s+u[s]-v[r]
-4\;Q(\omega,A),[u,v]-4\;Q(\omega,.),u[g]-v[f])
\nonumber\\
&&
Q^{ab}=E^a_j E^a_k \delta^{jk},\;\omega_a=f\;g_{,a}-g\;f_{,a},\; 
u[s]^j=u^a\; s^j_{,a},\;u[f]=u^a\; f_{,a}-u^a_{,a} \; f
\ea
Note that $s,f$ transform as scalars of density weight zero and minus one 
respectively under spatial diffeomorphisms.
We carry out the above programme in the simplest topology, background metric 
and Fock 
representation
$\sigma=\mathbb{R}^3,g_{ab}=\delta_{ab},\kappa=1$ 
and use the Hermite 
basis $e_I,\; I\in \mathbb{N}_0^3,\; |I|=I_1+I_2+I_3$ 
for concreteness to define the truncation structure. A possible triad is 
$h_a^j:=\delta_a^j$ and $\omega=[\det(g)]^{1/2}=1$. The annihilation 
operator is then given by 
\be \label{3.44} 
A_{mk}:=2^{-1/2}\{h^a_m\; A_a^l\;\delta_{lk}-i \delta_{ml} h_a^l E^a_k\}
\ee
and we have the non-vanishing commutators 
\be \label{3.45}
[A_{mj}(x),A^\ast_{nk}(y)]=\delta_{mn}\;\delta_{jk}\;\delta(x,y)
\;\;\Leftrightarrow\;\;
[E^a_j(x), A_b^k(y)]=i\delta^a_b\;\delta_j^k\;\delta(x,y)
\ee
and the smeared operators are simply 
\be \label{3.46}
A_{Imj}=<e_I,A_{jk}>_{L_2},\;\;
[A_{Imj},A_{Jnk}^\ast]=\delta_{IJ}\;\delta_{mn}\delta_{jk}
\ee
The normal ordered constraints can be expanded as (we define the derivatives 
$f_{,j}:=\delta^a_j f_{,a}$ and Einstein summation over twice 
repeated lower case indices independent of index position being understood) 
\ba \label{3.47} 
G(r) &=& 
-\frac{i}{2^{1/2}}\;\sum_I\; r_{Imi}\; \; [A_{Imi}-A_{Imi}^\ast]
+\frac{i}{2}\epsilon_{ijk}\sum_{I,J} r_{IJk}\; 
[A_{Imi}^\ast\;(A_{Jmj}-A_{Jmj}^\ast)+(A_{Jmj}-A_{Jmj}^\ast)A_{Imi}]
;\;\;
\nonumber\\
r_{Imi} &:=& <e_I,r^i_{,m}> ,\;r_{IJk}:= <r^k,e_I e_J>
\nonumber\\
C(u) &=& \frac{i}{2}\; \sum_{I,J}\; u_{IJmn}\;\{ 
[A_{Imj}-A_{Imj}^\ast]\;A_{Jnj}
+A_{Jnj}^\ast\;[A_{Imj}-A_{Imj}^\ast]\};\;\;
\nonumber\\
u_{IJmn} &:=&
\delta_{mn}\; <e_I,u_p e_{J,p}>+<e_I e_J, u_{n,m}>
\nonumber\\
-2\;C_0(f) &=&\;2^{1/2} 
\sum_{I,J,K}\; \;f_{IJnK}\;\epsilon_{ijk}
:[A_{Imi}-A_{Imi}^\ast]\;[A_{Jnj}-A_{Jnj}^\ast]]\; 
[A_{Kmk}+A_{Kmk}^\ast]:
\nonumber\\
&& +\frac{1}{2}  
\sum_{I,J,K,L}\; \;f_{IJnK}\;\epsilon_{rij}\;\epsilon_{rkl}
:[A_{Imi}-A_{Imi}^\ast]\;[A_{Jnj}-[A_{Jnj}^\ast]]\; 
[A_{Knk}+A_{Knk}^\ast]\;
[A_{Lmk}+A_{Kmk}^\ast]:
\nonumber\\
&& f_{IJnK}=<(e_I e_J f)_{,n}, e_K>,\;\;f_{IJKL}=<e_I e_J e_K e_L,f>
\ea
These inner products are well defined as the Hermite basis consists of 
Schwartz functions. 

The computations are organised by using the compound index notation 
\be \label{3.48}
\alpha=(Imi),\;
\beta=(Jnj),\;
\gamma=(Kpk),\;
\delta=(Lql)
\ee
and 
\be \label{3.49}
r_\alpha=r_{Imi},\; r_{\alpha\beta}=r_{IJk} \epsilon_{kij} \delta_{mn},\;
u_{\alpha\beta}=u_{IJmn} \delta_{ij},\;
f_{\alpha\beta\gamma}=f_{IJnK} \delta_{mp} \epsilon_{ijk},\;\;
f_{\alpha\beta\gamma\delta}=f_{IJKL}\; \epsilon_{rij}\;\epsilon_{rkl}\;
\delta_{mp}\;\delta_{nq}
\ee
which allows the compact notation (summation over repeated compound 
indices understood). 
\ba \label{3.50}
G(r) &=& G_1(r)+G_2(r)=
-i 2^{-1/2}\; r_\alpha\; F_\alpha+\frac{i}{2}\; r_{\alpha\beta}\;
:G_\alpha\; F_\beta:
\nonumber\\
C(u) &=& C_2(u)=\frac{i}{2}\; u_{\alpha\beta}\; :F_\alpha\; G_\beta:
\nonumber\\ 
-2C_0(f)&=& C_{0;3}(f)+C_{0;4}(f)=
2^{1/2}\; 
f_{\alpha\beta\gamma}\; :F_\alpha\; F_\beta\; G_\gamma:
+\frac{1}{2}
f_{\alpha\beta\gamma\delta}\; 
:F_\alpha\; F_\beta\; G_\gamma\; G_\delta:
\nonumber\\
&& F_\alpha:=A_\alpha-A_\alpha^\ast,\;G_\alpha:=A_\alpha+A_\alpha^\ast,\;
\ea
Here $:(.):$ denotes normal ordering of $(.)$ in the usual way, e.g.
\ba \label{3.51}
&& :\;F_\alpha\;F_\beta\;G_\gamma\;:
=A_\gamma^\ast:\;F_\alpha\;F_\beta\;:
+ :\;F_\alpha\;F_\beta\;:\;A_\gamma
\nonumber\\
&& :\;F_\alpha\;F_\beta\;:
= F_\alpha\;A_\beta - A_\beta^\ast\;F_\alpha
\ea
In (\ref{3.50}) the summation range of $\alpha,\beta,\gamma,\delta$ is 
unconstrained, that is, there is no restriction on the range 
of the mode label $I$ in $\alpha=(Imi)$ while of course $m,i\in\{1,2,3\}$
take only a finite range. Therefore the expressions (\ref{3.50}) are formal
only. 

We define the ordering on the compound indices 
\be \label{3.50a}
\alpha=(Imi)\le M \;\; \Leftrightarrow \;\;  I\in S_M
\ee
Let now 
\ba \label{3.52a}
G_{M_1;M_2,M_3}(r) &=& 
G_{1;M_1}(r)+G_{2;M_1,M_3}(r)=
-\frac{i}{2^{1/2}}\; \sum_{\alpha\le M_1}\;
r_\alpha\; F_\alpha+\frac{i}{2}\;\sum_{\alpha\le M_2,\beta\le M_3}\; 
r_{\alpha\beta}\;
:G_\alpha\; F_\beta:
\nonumber\\
C_{M_1,M_2}(u) &=& 
C_{2;M_1,M_2}(u) = 
\frac{i}{2}\; \sum_{\alpha\le M_1,\beta\le M_2}
\; u_{\alpha\beta}\; :F_\alpha\; G_\beta:
\nonumber\\ 
-2C_{0,M_1,..,M_7}(f)&=& 
C_{0;3;M_1,..,M_3}(f)+C_{0;4;M_4,..,M_7}(f)
\nonumber\\
&=&
2^{1/2}\;
\sum_{\alpha\le M_1,\beta\le M_2,\gamma\le M_3}\; 
f_{\alpha\beta\gamma}\; :F_\alpha\; F_\beta\; G_\gamma:
\nonumber\\
&& +\frac{1}{2}\; 
\sum_{\alpha\le M_4,\beta\le M_5,\gamma\le M_6,\delta\le M_7}
f_{\alpha\beta\gamma\delta}\; 
:F_\alpha\; F_\beta\; G_\gamma\; G_\delta:
\ea
be the expansions of $F,G$ in terms of the $e_I$, the sums
truncated as displayed. 
Thus 1. each occurring sum has now finite range and 2. individual normal 
ordered monomials have individual finite ranges. We can even take one step 
further and instead of summing over the same $\alpha\le M$ for 
$A_\alpha,\;A_\alpha^\ast$ occurring in $G_\alpha$ or 
$F_\alpha$ we can sum over different labels. This results in the explicit 
expressions 
\ba \label{3.52}
&& G_{M_1,..,M_8}(r) = 
G^0_{1;M_1}(r)
+G^1_{1;M_2}(r)
+G^0_{2;M_3,M_4}(r)
+G^1_{2;M_5,M_6}(r)
+G^2_{2;M_7,M_8}(r)
\nonumber\\
&& G^0_{1;M_1}(r)=-\frac{i}{2^{1/2}}\; 
\sum_{\alpha\le M_1}\;r_\alpha\; A_\alpha,\;
G^1_{1;M_2}(r)=\frac{i}{2^{1/2}}\; 
\sum_{\alpha\le M_2}\;r_\alpha\; A_\alpha^\ast
\nonumber\\
&& G^0_{2,M_3,M_4}=\frac{i}{2}\;\sum_{\alpha\le M_3,\beta\le M_4}\; 
r_{\alpha\beta}\;
A_\alpha\; A_\beta
\nonumber\\
&& G^1_{2,M_5,M_6}=\frac{i}{2}\;\sum_{\alpha\le M_5,\beta\le M_6}\; 
[r_{\alpha\beta}-r_{\beta\alpha}]\;
A_\alpha^\ast\; A_\beta
\nonumber\\
&& G^2_{2,M_7,M_8}=\frac{i}{2}\;\sum_{\alpha\le M_7,\beta\le M_8}\; 
r_{\alpha\beta}\;
A_\alpha^\ast\; A_\beta^\ast
\nonumber\\
&& C_{M_1..M_6}(u) = 
C^0_{2;M_1,M_2}(u) + C^1_{2;M_3,M_4}(u) + C^2_{2;M_5,M_6}(u)
\nonumber\\
&& 
C^0_{2;M_1,M_2}(u)=\frac{i}{2}\; \sum_{\alpha\le M_1,\beta\le M_2}
\; u_{\alpha\beta}\; A_\alpha\; A_\beta,\;\;\;
C^1_{2;M_3,M_4}(u)=\frac{i}{2}\; \sum_{\alpha\le M_3,\beta\le M_4}
\; [-u_{\alpha\beta}+u_{\beta\alpha}]\; A_\alpha^\ast\; A_\beta\;
\nonumber\\
&&C^2_{2;M_5,M_6}(u)=\frac{i}{2}\; \sum_{\alpha\le M_5,\beta\le M_6}
\; u_{\alpha\beta}\; A_\alpha^\ast\; A_\beta^\ast,
\nonumber\\ 
&& -2C_{0,M_1,..,M_{28}}(f)= 
\sum_{k=0}^3\; C^k_{0;3;M_{3k+1},M_{3k+2},M_{3k+3}}(f)
+\sum_{k=0}^4\; C^k_{0;4;M_{12+4k+1},M_{12+4k+2},M_{12+4k+3},M_{12+4k+4}}(f)
\nonumber\\
&& C^0_{0;3;M_1,M_2,M_3}(f)=
2^{1/2}\;
\sum_{\alpha\le M_1,\beta\le M_2,\gamma\le M_3}\; 
f_{\alpha\beta\gamma}\; A_\alpha\; A_\beta\; A_\gamma, \;\; ...
\nonumber\\
&& C^3_{0;3;M_{10},M_{11},M_{12}}(f)=
2^{1/2}\;
\sum_{\alpha\le M_{10},\beta\le M_{11},\gamma\le M_{12}}\; 
f_{\alpha\beta\gamma}\; A_\alpha^\ast\; A_\beta^\ast\; A_\gamma^\ast
\nonumber\\
&& 
C^0_{0;4;M_{13}..M_{16}}(f)=
\frac{1}{2}\; 
\sum_{\alpha\le M_{13},\beta\le M_{14},\gamma\le M_{15},\delta\le M_{16}}
f_{\alpha\beta\gamma\delta}\; 
A_\alpha\; A_\beta\; A_\gamma\; A_\delta,\;\; ...
\nonumber\\
&& C^0_{0;4;M_{25}..M_{28}}(f)=
\frac{1}{2}\; 
\sum_{\alpha\le M_{25},\beta\le M_{26},\gamma\le M_{27},\delta\le M_{28}}
f_{\alpha\beta\gamma\delta}\; 
A_\alpha^\ast\; A_\beta^\ast\; A_\gamma^\ast\; A_\delta^\ast
\ea
Here the superscript $k$ and subscript $N$ e.g. in 
$G^k_{N;M_1..M_N}(r);\;\; N=1,2;\; k=0,1,..,N$ means that we are looking 
at a normal ordered monomial of order $N$ with $k$ creation and $N-k$ 
annihilation operators. For each such monomial we have numerous 
ways of taking the $M_1,.., M_N\to \infty$. 
\begin{Definition} \label{def3.1} ~\\
Let for $0\le k\le N\in \mathbb{N}$
\be \label{3.53}
E^k_{N;M_1,..,M_N}:=\sum_{\alpha_1\le M_1,..,\alpha_N\le M_N}\;\;
E^k_{\alpha_1..\alpha_N}\; 
A_{\alpha_1}^\ast..A_{\alpha_k}^\ast\;
A_{\alpha_{k+1}}..A_{\alpha_N}
\ee
be an N-th order normal ordered monomial in $k$ creation operators and $N-k$
annihilation operators. \\
i.\\
A limiting pattern $P$ in $1\le l\le N$ indices is an ordered partition of 
$\{1,..,,N\}$ 
into subsets $U_1,..,U_l$ i.e. $\cup_{r=1}^l\; U_r=\{1,..,N\},\; 
U_r\cap U_s=\emptyset \;\forall\; r\not=s$.\\
ii.\\
A limiting process subordinate to a limiting pattern $P$ is the 
coincidence of cut-offs $M^r:=M_{s_1}=..=M_{s_{L_r}}$ for 
$U_r=\{s_1,..,s_{L_r}\}$ followed by a sequential limit: First 
take $M^1\to\infty$ at finite $M^2,..,M^l$, then take  
$M^2\to\infty$ at finite $M^3,..,M^l$, ... and finally take 
$M^l\to \infty$. This precise process is denoted by $P\to \infty$.\\
iii.\\
The extreme cases are the {\it total coincidence} limit $l=1$ 
with $M^1=M_1=..=M_N\to \infty$ and one 
of the $N!$ {\it total individual} 
limits $l=N$ defined by a permutation $\pi$ 
in $N$ elements i.e. $M_{i_{\pi(k)}}\to \infty$ at 
finite $M_{\pi(k+1)},..,M_{\pi(N)}$ in the order $k=1,2,..,N$. \\
iv.\\
The quadratic form limit $E^k_{N;P}$, if it exists, 
subordinate to a limiting pattern
$P$ is the weak operator topology limit
\be \label{3.54}
<\psi,\;E^k_P\;\psi'>:=\lim_{P\to \infty}\; 
<\psi,\;E^k_{N;M_1,..,M_N}\;\psi'>:
\ee
for all $\psi,\psi'\in {\cal D}$.\\
v.\\
A K-th order normal ordered polynomial is a sum of terms of the form 
(\ref{3.53}) with $0\le N\le K,\; 0\le k\le N$ each carrying its own 
$M^{(N)}_1,..,M^{(N)}_N$. A limiting pattern and process 
applied to the polynomial is the sum of the individual monomial limiting 
patterns and monomials.  
\end{Definition}
Note that the resulting $E^k_{N;P}$ is defined independently of 
the choice of $\psi,\psi'$ in (\ref{3.54}), i.e. the limiting process does not 
depend on those, so that the resulting object is indeed a quadratic form, i.e.
a sesqui-linear form with form domain $\cal D$.

Various generalisations of the limiting pattern are conceivable and should 
be kept in mind for potential future applications:\\ 
1.\\ We could refine the pattern
and make it also dependent on all the finite indices $m_k,i_k$ of the 
$\alpha_k=(I_k,m_k,i_k),\; k=1,..,N$.\\ 
2.\\
Instead of a sequential limit 
we could consider coincidence limits at different ``rates'', say 
$M^r=\exp(c_r\; m),\; m\to\infty,\;c_r\in \mathbb{R}_+$ for 
some of the $r$ or a mixture of 
those with sequantial limits etc.\\
3.\\
An expression of the form (\ref{3.53}) could be split into several 
of the same form with different coefficients 
$E^k_{\alpha_1,..,\alpha_N;s},\; s=1,..,S$ and each of them could be 
treated individually with its own limiting pattern.\\
 
For the purpose of this paper, the 
above, finite number of patterns will be sufficient. Then we have 
\begin{Proposition} \label{prop3.1} ~\\
The objects (\ref{3.50}) are quadratic forms on the Fock space with dense 
form domain ${\cal D}$ obtained from (\ref{3.52}) or (\ref{3.52a}) via the 
weak limits (\ref{3.54}) subordinate to any limiting pattern $P$, i.e. 
all those limits corresponding to different $P$ coincide.
\end{Proposition}
Note that the monomials of different order $N$ or of different
number of creation operators $k$ in (\ref{3.52}) are subjected 
to the limiting process individually, e.g. the first two terms of order $N=1$  
in the Gauss constraint are treated in the sense of (\ref{3.54})
independently of each other and independently 
of any of the three terms of order $N=2$ in accordance 
with item v. of above definition. The same applies for the various $N=3$ 
contributions among themselves, the various $N=4$ contributions among 
themselves and mutually independently of each other  
in the Hamiltonian constraint. The expression (\ref{3.52a}) is a special
case of (\ref{3.52}) in which we identify the cut-offs of those labels 
$\alpha,\beta$ respectively that originate from the same $F_\alpha,G_\beta$
respectively.\\
\begin{proof}
We sandwich (\ref{3.52}) between fixed vectors $\psi,\psi'\in {\cal D}$. Then 
we find $M_0\in \mathbb{N}$ such that $\psi,\psi\in {\cal D}_{M_0}$, the 
finite linear span of Fock states satisfying $n_I=0$ for $I\not \in S_{M_0}$.
For each monomial of order $N=1,2,3,4$ that appears in (\ref{3.52}) we 
let annihilators act to the right on $\psi'$ and    
creators act to the left on $\psi$. Following the pattern, we start with the 
the $\alpha's$ corresponding to $U_1$ to take full range. 
As the annihilators respectively 
creators commute among 
each other respectively, we can take those annihilators and creators to
the outmost right and left respectively. Whenever such $\alpha$ violates 
$\alpha\le M_0$ the corresponding term in the sum vanishes. In this way
all those $\alpha$ are confined to $\alpha\le M_0$. We proceed the same 
way with $U_2$ etc.\\
\end{proof}
Remarks:\\
A.\\
Note that we could have considered also labelling the truncated Gauss  
constraint by only two cut-offs $M_1,M_2$ where $M_1$ labels the first 
order monomial and $M_1,M_2$ label the second order monomial. The same 
applies to the truncated Hamiltonian constraint which could be labelled 
by only four cut-offs $M_1,..,M_4$ where only the first three label the 
third order monomial. This corresponds to a certain coincidence limit
of cut-off labels  
between quadratic forms of different 
monomials order and is much more restrictive than 
item v. of definition \ref{def3.1}. Similarly the labelling (\ref{3.52a}) 
corresponds to a coincidence limit that is obtained by expanding 
the continuum objects $F=A-A^\ast,\; G=A+A^\ast$ with respect to the $e_I$
thereby obtaining $F_\alpha=A_\alpha-A_\alpha^\ast,
G_\alpha=A_\alpha+A_\alpha^\ast$
and then decomposing the quadratic form into the $A_\alpha, A_\alpha^\ast$ 
so obtained while 
the more general scheme first decomposes the quadratic form into the 
continuum $A,A^\ast$ and then expands each factor $A,A^\ast$ individually 
with respect to the $e_I$ in each monomial. This therefore 
corresponds to a coincidence limit between cut-off labels of  
quadratic forms of the same monomial order but different numbers of 
creation operators    
$F_\alpha=A_\alpha-A_\alpha^\ast,\; G_\alpha=A_\alpha+A_\alpha^\ast$. 
which is also much more restrictive than item v. of the same definition.
While for the objects 
(\ref{3.52}) these coincidence limits exists and can be shown 
to coincide with the above limit these coincidental limits offer less 
flexibility in the choice of limiting patterns which becomes 
vital in the commutator calculations below.\\ 
B.\\
In  \cite{29} we have applied truncation structures and limiting patterns 
already to the spatial diffeomorphism constraint. There the mixture 
issue mentioned in A. with respect to $N$ 
does not arise as this is a monomial $N=2$ constraint
and we have applied the same 
total coincidence limit $M_1=M_2=M\to \infty$
individually for all three occurring normal ordered monomials of the form 
$A_\alpha^\ast\; A_\beta^\ast,\; A_\alpha^\ast\; A_\beta,\;
A_\alpha A_\beta$ i.e. we have taken a coincidence limit with respect to 
the $k=0,1,2$ labelling. What is more, when computing the commutator 
as defined below, we not only took 
$M_1=M_2=M$ for the first SDC quadratic form and  
$M'_1=M'_2=M'$ for the second SDC quadratic form but even $M=M'$.
This is not necessary and just one of the 
possibilities of taking the limit, leading however 
to the same unambiguous result 
in this case.\\
\\
When computing commutators of the above constructed quadratic forms
we need the following 
\begin{Lemma} \label{la3.1} ~\\
We have 
\ba \label{3.55}
&&[A_{\alpha_1}^\ast..A_{\alpha_k}^\ast\; A_{\beta_1}..A_{\beta_l},
A_{\gamma_1}^\ast..A_{\gamma_p}^\ast\; A_{\delta_1}..A_{\delta_q}]
\\
&=&
A_{\alpha_1}^\ast..A_{\alpha_k}^\ast\;
[A_{\beta_1}..A_{\beta_l},A_{\gamma_1}^\ast..A_{\gamma_p}^\ast]\;
A_{\delta_1}..A_{\delta_q}
-
A_{\gamma_1}^\ast..A_{\gamma_p}^\ast\;
[A_{\delta_1}..A_{\delta_q},A_{\alpha_1}^\ast..A_{\alpha_k}^\ast]
\;A_{\beta_1}..A_{\beta_l}
\nonumber
\ea
and 
\be \label{3.56}
[A_{\beta_1}..A_{\beta_l},A_{\gamma_1}^\ast..A_{\gamma_p}^\ast]\;
=([A_{\gamma_1}..A_{\gamma_p},A_{\beta_1}^\ast..A_{\beta_l}^\ast])^\ast
\ee
and for $k\le l$
\be \label{3.57}
[A_{\alpha_1}..A_{\alpha_k},A_{\beta_1}^\ast..A_{\beta_l}^\ast]\;
=\sum_{n=1}^k\; \sum_{1\le r_1<..<r_n\le k; 1\le s_1<..<s_n\le l}\;
[\sum_{\pi\in S_n}\;\prod_{m=1}^n\;\delta_{\alpha_{r_{\pi(m)}}\beta_{s_m}}]\;
[\prod_{s\not\in\{s_1,..,s_n\}} A_{\beta_s}^\ast]\;
[\prod_{r\not\in\{r_1,..,r_n\}} A_{\alpha_r}]
\ee
\end{Lemma}
The proof is not difficult: (\ref{3.55}) follows directly from the basic
CCR, (\ref{3.56}) directly from the definition of the adjoint which allows 
us to assume $k\le l$ in (\ref{3.57}) and (\ref{3.57}) follows by induction 
over $k$. Lemma \ref{la3.1} gives the explicit formula for restoring 
normal order in commutators of normal ordered monomials.\\
\\
Suppose now that we have quadratic forms 
\be \label{3.58}
Q_j=\sum_{N=0}^{N_j}\;\sum_{k=0}^N\; [Q_j]^k_N,\;\;
[Q_j]^k_N:= 
[Q_j]^k_{N;\alpha_1..\alpha_N}\; A_{\alpha_1}^\ast..A_{\alpha_k}^\ast
A_{\alpha_{k+1}}..A_{\alpha_N}
\ee
with $j=1,2$ which as in (\ref{3.50}) are uniquely defined using any limiting 
pattern so that the formal infinite range sums are in fact finite in any 
matrix element calculation. To compute the quadratic form commutator 
$[Q_1,Q_2]$, if it exists, we have three possibilities.\\ 
\\
Possibility 1: {\it Truncating after commuting}\\
We use the the expressions (\ref{3.58}) with all labels $\alpha$
unconstrained, compute the commutator $[Q_1,Q_2]$ term by term 
corresponding to the terms labelled by 
$N,k$ and $N',k'$ respectively with $N\le N_1, k\le N,\;N'\le N_2,
k'\le N'$ respectively at fixed 
labels $\alpha$ and restore normal order using lemma \ref{la3.1}.  
Then for each normal ordered monomial that is produced we individually 
truncate the sums over all occurring 
$\alpha$'s in that normal ordered monomial and finally for each of those 
monomials investigate the possible limiting patterns which result in 
well defined quadratic forms.\\
Possibility 2: {\it Truncating before commuting}\\
We use expression (\ref{3.58}) and truncate the sums over 
$\alpha_1,..,\alpha_N$ in the monomial corresponding to $N,k$ 
by $\alpha_l\le M_{j,N,k,l};\;j=1,2;\;N=1,..,N_j\;
k=0,..N;\; l=1.,,.N$, then we compute the commutator of 
$[Q_{1,\{M_1\}},Q_{2,\{M_2\}}]$ of the so truncated quadratic forms where 
$\{M_j\}$ stands for the collection of those $M_{j,N,k,l}$, restore 
normal order using lemma \ref{3.1} and finally investigate the possible 
limiting patterns that result in a well defined quadratic form.\\
Possibility 3: {\it Mixture}\\
One applies possibility 1 in the sense that one defines
\be \label{3.59}
[Q_1,Q_2]:=\sum_{N=0,N'=0}^{N_1,N_2}\;\sum_{k,k'=0}^{N,N'}\;
[[Q_1]^k_N,[Q_2]^{k'}_{N'}]
\ee
and then applies possibility 2 to each term in the finite sum (\ref{3.59}).\\
\\
The difference between the first two possibilities is that in the first 
possibility we have more truncation lables than in the second possibility.
This is because in the second possibility only the original truncation labels
$\{M_1\}, \{M_2\}$ appear in all possible monomials that are produced 
by the normal ordering identity (\ref{3.57}) while in the first possibility 
each term on the right hand side of (\ref{3.57}) gets equipped with its own 
individual truncation labels. Thus the first possibility allows for 
more limiting patterns and thus more possibilities to generate a well
defined quadratic form. For the same reason it potentially also generates
more ambiguity as there may be more than one way to generate a well 
defined quadratic form. One may argue for the first possibility by taking 
the point of view that the non-truncated quadratic forms $Q_1,Q_2$ are 
well defined, thus the commutator should be based on those and the 
{\it afterwards}  termwise truncation in the decomposition into normal 
ordered expressions
is part of the regularisation procedure for the commutator just as it 
is for any quadratic form by definition \ref{def3.1} as it was also 
used to defined $Q_1,Q_2$. One may argue for the second 
possibility by taking the point of view that products of quadratic forms 
are ill-defined and thus they should be truncated before computing the 
commutator. That each term $[Q_j]^k_N$ is equipped with its individual 
truncation labels $M_{j,N,k,l}$ is justified by the assumption that all
limiting patterns on $\{M_j\}$ result in the same quadratic form $Q_j$.
Finally the mixed point of view can be justified when each term $[Q_j]^k_N$ 
by itself is a well defined quadratic form so that each of them by itself 
is obtained via any of the limiting patterns. Thus (\ref{3.59}) is 
a possible definition if each of the individual 
$[[Q_1]^k_N,[Q_2]^{k'}_{N'}]$ can be turned into a well defined quadratic 
form.\\
\\
Since for each of these prescriptions and limiting patterns the 
unconstrained $[Q_1,Q_2]$ is formally obtained, all of these possibilities 
appear 
to be equally justified a priori. It is similar to the problem of giving 
meaning 
to a formal multiple sum series of complex numbers (in our case obtained 
as matrix elements) which is not absolutely convergent 
but may have 
several sequential conditionally convergent sums. Thus 
giving a concrete conditional summing prescription is {\it part of 
the definition} of such a series.
\begin{Example} \label{ex3.1} ~\\
Let for $m,n\in \mathbb{N}_0$ 
\be \label{3.58a}
a_{m,4n}:=[m+1+4n]^{-1},\;
a_{m,4n+1}:=-[m+1+4(n+1)]^{-1},\;
a_{m,4n+2}:=-[m+2+4n]^{-1},\;
a_{m,4n+3}:=[m+2+4(n+1)]^{-1}
\ee
Then the series $s:=\sum_{(m,n)\in\mathbb{N}_0^2}\; a_{m,n}$ is ill defined as 
it stands as it is not absolutely convergent. The row sums 
\be \label{3.58b}
\sum_m a_{m,n}=\pm \infty
\ee
are plain divergent. The column sums are conditionally convergent
\ba \label{3.58b1}
&& \sum_n a_{m,n}:=
\sum_{n=0}^\infty\{
\frac{1}{m+1+4n}-\frac{1}{m+1+4(n+1)}\}
-\sum_{n=0}^\infty\{
\frac{1}{m+2+4n}-\frac{1}{m+2+4(n+1)}\}
\nonumber\\
&=&
4\sum_{n=0}^\infty
\frac{1}{(m+1+4n)(m+1+4(n+1))}
-4\sum_{n=0}^\infty
\frac{1}{(m+2+4n)(m+2+4(n+1))}
\nonumber\\
&=&\frac{1}{m+1}-\frac{1}{m+2}=\frac{1}{(m+1)(m+2)}
\ea
and thus the sequential and conditional sum
\be \label{3.58c}
s:=\sum_m\;\sum_n\; a_{m,n}=1
\ee
converges in this limiting pattern.
\end{Example}
Note that this example is smilar to but different from examples 
illustrating Riemann's rearrangement theorem \cite{33} which states 
that a conditionally convergent series $\sum_N c_N$ of 
real numbers can be made to converge 
to any real number by using a bijection 
$b: \mathbb{N}\to \mathbb{N};\; N\mapsto b(N)$ on the indices 
and considerung $\sum_N c_{b(N)}$ instead. While $\mathbb{N}_0^2$ is in 
bijection with $\mathbb{N}$ (e.g. via Cantor's map \cite{33} $B$) and we can thus 
consider above $a_{m,n}=a_{B(N)}=c_N$ as a series in $c_N$, summing rows or 
columns first does not correspond to $N\to \infty$ but rather to 
taking limits of subseries of $c_N$ and then summing the limits 
of those subseries. 
The rearrangement theorem however can be applied to those subsequences 
which in the above example are also only conditionally convergent.\\  
~\\
In what follows we explore the mixed possibility 3, i.e. 
we will treat linear and bi-linear terms in the Gauss constraint as well
as the tri-linear and quadri-linear terms in the Hamiltonian constraint as 
individually defined quadratic forms and then consider for those truncation 
before commutation.  Moreover, we will constrain the limiting patterns 
to be such that in normal ordered expressions of the form
$: F_{\alpha_1} .. F_{\alpha_k} G_{\alpha_{k+1}}..G_{\alpha_N}:$ we truncate 
$\alpha_k\le M_k$. Definition \ref{def3.1} allows to be more general, e.g. 
in the normal ordered decomposition
\be \label{3.59a} 
\sum_{\alpha,\beta} E_{\alpha\beta} :F_\alpha G_\beta:=
\sum_{\alpha,\beta} E_{\alpha\beta} A_\alpha A_\beta
+[-E_{\alpha\beta}+E_{\beta\alpha}]\; A_\alpha^\ast A_\beta 
-E_{\alpha\beta} A^\ast_\alpha A^\ast_\beta
\ee
we could truncate $\alpha,\beta$ in the the first, second, third term
respectively by $M_1,M_2$ and $M_3,M_4$ and $M_5,M_6$ respectively instead
of synchronising $M_1=M_3=M_5, M_2=M_4=M_6$. This is exactly the difference 
between (\ref{3.52a}) and (\ref{3.52}). Thus we choose to work with 
(\ref{3.52a}). The advantage will be a much reduced computational 
effort using the fact that $[F_\alpha,F_\beta]=[G_\alpha,G_\beta]=0$.
\begin{Proposition} \label{prop3.2} ~\\
Consider the truncated constraint ingredients
\be \label{3.60a}
G_{1;M_1}(r),\;G_{2;M_1,M_2}(r),\;C_{2;M_1,M_2}(u),\;
C_{0;3,M_1..M_3}(f),\;C_{0;4,M_1..M_4}(f)
\ee
defined in (\ref{3.52a}) and 
define the constraint commutators 
by 
\ba \label{3.60b}
[G(r),G(s)] &:=& \frac{1}{2}\;\{
[G_1(r),G_1(s)]-[G_1(s),G_1(r)]\}  
+\frac{1}{2}\{[G_2(r),G_2(s)]-[G_2(s),G_2(r)]\}
\nonumber\\
&& +\{[G_1(r),G_2(s)]-[G_1(s),G_2(r)]\}
\nonumber\\
{[}C(u),C(v)] &:=& \frac{1}{2}\;\{[C_2(u),C_2(v)]-[C_2(v),C_2(u)]\}
\nonumber\\
{[}C(u),G(r)] &:=& \{[C_2(u),G_1(r)]\}+\{[C_2(u),G_2(r)]\}
\nonumber\\
{[}G(r),C_0(f)] &:=&
\{[G_1(r),C_{0,3}(f)]\}+\{[G_2(r),C_{0,3}(f)]\}
+\{[G_1(r),C_{0,4}(f)]\}+\{[G_2(r),C_{0,4}(f)]\}
\nonumber\\
{[}C(u),C_0(f)] &:=&
\{[C_2(u),C_{0,3}(f)]\}+\{[C_2(u),C_{0,4}(f)]\}
\nonumber\\
{[}C_0(f),C_0(g)] &:=& \frac{1}{2}\;\{
[C_{0;3}(f),C_{0;3}(g)]-[C_{0;3}(g),C_{0;3}(f)]\}
+\frac{1}{2}\;\{[C_{0;4}(f),C_{0;4}(g)]-[C_{0;4}(g),C_{0;4}(f)]\}
\nonumber\\
&& +\{[C_{0;3}(f),C_{0;4}(g)]-[C_{0;3}(g),C_{0;4}(f)]\}
\ea
Consider each curly bracket expression on the right hand side of 
(\ref{3.60b}) individually. Each such curly bracket involves either one 
or two commutators. If one commutator appears, we just replace the 
quadratic forms by their correspondent in the list (\ref{3.60a}). 
E.g. $[C_2(u),C_{0;3}(f)]$ is replaced by  
$[C_{2;M_1,M_2}(u),C_{0;3;M_3,..,M_5}(f)]$ with arbitrary finite 
cut-offs $M_1..M_5$. 
If two commutators appear we replace the quadratic forms also 
by the correspondent in the list (\ref{3.60a}) but we use the same cut-offs 
in both commutators. For example we replace 
$[C_{0;3}(f),C_{0;4}(g)]-f\leftrightarrow g$ by 
$[C_{0;3;M_1,..,M_3}(f),C_{0;4;M_4,..,M_7}(g)]- f\leftrightarrow g$ with 
arbitrary finite cut-offs $M_1,..,M_7$.
  
Then for each such curly bracket there exists 
a limiting pattern such that 1. the limiting process results 
in a well defined 
quadratic form and 2. those resulting quadratic forms add up 
exactly to the normal 
ordered quadratic forms corresponding to the Fock quantisations of the 
right hand sides of the corresponding classical Poisson bracket calculations.
That is
\be \label{3.60}  
[\;:Z_1:\;,\; :Z_2:\;]=i\;:\;\{Z_1,Z_2\}\;:
\ee
where $Z_1,Z_2$ are any of the $G(r), C(u), C_0(f)$. 
\end{Proposition}
Before going through the straightforward but tedious and lengthy proof
a few remarks are in order:\\
\\
{\it 1. Sources of divergences in commutator calculations and 
mechanisms to avoid them}\\
\\ 
When we use (\ref{3.57}) to restore normal order 
in a commutator we end up with normal ordered monomials that come 
with coefficients $D$ that are products of the coefficients of the quadratic 
forms that we compute the commutator of, times products of Kronecker symbols.
Thus such a term is of the form 
\be \label{3.61}
\sum_{\alpha_1\le M_1..\alpha_L\le M_L}\; 
D_{\alpha_1..\alpha_L}\; \prod_{k=1}^n\; \delta_{\alpha_{r_k},\alpha_{s_k}}\;
:\prod_{t\in\{1,..,L\}-\{r_1,..,r_n,s_1,..,s_n\}}\;C_{\alpha_s\sigma_s}:\; 
\ee
where $\sigma=\pm, C_{\alpha,+}=A_\alpha^\ast, C_{\alpha,-}=A_\alpha$ and
$\{r_1,..,r_n,s_1,..,s_n\}\subset \{1,..,L\}$ with $2n\le L$ and $n\ge 1$.
Clearly the labels $\alpha_t,\; t\not \in \{r_1,..,r_n,s_1,..,s_n\}$ are 
automatically truncated when sandwiching (\ref{3.61}) between Fock states.
Thus it is the sum over the remaining labels that is in danger of 
diverging. It turns out that only terms with $n\ge 2$ are dangerous. 
This can be seen from the typical example 
$D_{\alpha_1..\alpha_4}=<e_{I_1} e_{I_2}, f>\; <e_{I_3} e_{I_4}, g>$
where $f,g$ are smearing functions of the constraints.

If we now  perform the sum over $I_1=I_3$ then by the completeness relation 
we end up with $<e_{I_2}\; e_{I_4}, f\;g>$ which is finite. If we perform 
a second sum over $I_2=I_4$ then we end up with $<[\sum_I e_I^2], fg>$
which diverges because $\sum_I e_I(x)\;e_J(y)=\delta(x,y)$ so that 
\be \label{3.62}
\Lambda(x):=\lim_{M\to \infty}\; \sum_{I\in S_M}\; e_I(x)^2
\ee
diverges being formally equal to $\delta(x,x)$. Such diverging sequences 
appear all over the place in the {\it product} of quadratic forms. 
However, the commutator maybe better behaved because before we sum over 
$I_2=I_4$ there is a second term $<e_{I_2}\; e_{I_4}, g\;f>$ 
which comes from the fact that 
$[Q(f),Q(g)]=\frac{1}{2}([Q(f),Q(g)]-[Q(g),Q(f)]$ and these two terms would in
fact cancel (so called ultra locality). This is the reason for why 
we considered the curly bracket expressions with two commutators and equally 
chosen cut-offs so that such a cancellation can occur.
This also shows why a sequential limiting pattern is of advantage.

On the other hand, we may have a term with a derivative of 
a smearing function
$D_{\alpha_1..\alpha_4}=<e_{I_1} e_{I_2}, f>\; <e_{I_3} e_{I_4}, g_{,j}>$.
Performing the first sum still delivers a finite result
$<e_{I_2}\; e_{I_4}, f\;g_{,j}>$ but subtracting the term with $f,g$ 
interchanged in this case does not vanish but yields
$<e_{I_2}\; e_{I_4}, [f\;g_{,j}-f_{,j} g]>$. Now performing the second sum 
again gives $<[\sum_I e_I^2]\;[f\;g_{,j}-f_{,j} g]>$ which diverges. 
Thus, to avoid such divergences one will try to find a limiting pattern 
which avoids derivatives of smearing functions $f,g$ of the constraints. 
For instance in 
\be \label{3.62a}
<e_I\; f, e_{J,j}>\; <e_K\; g, e_L> \; \delta_{IK}\delta_{JL} 
- f\leftrightarrow g
\ee
carrying out the sum over $I=K$ first gives zero by the same mechanism as 
above while carrying out the sum over $J=L$ first gives (integrations 
by parts do not create boundary terms because the $e_I$ are Schwartz functions
for $\sigma=\mathbb[R]^3$ or all functions are periodic for $\sigma=T^3$)
\be \label{3.63}
-<[e_I\; f]_{,j} e_K, g> \; \delta_{IK} 
- f\leftrightarrow g
=<[e_I\; e_K, [g_{,j} f- f_{,j} g]> \; \delta_{IK} 
\ee
and the sum over $I=K$ is again divergent. This demonstrates that the 
existence or non existence of the commutator as a quadratic form 
can depend rather drastically on the limiting pattern which is reminiscent 
of the situation with double sum series mentioned in example \ref{ex3.1}. 

Now the observation 
is that the constraints except for the linear term in the Gauss constraint
which due to its linearity cannot produce more than one Kronecker symbol 
and thus does not suffer from the above danger, all constraints contain the 
$E^a_j$ witout derivatives (if we use $F_{ab}^j E^a_j=C_a+A_a^j G_j$ 
instead of $C_a$). Thus a promising limiting pattern is such that 
one carries out the sum over those $\alpha$ that label the $F_\alpha$ before 
carrying out the sum over those $\alpha$ that label the $G_\alpha$ which 
do come with derivatives. This strategy turns out to be successful here and 
may carry over to the Lorentzian constraints in all dimensions formulated 
in terms of ADM variables $P,q$ where only $q$ carries (up to second order) 
derivatives but not $P$. \\
\\
2. {\it General strategy of the computation}\\
\\
Step 1: At finite cut-offs we simply decompose a term of the form 
\be \label{3.63a}
[:F_{\alpha_1}..F_{\alpha_p}\;G_{\beta_{p+1}}..G_{\beta_{m}}:,\;
F_{\beta_1}..F_{\beta_q}\;G_{\beta_{q+1}}..G_{\beta_n}]
\ee
into annihilation and creation operators. This yields 
$2^{m+n}$ terms of the form 
\be \label{3.63b}
[A_{\gamma_1}^\ast..A_{\gamma_k}^\ast\;A_{\gamma_{k+1}}..A_{\gamma_m},\;
A_{\delta_1}^\ast..A_{\delta_l}^\ast\;A_{\delta_{l+1}}..A_{\delta_n}]
\ee
where $0\le k\le m,\; 0\le l\le n$ and the $\gamma,\delta$ respectively 
are permutations of the $\alpha,\beta$ labels respectively.\\
\\
Step 2: Then we apply lemma \ref{la3.1} to (\ref{3.63}). 
For each $u=1,..,{\sf min}(m-k,l)$ this gives 
${m-k \choose u}\;{l \choose u}$ 
and for each $u=1,..,{\sf min}(k,n-l)$ this gives 
${k \choose u}\;{n-l \choose u}$ normal ordered
products of $k+l-u$ creation and $m+n-k-l-u$ annihilation operators
respectively 
with coefficients that are totally symmetric 
linear combinations of $u!$ products of $u$ 
Kronecker symbols.\\ 
\\
Step 3: We now substitute 
$A_\alpha=[G_\alpha+F_\alpha]/2,\;
A_\alpha^\ast=[G_\alpha-F_\alpha]/2$ within the normal ordering symbol
which gives $2^{m+n-2u}$ terms at given $u$ in step 2.
Since within the normal ordered symbol we can permute all $F,G$ operators,
we can rearrange the commutator into normal ordered monomials of 
the $F,G$ with coefficients which involve at least one and 
at most ${\sf min}(m,n)$ 
products of Kronecker symbols of the form 
$\delta_{\alpha_i,\beta_j},\; i=1,..,m; j=1,..,n$. Note that for 
$m=n, p=q$ there will be no term with $n$ factors of Kronecker symbols.
This is because, suppressing indices, in 
$[:F^p G^{m-p}:,:F^p G^{m-p}:]$ the contributions to $\delta^m$ 
are coming 
from 
\be \label{3.63c}
[A^m+(-1)^p (A^\ast)^m,A^m+(-1)^p (A^\ast)^m]=
(-1)^p\; (
[{\sf symm}\otimes^m \delta]-[{\sf symm}\otimes^m \delta]^T)=0
\ee
where symm means total symmetrisation and T transposition.\\
\\
Step 4: To estimate the number of terms that one has to manipulate 
in these first 3 steps we note that of the $2^{m+n}$ 
terms obtained from step 1 there are 
${k \choose m}\; {l \choose n}$ 
terms of the type (\ref{3.63c}), hence 
we obtain 
\be \label{3.63d}
N_{m,n}=
\sum_{k=0}^m\;\sum_{l=0}^n\;{m \choose k}\; {n \choose l}\;
[
\sum_{u=1}^{{\sf min}(m-k,l)}\; u!\; {m-k \choose u}\;{l \choose u}\; 
2^{m+n-2u} 
+
\sum_{u=1}^{{\sf min}(n-l,k)}\; u!\; {n-l \choose u}\;{k \choose u}\; 
2^{m+n-2u} 
]
\ee
terms in total. To estimate the number of these terms we 
approximate the binomial distribution 
by the Poisson distribution  $2^{-k} {k \choose u}\approx e^{-1/2}
\frac{2^{-u}}{u!}$. This allows to estimate the sums over $u$ by extending 
to $u=\infty$. Then  
\ba \label{3.63e}
&& N_{m,n}\approx
\sum_{k=0}^m\;\sum_{l=0}^n\;{m \choose k}\; {n \choose l}\;
2^{m+n}\;
e^{-1}\;[\sum_{u=0}^\infty\;\frac{2^{-4u}}{ u!}]
[2^{m-k+l}+2^{n-l+k}]
\nonumber\\
&=& 2^{m+n}\;e^{15/16}\;
\sum_{k=0}^m\;\sum_{l=0}^n\;{m \choose k}\; {n \choose l}\;
[2^{m-k+l}+2^{n-l+k}]
\ea
Once again substituting for the Poisson distribution alllows to 
perform the sum over $k,l$ which we extend to $\infty$ 
\ba \label{3.63f}
N_{m,n}\approx
&=& 
2^{m+n}\;e^{-31/16}\;
\sum_{k=0}^m\;\sum_{l=0}^n\;
[2^{m-k+l}+2^{n-l+k}]\; \frac{2^{m-k}}{k!}\frac{2^{n-l}}{l!}
\nonumber\\
&\approx&
4^{m+n}\;e^{-31/16}\;
\sum_{k=0}^m\;\sum_{l=0}^n\;
[2^{-2k}+2^{-2l}]
\frac{1}{k!\;l!}
=2\; 4^{m+n}\;e^{-13/16}\;
\ea
Thus we obtain an order of $4^{m+n}$ operations. 

With an eye towards 
the computation for Lorentzian signature in $q,P$ variables with 
top polynomial degree $m=n=10$ in 3 spatial 
dimensioons this yiels $10^{12}$ terms to consider.
For the present Euclidian calculation we have top degree $m=n=4$ which yields 
merely $10^5$ terms to consider. Yet, the final expression
in terms of normal ordered monomials of $F,G$ will contain
only a small number of linearly independent terms. This is due to 
the fact that the original expression (\ref{3.63a}) is totally 
symmetric in the four groups of indices corresponding to the 
$p,q$ factors of $F$ respectively and the $m-p, n-q$ factors of 
$G$ respectively. While in the present case a calculation by hand is still
possible and was indeed done by the author, it is still a good idea 
to use computer algebra tools to verify the calculations. In the present case
we wrote a simple, not yet speed optimised code using pthon's SymPy package 
for symbolic calculus
which in principle can deal with any $m,n$. For $m=n=4$ the computation 
takes a few seconds with a 2.3MHz Intel core i5 processor. 
Therefore the
computation will definitely be doable in Lorentzian signature in
$(p,q)$ variables using parallel computing which on the author's machine 
would momentarily take half a year.\\
\\
Step 5: The crucial step is now to show that there exists at least 
one limiting pattern such that the terms involving more than 
one factor of Kronecker $\delta$ factors yield a finite result 
in the limiting process subordinate to that pattern.
Note that the classical Poisson bracket 
calculation gives precisely one factor of Kronecker $\delta$ and 
therefore this term returns the result of the classical calculation 
automatically normal ordered. 
The chosen pattern must be the same for the terms 
involving $1,2,..,{\sf min}(m,n)-\delta_{m,n}$ factors of Kronecker
$\delta$ because we have chosen truncation before commuting but 
individually for each $F,G$ monomial into which a constraint 
decomposes which we regard as quadratic forms in their own right.\\
\\
\begin{proof}:\\
\\
{\it Organisation:}\\
We compute commutators in increasing order of complexity and 
aim to show that there exist limiting pattern such that (all 
constraints are normnal ordered by definition)\\
1. Gauss-Gauss: $[G(r),G(s)]=i\;G(-r\times s)$.\\
2. Diffeomorphism-Diffeomorphism: $[C(u),C(v)]=i\;C(-[u,v])$.\\   
3. Diffeomorphism-Gauss: $[C(u),G(r)]=i\;G(-u[r])$.\\   
4. Gauss-Hamiltonian: $[G(r),C_0(f)]=0$.\\
5. Diffeomorphism-Hamiltonian: $[C(u),C_0(f)]=i\;H(-u[f])$.\\   
6. Hamiltonian-Hamiltonian
\be \label{3.64}
[C_0(f),C_0(g)]=i\;\int\;d^Dx [4(f g_{,a}-f_{,a} g)]\; 
:\;Q^{ab}\; F_{bc}^j \; E^b_j\;:
\ee   
The last relation is, like the others, exactly $i$ times 
the normal ordering of the 
classical Poisson bracket 
expression and thus {\it non-anomalous}. However, while 
$Q^{ab} F_{bc}^j E^b_j=Q^{ab}\;[C_b+A_b^j G_j]$ is a linear combination  
of constraints with second order and third order monomial 
structure functions we have 
\be \label{3.65}
:\;Q^{ab}\; F_{bc}^j \; E^b_j\;:\not=
C_b Q^{ab}+G_j A_a^j
\ee
because normal ordering is in particular symmetric. This is because the 
expression on the right hand side of (\ref{3.65}) is ill defined on the 
Fock space even as a quadratic form.  There is no contradiction 
for a solution $l$ of the constraints to disobey $l([C_0(f),C_0(g)]\psi)=0$
even if $l[C(f)\psi]=0$ for all $f,\psi\in {\cal D}$ as $C_0(f)\psi\not\in 
{\cal D}$ as we have emphasised before.\\
\\
To organise the calculation we follow the above five steps applied 
to the commutators in the order 1.-6. listed above. Now the Gauss 
constraint consists of $G=G_1+G_2$ where $G_1,G_2$ are of type
$:F:,\;:F\; G:$ respectively, the spatial diffeomorphism constraint
is $C=C_2$ which is of type $:F\; G:$ as well and finally the Hamiltonian 
constraint is $C_0=C_{0;3}+C_{0;4}$ where $C_{0;3},\; C_{0;4}$ are of 
type $:F\; F\; G:,\; :F\; F\; G\;G:$ respectively. Thus we have 4 different 
types of normal ordered monomials of $:F^p\; G^{m-p}:$ 
to consider with $(p,m)=(1,1), \; (1,2),\; (2,3),\; (2,4)$ and correspondingly
10 different commutators of type (\ref{3.63a}) corresponding to 
$(q,n)=(1,1), \; (1,2),\; (2,3),\; (2,4)$ as well and 
it suffices to consider $m\le n$. This part of the calculation can 
be performed following steps 1-3 above and is listed in appendix \ref{sa}.
It has been obtained both by hand and by using computer algebra.
After that we consider the terms with at least two factors of Kronecker
symbols starting with the highest number of factors. As we have 
explained above, two Kroneckers can only be obtained for 
$m=2<n\le 4,\; 3\le m\le n\le 4$ (2+3 
terms) and  
three Kroneckers can only be obtained for $m=3<n\le 4,\; m=n=4$ (1+1 terms).
All other commutators are safe of singularities and the precise 
limiting pattern will be immaterial. This means that the 
following comutators in total are unproblematic: Gauss -- Gauss, 
Gauss -- Diffeo, Diffeo -- Diffeo. The fact that 
there is no anomalous 
term in the Diffeo -- Diffeo commutator appears to 
be in conflict with the fact that in one 
spatial dimension there is the well known Virasoro central term in the 
constraint algebra generated by the spatial diffeomorphism $D$ 
and Hamiltonian 
constraint $C$, say in parametrised field theory PFT \cite{31},
given the fact that the spatial diffeomorphism constraint is 
theory independent in the sense that it follows purely algebraically
from the tensor type that one considers, the information 
about the concrete Lagrangian of the theory is not encoded in the spatial
diffeomorphism constraint but rather in the Hamiltonian constraint.
However, a closer look reveals the following: In PFT  
one considers instead of $C,D$ (in one spatial dimension these 
are quadratic in the fields and their conjugate momenta) 
two Poisson 
commuting constraints $C_\pm=D\pm C$ with Poisson algebras ismorphic to 
that of $D$ and the algebra of the $C_\pm$ {\it is} anomalous by central 
terms. However, these two central terms cancel in the algebra 
of the $D=C_+ + C_-$ while they add in $C=C_+ - C_-$. Thus only 
$C$ is anomalous in PFT, there is no conflict as far 
as ``universal'' (theory independent) part $D$ of the algebra is concerned. 
The absence of the anomaly in Euclidian GR in 3+1 dimensions is also not in
conflict with the presence of the anomaly in PFT in 1+1 dimensions because 
this concerns the ``non-universal'' (theory dependent) part $C$ of the 
algebra which is very different in these two theories.\\
\\
{\it Sandwiching between Fock states}\\
In what follows we will use all the time the following ``sandwiching 
between Fock states'' argument: All commutators and normal reorderings are 
performed in a controlled way at finite cut-offs. The end result of that 
step depends on Kronecker symbols and surviving monomials of creation and 
annihilation operators in normal ordered form. We are seeking for 
weak limit of that object, i.e. we sandwich it between any Fock states 
$\psi,\psi'$ and remove the cut-offs according to some limiting 
pattern to be derived, and which is independent of $\psi,\psi'$.  
Now for any $\psi,\psi'$ we find some $M_0$ 
dependening on $\psi,\psi'$ such that $\psi,\psi'\in {\cal D}_{M_0}$.
As we take the cut-offs to infinity according to some pattern, eventually
even the smallest cut-off of that pattern
exceeds $M_0$. At this point all labels $\alpha$
which label one of the surviving annihilation and creation operators is 
constrained by $M_0$ rather than its assigend cut-off. The sum with 
respect to that label is no longer dependent on it so that the removal of the 
pre-assigned cut-off is trivial and the details of the limiting 
pattern get reduced to the remaining labels. The remaining task is then to 
solve the Kronecker symbols, to determine a limiting pattern consistent 
with them and to carry out the remaining sums in the prescribed order. 
Thus if a Kronecker symbol $\delta_{\alpha\beta}$ is given with 
$\alpha\le M_1, \beta\le M_2$ and the pattern is such that $M_1\le M_2$ 
then we must solve for $\beta$ rather than $\alpha$. The end result 
will then depend on finite sums constrained by $M_0$ which then can 
be recognised as a quadratic form sandwiched between those $\psi,\psi'$ 
and the constraint on the sums by $M_0$ can be taken away so that the 
final result is independent of $\psi,\psi'$.\\ 
\\
\\
{\bf 1. Gauss-Gauss}\\
\\
We have to compute three curly brackets. \\
\\
{\bf 1A.}\\
The first curly bracket is, see (\ref{a.1})
\be \label{3.66}
[G_{1;M_1}(r),G_{1;M_2}(s)]
-[G_{1;M_1}(s),G_{1;M_2}(r)]=0-0=0
\ee
It vanishes trivially for any $M_1,M_2$ 
because $G_{1;M}(r)$ just depends on the combinations 
$F_\alpha$ which mutually commute. Therefore the resulting 
commutator quadratic form,
obtained for any limiting pattern, is given by
%e.g. for the coincident choice $M_1=M_2\to \infty$, is
\be \label{3.66a}
[G_1(r),G_1(s)]=0
\ee
~\\
{\bf 1B.}\\
The second curly bracket is, see (\ref{a.2}) 
\ba \label{3.67}
&&[G_{1;M_1}(r),G_{2;M_2,M_3}(s)]-
[G_{1;M_1}(s),G_{2;M_2,M_3}(r)]
=\frac{1}{2\sqrt{2}}\;
\sum_{\alpha\le M_1,\;\beta\le M_2,\gamma\le M_3}\; 
(r_\alpha s_{\beta\gamma}-r\leftrightarrow s)\; 
[F_\alpha,:G_\beta\; F_\gamma:]
%A_\beta^\ast\; F_\gamma+F_\gamma A_\beta]
\nonumber\\
&=&
 \frac{1}{\sqrt{2}}\;
\sum_{\alpha\le M_1,\;\beta\le M_2,\gamma\le M_3}\; 
(r_\alpha s_{\beta\gamma}-r\leftrightarrow s)\; 
\;\delta_{\alpha,\beta}\; F_\gamma
\ea
By the sandwiching argument, $\gamma\le M_0$ gets locked. Pick 
e.g. the limiting pattern $M_1\to \infty$ before $M_2$
and carry out the Kronecker symbol in summing over $\alpha$
resulting in 
\be \label{3.70}
\sum_{\beta\le M_2,\gamma\le M_0}\; 
[(r_\beta s_{\beta\gamma}-r\leftrightarrow s)\; F_\gamma
\ee
The order in which we now take $M_2,M_3\to \infty$ is irrelevant. 
We now take $M_2\to \infty$ and 
have with $\beta=(Imi),\gamma=(Jnj)$ (again summation over repeated 
lower case latin letters is understood)
\be \label{3.71}
\sum_\beta r_\beta\;s_{\beta\gamma}
=\sum_I\; r_{Imi} s_{Imi,Jnj}  
=\sum_I\; <e_I,r^i_{,m}>\;<s^l\;e_J,e_I>\;\delta_{mn} \epsilon_{lij}  
=<e_J,r^i_{,n} s^l>\; \epsilon_{lij}  
\ee
Thus the final result is 
\ba \label{3.72}
&& 
\frac{1}{\sqrt{2}}
\sum_{I\le M_0} \;
[<e_I,r^i_{,m} s^l>- r\leftrightarrow s] \epsilon_{lij}\; F_{Imj}
\nonumber\\
&=&
\frac{1}{\sqrt{2}}
\sum_{I\le M_0} \;
<e_I,(r^i  s^l\epsilon_{lij}>)_{,m}\; F_{Imj}
=i [-\frac{i}{\sqrt{2}}
\sum_{I\le M_0} \;
(-r\times s)_{Imi}\;F_{Imi}]
\ea
where in the last step we applied
the product rule.
Thus the second curly bracket with e.g. the limiting pattern
$M_1$ before $M_2,M_3$ yields
\be \label{3.73}
[G_1(r),G_2(s)]=i\; G_1(-r\times s)
\ee
{\bf 1C.}\\
The third curly bracket is given by, see (\ref{a.5}) 
\ba \label{3.74} 
&& -\frac{1}{8}\sum_{\alpha\le M_1,\beta\le M_2,\gamma\le M_3,\delta\le M_4}\;
[r_{\alpha\beta}\;s_{\gamma\delta}-r\leftrightarrow s]\;
[:G_\alpha F_\beta:,:G_\gamma F_\delta:]
\nonumber\\
&=& -\frac{1}{4}\sum_{\alpha\le M_1,\beta\le M_2,\gamma\le M_3,\delta\le M_4}\;
[r_{\alpha\beta}\;s_{\gamma\delta}-r\leftrightarrow s]\;
[\delta_{\beta\gamma}\;:F_\delta\; G_\alpha: 
-\delta_{\delta\alpha}\;:F_\beta\; G_\gamma:]
\ea
Pick the pattern $M_2,M_4$ before $M_1= M_3$ and solve the kroneckers 
for $\beta,\delta$ respectively, then upon sandwiching this becomes 
\ba \label{3.77} 
&& -\frac{1}{4}\sum_{\alpha\le M_0,\gamma\le M_3,\delta\le M_0}\;
[r_{\alpha\gamma}\;s_{\gamma\delta}-r\leftrightarrow s]\;
:F_\delta\; G_\alpha: 
\nonumber\\
&& +\frac{1}{4}\sum_{\alpha\le M_1,\beta\le M_0,\gamma\le M_0}\;
[r_{\alpha\beta}\;s_{\gamma\alpha}-r\leftrightarrow s]\;
:F_\beta\; G_\gamma:
\nonumber\\
&=& -\frac{1}{4}\sum_{\alpha\le M_0,\beta\le M_0,\gamma\le M_3}\;
[r_{\alpha\gamma}\;s_{\gamma\beta}-r\leftrightarrow s]\;
:G_\alpha \; F_\beta: 
\nonumber\\
&& +\frac{1}{4}\sum_{\gamma\le M_1,\beta\le M_0,\alpha\le M_0}\;
[r_{\gamma\beta}\;s_{\alpha\gamma}-r\leftrightarrow s]\;
:G_\alpha\; F_\beta:
\nonumber\\
&=& -\frac{1}{2}\sum_{\alpha\le M_0,\beta\le M_0,\gamma\le M_3}\;
[r_{\alpha\gamma}\;s_{\gamma\beta}-r\leftrightarrow s]\;
:G_\alpha \; F_\beta: 
\ea
We can take now the remaining cut-off away and have 
with $\alpha=(Imi),\beta=(Jmj),\gamma=(Kuk),\delta=(Lvl)$
\be \label{3.77a}
\sum_\gamma\; r_{\alpha\gamma} s_{\gamma\beta} 
=\sum_J\; <r_p e_I,e_J>  \epsilon_{pij}\;\delta_{mn}
<e_J, s_q e_L> \epsilon_{qkl} \delta_{uv}
\delta_{jk} \delta_{nu} 
= <r_p s_q, e_I e_L> \;\delta_{mv} \; 2\delta_{p[l}\delta_{q]i}
\ee
This results in 
\ba \label{3.78}
&& -\sum_{I,J\le M_0}\; <e_I e_J, (r_p s_q-r_q s_p)>\; 
\delta_{p[l}\delta_{q]i} :G_{Imi}\; F_{Jml}:
=-\frac{1}{2}\sum_{I,J\le M_0}\; <e_I e_J, (r_l s_i-r_i s_l)>\; 
:G_{Imi}^\ast F_{Jml}:
\nonumber\\
&=& \frac{1}{2}\sum_{I,J\le M_0}\; (r\times s)^l_{IJ} \epsilon_{lij} \delta_{mn}
:G_{Imi}\; F_{Jnj}:
\nonumber\\
&=& i\; 
[\frac{i}{2} 
\sum_{\alpha\beta\le M_0}\; (-r\times s)_{\alpha\beta}
\;:G_\alpha\; F_\beta:]
\ea
Thus the third curly bracket e.g. with the limiting pattern $M_2,M_4$ before
$M_1=M_3$ 
\be \label{3.79}
[G_2(r),G_2(s)]=i\; G_2(-r\times s)
\ee
It follows altogether
\be \label{3.80}
[G(r),G(s)]= i\; G(-r\times s)
\ee
\\
{\bf 2. Diffeomorphism-Diffeomorphism}\\
\\
In this case the only curly bracket to consider is, see (\ref{a.5}) 
\ba \label{3.81}
-\frac{1}{8}
&&\sum_{\alpha\le M_1,\beta\le M_2,\gamma\le M_3,\delta\le M_4}\;
[u_{\alpha\beta} \;v_{\gamma\delta}-u\leftrightarrow v]\;
[:F_\alpha G_\beta:,:F_\gamma G_\delta:]
\\
&=&
-\frac{1}{4}
\;\sum_{\alpha\le M_1,\beta\le M_2,\gamma\le M_3,\delta\le M_4}\;
[u_{\alpha\beta} \;v_{\gamma\delta}-u\leftrightarrow v]\;
(\delta_{\alpha\delta}\; :F_\gamma G_\beta:
-\delta_{\gamma\beta}\; :F_\alpha G_\delta:)
\ea
In \cite{29} we showed that it is possible to choose the coincidence limit 
$M_1=..=M_4\to \infty$. We show here that we may choose more general 
patterns without changing the result. 

We pick the pattern $M_1, M_3$ before $M_2=M_4$ and use the sandwiching 
argument, so that we can solve the Kronecker symbols for $\gamma,\delta$ 
respectively and obtain
\ba \label{3.85}
&& -\frac{1}{4}
\sum_{\beta\le M_0,\gamma\le M_0,\delta\le M_2}\;
[u_{\delta\beta} \;v_{\gamma\delta}-u\leftrightarrow v]
\;:F_\gamma\; G_\beta: 
+\frac{1}{4}
\sum_{\alpha\le M_0,\beta\le M_2,\delta\le M_0}\;
[u_{\alpha\beta} \;v_{\beta\delta}-u\leftrightarrow v]
\;:F_\alpha\; G_\delta:
\nonumber\\
&=& -\frac{1}{4}
\sum_{\alpha\le M_0,\beta\le M_0,\gamma\le M_2}\;
\{[v_{\alpha\gamma} \;u_{\gamma\beta}-u\leftrightarrow v]
-[u_{\alpha\gamma} \;v_{\gamma\beta}-u\leftrightarrow v]\}
\;:F_\alpha\; G_\beta:
\nonumber\\
&=& \frac{1}{2}
\sum_{\alpha\le M_0,\beta\le M_0,\gamma\le M_2}\;
[u_{\alpha\gamma} \;v_{\gamma\beta}-u\leftrightarrow v]
\;:F_\alpha\; G_\beta:
\ea
We can now take $M_2\to\infty$ and carry out the sums over $\gamma$ 
\ba \label{3.86}
&& \sum_\gamma
\{[u_{\alpha\gamma} v_{\gamma\beta} - u\leftrightarrow v\}
\nonumber\\
&=& \sum_K\; \{\; 
[<e_I, u_p [e_K]_{,p}> \delta_{mr}+<e_I e_K, u_{r,m}>]\;\delta_{ik}\;
[<e_K, v_q [e_J]_{,q}> \delta_{rn}+<e_K e_J, v_{n,r}>]\;\delta_{kj}
- u\leftrightarrow v\}
\nonumber\\
&=& \sum_K\; \{\;
[-<[e_I u_p]_{,p},e_K> \delta_{mr}+<e_I u_{r,m},e_K>]
[<e_K, v_q [e_J]_{,q}> \delta_{rn}+<e_K, e_J v_{n,r}>]\;\delta_{ij}
- u\leftrightarrow v\}
\nonumber\\
&=& 
\{(-[e_I u_p]_{,p} \delta_{mr}+e_I u_{r,m}),\;
(v_q [e_J]_{,q}\delta_{rn}+e_J v_{n,r})>\;\delta_{ij} - u\leftrightarrow v
\}
\nonumber\\
&=& 
\delta_{ij}\{
<e_I,[u,v]_p e_{J,p}>\; \delta_{mn}
+<e_I e_J, [u_{r,m} v_{n,r}-v_{r,m} u_{n,r}]>
+<e_I,[u_r(v_{n,m} e_J)_{,r}-v_r(u_{n,m} e_J)_{,r}]>
\nonumber\\
&&
+<e_I,[u_{n,m} v_r e_{J,r}-v_{n,m} u_r e_{J,r}]>
\}
\nonumber\\
&=&
\{<e_I, [u,v]_p e_{J,p}> \; \delta_{mn}+<e_I e_J, [u,v]_{n,m}>\}\delta^{ij}
\ea
It follows that the result of the calculation is 
\be \label{3.88}
i\; \frac{i}{2}
\sum_{\alpha\le M_0,\beta\le M_0}\; (-[u,v])_{\alpha\beta}
\;:F_\alpha\; G_\beta:
\ee
Thus with e.g. the limiting pattern $M_1=M_3$ before $M_2=M_4$ we have
\be \label{3.89}
[C(u),C(v)]=\; i\; C(-[u,v])
\ee
~\\
\\
{\bf 3. Diffeomorphism-Gauss}\\
\\
We have to consider two curly brackets.\\
\\
{\bf 3A.}\\
The first curly bracket is, see (\ref{a.2})
\ba \label{3.90}
&&\frac{1}{2\sqrt{2}}\sum_{\alpha\le M_1,\beta\le M_2, \gamma\le M_3}\;
u_{\alpha\beta}\; r_\gamma\; 
[:F_\alpha \; G_\beta:, F_\gamma]
\nonumber\\
&=&-\frac{1}{\sqrt{2}}\sum_{\alpha\le M_1,\beta\le M_2, \gamma\le M_3}\;
u_{\alpha\beta}\; r_\gamma\; \delta_{\beta\gamma}\; 
F_\alpha
\ea
By sandwiching, $\alpha$ gets locked and we may pick e.g. the pattern 
$M_3$ before $M_1,M_2$. Then we solve the Kronecker for $\gamma$ and 
(\ref{3.90}) becomes 
\be \label{3.91}
-\frac{1}{\sqrt{2}}\sum_{\alpha\le M_0,\beta\le M_2}\;
u_{\alpha\beta}\; r_\beta\; F_\alpha
\ee
The cut-off $M_2$ can now be removed and we find with $\alpha=(Imi),\;
\beta=(Jjn)$ 
\ba \label{3.92}
\sum_\beta \;  
&&u_{\alpha\beta}\; r_\beta
=\sum_J [-<[u_p e_I]_{,p},e_J>\delta_{mn}+<e_I u_{n,m},e_J>]\;\delta_{ij}
<e_J,r^j_{,n}>
\nonumber\\
&=& <e_I,u_p r^i_{,m,p}+u_{p,m} r^i_{,p}>=<e_I,(u_p r^i_{,p})_{,m}>
=(u[r])_\alpha
\ea
It follows that (\ref{3.91}) becomes 
\be \label{3.92a}
i [-\frac{i}{\sqrt{2}}\sum_{\alpha\le M_0}\;
(-u[r])_\alpha\; F_\alpha]
\ee
Thus for e.g. the limiting pattern $M_3$ before $M_1,M_2$
we find 
\be \label{3.93}
[C_2(u),G_1(r)]=i G_1(-u[r])
\ee
~\\
{\bf 3B.}\\ 
The second curly bracket is given by, see (\ref{a.5}) 
\ba \label{3.94}
&& -\frac{1}{4}\sum_{\alpha\le M_1,\beta\le M_2,\gamma\le M_3, \delta\le M_4}
\; u_{\alpha\beta}\; r_{\delta\gamma}\;
[:F_\alpha G_\beta:,:F_\gamma G_\delta:]
\nonumber\\
&=& -\frac{1}{2}\sum_{\alpha\le M_1,\beta\le M_2,\gamma\le M_3, \delta\le M_4}
\; u_{\alpha\beta}\; r_{\delta\gamma}\;
[\delta_{\alpha\delta}\;:F_\gamma\; G_\beta: -
\delta_{\gamma\beta}\;:F_\alpha\; G_\delta:]
\ea
Picking again the pattern with $M_1,M_3$ before $M_2,M_4$ and sandwiching 
between Fock states this becomes
\ba \label{3.98}
&& -\frac{1}{2}\sum_{\beta\le M_0,\gamma\le M_0, \delta\le M_4}
\; u_{\delta\beta}\; r_{\delta\gamma}\;
:F_\gamma\; G_\beta: 
\nonumber\\
&& +\frac{1}{2}\sum_{\alpha\le M_0,\beta\le M_2, \delta\le M_0}
\; u_{\alpha\beta}\; r_{\delta\beta}\;
:F_\alpha\; G_\delta:
\ea
We now take $M_2=M_4\to \infty$ and obtain 
\ba \label{3.99} 
&& -\frac{1}{2}\sum_{\alpha\le M_0,\beta\le M_0, \gamma}
\; u_{\gamma\alpha}\; r_{\gamma\beta}\;
:G_\alpha\; F_\beta: 
\nonumber\\
&& +\frac{1}{2}\sum_{\alpha\le M_0,\beta\le M_0,\gamma}
\; u_{\beta\gamma}\; r_{\alpha\gamma}\;
:G_\alpha\; F_\beta:
\nonumber\\
&=& -\frac{1}{2}\sum_{\alpha\le M_0,\beta\le M_0, \gamma}
\; [u_{\gamma\alpha}\; r_{\gamma\beta}\;
-u_{\beta\gamma}\; r_{\alpha\gamma}]
:G_\alpha\; F_\beta: 
\ea
We have with $\alpha=(Imi), \beta=(Jnj), \gamma=(Kqk)$
\ba \label{3.99a}
&&\sum_\gamma \;
[u_{\gamma\alpha}\; r_{\gamma\beta}\;
-u_{\beta\gamma}\; r_{\alpha\gamma}]
\nonumber\\
&=& \sum_K\; \{
[<e_K, u_p e_{I,p}> \delta_{qm}+<e_K e_I, u_{m,q}>]\;\delta_{ki}\;
\epsilon_{lkj}\;<e_K e_J,r_l> \delta_{qn}
\nonumber\\
&& -[<e_J, u_p e_{K,p}> \delta_{nq}+<e_J e_K, u_{q,n}>]\delta_{jk}
\epsilon_{lik} <e_I e_K, r_l> \delta_{mq}
\}
\nonumber\\
&=& \sum_K\; \{
[<e_K, u_p e_{I,p}> \delta_{nm}+<e_K, e_I u_{m,n}>]\;
\epsilon_{lij}\;<e_J r_l, e_K>
\nonumber\\
&& -[-<(e_J u_p)_{,p}, e_K> \delta_{nm}+<e_J u_{m,n},e_K>]\;
\epsilon_{lij} <e_K, e_I r_l> 
\}
\nonumber\\
&=& \epsilon_{lij} \{
[<e_J r_l, u_p e_{I,p}> \delta_{nm}+<e_J r_l, e_I u_{m,n}>]
-[-<(e_J u_p)_{,p}, e_I r_l> \delta_{nm}+<e_J u_{m,n},e_I r_l>]
\}
\nonumber\\
&=& -\epsilon_{lij} <e_I e_J, u_p r_{l,p}>\; \delta_{mn}
=(u[r])_{\alpha\beta}
\ea
It follows the end result
\be \label{3.100}
i \frac{i}{2}\sum_{\alpha\le M_0,\beta\le M_0}\; (-u[r])_{\alpha\beta}
:G_\alpha\; F_\beta: 
\ee
Therefore for e.g. the pattern $M_1,M_3$ before $M_2=M_4$ 
\be \label{3.101}
[C_2(u),G_2(r)]=i G_2(-u[r])
\ee
and altogether
\be \label{3.102}
[C(u),G(r)]=i G(-u[r])
\ee
~\\
{\bf 4. Gauss-Hamiltonian}\\
\\
We have to compute four curly bracket terms.\\
\\
{\bf 4A}.\\
The first curly bracket is, see (\ref{a.3}) 
\ba \label{3.103}
%-i/2^{1/2} (-2^{-1/2})
&&\frac{i}{2}\sum_{\alpha\le M_1,\beta\le M_2,\gamma\le M_3,\delta\le M_4}
r_\alpha \; f_{\beta\gamma\delta}
[F_\alpha,:F_\beta F_\gamma G_\delta:]
\nonumber\\
&=& i\;\sum_{\alpha\le M_1,\beta\le M_2,\gamma\le M_3,\delta\le M_4}\;
r_\alpha \; f_{\beta\gamma\delta}\; \delta_{\alpha\delta}\;
:F_\beta F_\gamma:
\ea
By sandwiching, $\beta,\gamma$ get locked. We choose e.g. $M_1$ before 
$M_2,M_3,M_4$ and solve the Kronecker for $\alpha$. Then 
(\ref{3.103}) becomes
\be \label{3.107}
i\;\sum_{\beta\le M_0,\gamma\le M_0,\delta\le M_4}\;
r_\delta \; f_{\beta\gamma\delta}\; 
:F_\beta F_\gamma:
\ee
The sum over $\delta$ can now be performed unconstrained and we have 
for $\beta=(Imi), \gamma=(Jnj), \delta=(Kpk)$
\ba \label{3.108}
\sum_\gamma
&& r_\delta \; f_{\beta\gamma\delta} 
=\sum_K <e_K, r_{k,p}> <[e_I e_J f]_{,n}, e_K> \epsilon_{ijk}\delta_{mp}
\nonumber\\
&=& 
<r_{k,m},[e_I e_J f]_{,n}> \epsilon_{ijk}
=
-<r_k,[e_I e_J f]_{,n,m}> \epsilon_{ijk}
=:[f_r]_{\beta\gamma}
\ea
Since $[f_r]_{\beta\gamma}$ is antisymmetric while $:F_\beta \; F_\gamma:$ 
is symmetric under exchange of $\gamma,\beta$ and since 
we sum both $\beta,\gamma$ over full range 
in $S_{M_0}$, expression (\ref{3.107}) vanishes identically. 
It follows for e.g. $M_1$ before $M_2=M_3,M_4$
\be \label{3.109}
[G_1(r), C_{0;3}(f)]=0
\ee
~\\
{\bf 4B.}\\
The second curly bracket is, see (\ref{a.4}) 
\ba \label{3.110}
%-i/2^{1/2} (-1/4)
&& \frac{i}{4\sqrt{2}} 
\sum_{\alpha\le M_1,\beta\le M_2,\gamma\le M_3,\delta\le M_4,\mu\le M_5}
r_\mu\;f_{\alpha\beta\gamma\delta}\; 
[F_\mu, :F_\alpha F_\beta G_\gamma G_\delta:]
\nonumber\\
&=& \frac{i}{\sqrt{2}} 
\sum_{\alpha\le M_1,\beta\le M_2,\gamma\le M_3,\delta\le M_4,\mu\le M_5}
r_\mu\;f_{\alpha\beta\gamma\delta}\; 
:F_\alpha F_\beta G_{(\gamma}: \delta_{\delta)\mu}
\ea
By sandwiching, $\alpha,\beta$ and $\gamma$ or $\delta$ get locked.
We pick the pattern $M_5$ before $M_1,M_2,M_3=M_4$ and (\ref{3.110}) becomes
after solving the Kronecker for $\mu$
\be \label{3.113}
\frac{i}{2\sqrt{2}} 
\sum_{\alpha\le M_0,\beta\le M_0,\gamma\le M_0,\delta\le M_4}
r_\delta\;f_{\alpha\beta\gamma\delta}\; 
\; :F_\alpha F_\beta G_\gamma: 
+\frac{i}{2\sqrt{2}} 
\sum_{\alpha\le M_0,\beta\le M_0,\gamma\le M_3,\delta\le M_0}
r_\gamma\;f_{\alpha\beta\gamma\delta}\; 
\; :F_\alpha F_\beta G_\delta: 
\ee
Taking the remaining cut-off dependence to $\infty$ we find 
\be \label{3.114}
\frac{i}{2\sqrt{2}} 
\sum_{\alpha\le M_0,\beta\le M_0,\gamma\le M_0,\delta}
r_\delta\;[f_{\alpha\beta\gamma\delta}+f_{\alpha\beta\delta\gamma}] 
\; :F_\alpha F_\beta G_\gamma: 
\ee
and with $\alpha=(Imi), \; \beta=(Jnj),\;\gamma=(Kpk), \delta=(Lql)$
\ba \label{3.115}
&& \sum_\delta 
r_\delta\;[f_{\alpha\beta\gamma\delta}+f_{\alpha\beta\delta\gamma}] 
=\sum_L \; 
<e_L,r_{l,q}>\;
[<e_I e_J e_K e_L, f> \delta_{mp}\delta_{nq}\epsilon_{tij}\epsilon_{tkl}
+<e_I e_J e_L e_K, f> \delta_{mq}\delta_{np}\epsilon_{tij}\epsilon_{tlk}]
\nonumber\\
&=&
[<e_I e_J e_K, f r_{l,n}> \delta_{mp} 
-<e_I e_J e_K, f r_{l,m}> \delta_{np}]\epsilon_{tij}\epsilon_{tkl}
=: [f^r]_{\alpha\beta\gamma}
\ea
Accordingly with e.g. the pattern $M_5$ before $M_1,M_2, M_3=M_4$
\be \label{3.116}
[G_1(r),C_{0;4}(f)]=
\frac{i}{2\sqrt{2}} 
\sum_{\alpha,\beta,\gamma}\; [f^r]_{\alpha\beta\gamma}
\; :F_\alpha F_\beta G_\gamma: 
\ee
~\\
{\bf 4C.}\\
The third curly bracket is, see (\ref{a.6})
\ba \label{3.117}
%(i/2)(-1/(2^{1/2})
&& -\frac{i}{2\;2^{1/2}}\;
\sum_{\alpha\le M_1,\beta\le M_2,\gamma\le M_3,\mu\le M_4,\rho\le M_5}
r_{\rho\mu}\;f_{\alpha\beta\gamma}\;
[:F_\mu\; G_\rho:, :F_\alpha\; F_\beta\; G_\gamma:]
\\
&=& 
-\frac{i}{2\;2^{1/2}}\;
\sum_{\alpha\le M_1,\beta\le M_2,\gamma\le M_3,\mu\le M_4,\rho\le M_5}
r_{\rho\mu}\;f_{\alpha\beta\gamma}\;
(
2\delta_{\mu\gamma} \; :F_\alpha F_\beta G_\rho:
-4\delta_{\rho(\alpha}\; :F_{\beta)} F_\mu G_\gamma:
+4\delta_{\rho(\alpha} \delta_{\beta)\mu} G_\gamma
)
\nonumber
\ea
~\\
{\bf 4C.i}\\
This is the first incident that a potential divergence can appear because 
there is a term involving the product of two Kronecker symbols which we 
treat first. 
We assume e.g. a pattern with $M_4\ge M_1,M_2\ge M_3,M_5$, then 
in the term with two Kronecker symbols
the index $\gamma$ is locked by sandwiching and we can solve 
the Kroneckers for $\mu$
and $\alpha$ or $\beta$. The result, using 
$\alpha=(Imi), \; \beta=(Jnj),\; \gamma=(Kpk),\; \mu=(Mur), \rho=(Nvs)$, 
is proportional to 
\ba \label{3.120}
&& \sum_{\rho\le M_5, \gamma\le M_0}
[\sum_{\beta\le M_2}\; r_{\rho\beta}\; f_{\rho\beta\gamma} 
+\sum_{\alpha\le M_1}\; r_{\rho\alpha}\; f_{\alpha\rho\gamma}]
G_\gamma
\nonumber\\
&=&
\sum_{N\le M_5, K\le M_0}
[\sum_{J\le M_2}\; r_{NJt}\; \delta_{vn}\; 
\epsilon_{tsj}\;f_{NJnK}\; \delta_{vp}\;
\epsilon_{sjk} 
+\sum_{I\le M_1}\; r_{NIt}\; \delta_{vm}\epsilon_{tsi}\;
f_{INvK}\; \delta_{mp}\;\epsilon_{isk}]\; G_{Kpk}
\nonumber\\
&=&
2\;\sum_{N\le M_5, K\le M_0}
[\sum_{J\le M_2} r_{NJk}\; f_{NJpK} 
-\sum_{I\le M_1} r_{NIk}\; f_{INpK}] \; G_{Kpk}
\nonumber\\
&=& 0
\ea
if $M_1=M_2$ since $f_{INpK}=<(e_I e_N f)_{,p}, e_K>=f_{NIpK}$. Thus 
we will adopt $M_1=M_2$ from now on. Note that this result crucially 
rests on properties of $r_{\alpha\beta}, f_{\alpha\beta\gamma}$.\\
\\
{\bf 4C.ii}\\ 
The contribution to (\ref{3.117}) from the single Kronecker symbol terms is
\be \label{3.121}
-\frac{i}{2^{1/2}}\;
\sum_{\alpha\le M_1,\beta\le M_2,\gamma\le M_3,\mu\le M_4,\rho\le M_5}
r_{\rho\mu}\;f_{\alpha\beta\gamma}\;
[\delta_{\mu\gamma} \; :F_\alpha F_\beta G_\rho:
-2\delta_{\rho(\alpha}\; :F_{\beta)} F_\mu G_\gamma:]
\ee
By sandwiching and solving for $\mu$ in the first term and 
$\alpha$ or $\beta$ in the second respectively according to 
above pattern this becomes 
\ba \label{3.122}
&& 
-\frac{i}{2^{1/2}}\;\{
\sum_{\alpha,\beta,\rho\le M_0,\gamma\le M_3}\;
r_{\rho\gamma}\;f_{\alpha\beta\gamma}\; :F_\alpha F_\beta G_\rho:
-\sum_{\beta,\gamma,\mu\le M_0,\rho\le M_5}\;
r_{\rho\mu}\;f_{\rho\beta\gamma}\; :F_\beta F_\mu G_\gamma:
\nonumber\\ &&
-\sum_{\alpha,\gamma,\mu\le M_0,\rho\le M_5}\;
r_{\rho\mu}\;f_{\alpha\rho\gamma}\;
:F_\alpha F_\mu G_\gamma:
\}
\nonumber\\
&=& 
-\frac{i}{2^{1/2}}\;\{
\sum_{\alpha,\beta,\gamma\le M_0,\rho\le M_3}\;
r_{\gamma\rho}\;f_{\alpha\beta\rho}\; :F_\alpha F_\beta G_\gamma:
-\sum_{\alpha,\beta,\gamma\le M_0,\rho\le M_5}\;
r_{\rho\beta}\;f_{\rho\alpha\gamma}\; :F_\alpha F_\beta G_\gamma:
\nonumber\\ &&
-\sum_{\alpha,\gamma,\beta\le M_0,\rho\le M_5}\;
r_{\rho\beta}\;f_{\alpha\rho\gamma}\;
:F_\alpha F_\beta G_\gamma:
\}
\nonumber\\
&=& 
-\frac{i}{2^{1/2}}\;\{
\sum_{\alpha,\beta,\gamma\le M_0,\rho\le M_3}\;
r_{\gamma\rho}\;f_{\alpha\beta\rho}\;
-\sum_{\alpha,\beta,\gamma\le M_0,\rho\le M_5}\;
r_{\rho\beta}\;[f_{\rho\alpha\gamma}+f_{\alpha\rho\gamma}]\; 
\}\;:F_\alpha F_\beta G_\gamma:
\ea
We now take the remaining cut-offs on $\rho$ 
away and carry out the unconstrained sum
over $\rho=(Lql)$. This gives for the first term in 
the curly bracket of (\ref{3.122}) 
\ba \label{3.123}
\sum_\rho
r_{\gamma\rho}\;f_{\alpha\beta\rho}
&=& \sum_L\; r^t_{KL}\; \delta_{pq}\epsilon_{tkl}\; f_{IJnL} \;\delta_{mq}\;
\epsilon_{ijl} 
\nonumber\\
&=& \sum_L \; <r_t e_K,e_L>\; <e_L,(e_I e_J f)_{,n}> \delta_{mp}
\epsilon_{tkl}\epsilon_{ijl}
\nonumber\\
&=& <r_t e_K,(e_I e_J f)_{,n}> \delta_{mp}
\epsilon_{tkl}\epsilon_{ijl}
\nonumber\\
\sum_\rho
r_{\rho\beta}\;[f_{\rho\alpha\gamma}+f_{\alpha\rho\gamma}]\; 
&=& \sum_L\; 
r^t_{LJ} \epsilon_{tlj} \delta_{qn}\;[f_{LImK} \delta_{qp} 
\epsilon_{lik}+f_{ILqK}\delta_{mp}\epsilon_{ilk}]
\nonumber\\ 
&=& \sum_L\; 
r^t_{LJ} \epsilon_{tlj}\;[f_{LImK} \delta_{np}-f_{ILnK}\delta_{mp}]
\epsilon_{lik}
\nonumber\\ 
&=& -\sum_L\; 
<r_t e_J,e_L>\; \epsilon_{tlj}\;[<e_L, f e_I e_{K,m}>
\delta_{np}-<e_L, f e_I e_{K,n}> \delta_{mp}]
\epsilon_{lik}
\nonumber\\ 
&=& -\; \epsilon_{tlj}\;\epsilon_{lik}
[<r_t e_J f e_I, e_{K,m}>\; \delta_{np}
-<r_t e_J f e_I, e_{K,n}>\; \delta_{mp}]
\ea 
We use that $:F_\alpha F_\beta G_\gamma:=:F_\beta F_\alpha G_\gamma:$ and 
replace in (\ref{3.123}) both equalities by half the sum of 
(\ref{3.123}) and (\ref{3.123}) 
with $\alpha,\beta$ interchanged which results in 
\ba \label{3.124}
1. &=& \frac{1}{2}[
<r_t e_K,(e_I e_J f)_{,n}> \delta_{mp}
-<r_t e_K,(e_I e_J f)_{,m}> \delta_{np}]
\epsilon_{tkl}\epsilon_{ijl}
\nonumber\\
2. &=& \frac{1}{2}[
\epsilon_{tlj}\;\epsilon_{lik}
-\epsilon_{tli}\;\epsilon_{ljk}]\;
[<r_t e_J f e_I, e_{K,n}>\; \delta_{mp}
-<r_t e_J f e_I, e_{K,m}>\; \delta_{np}]
\nonumber\\
&=&
\frac{1}{2} \epsilon_{ijs} \epsilon_{uvs}\;\epsilon_{luk}\epsilon_{lvt} 
[<r_t e_J f e_I, e_{K,n}>\; \delta_{mp}
-<r_t e_J f e_I, e_{K,m}>\; \delta_{np}]
\nonumber\\
&=&
-\frac{1}{2} \epsilon_{ijs} \epsilon_{tks}
[<r_t e_J f e_I, e_{K,n}>\; \delta_{mp}
-<r_t e_J f e_I, e_{K,m}>\; \delta_{np}]
\ea
for the first and second equality respectively. The first minus the second 
of those two expressions that is required in (\ref{3.122})
is given by 
\ba \label{3.125}
&& \frac{1}{2}\; \epsilon_{tkl}\epsilon_{ijl}\{
<r_t e_K,(e_I e_J f)_{,n}> \delta_{mp}
-<r_t e_K,(e_I e_J f)_{,m}> \delta_{np}
+<r_t e_J f e_I, e_{K,n}>\; \delta_{mp}
-<r_t e_J f e_I, e_{K,m}>\; \delta_{np}
\}
\nonumber\\
&=& -\frac{1}{2}\; \epsilon_{tkl}\epsilon_{ijl}\{
<r_{t,n}, e_I e_J e_K f > \delta_{mp}
-<r_{t,m}, e_I e_J e_K f > \delta_{np}
\}
\nonumber\\
&=& \frac{1}{2} [f^r]_{\alpha\beta\gamma}
\ea
where the expression (\ref{3.115}) reappeared. 

It follows e.g. for the pattern $M_4$ before $M_1=M_2$ before 
$M_3,M_5$
\be \label{3.126}
[G_2(r),C_{0;3}(f)]=
-\frac{i}{2\sqrt{2}}\; 
\sum_{\alpha\beta\gamma} \; [f^r]_{\alpha\beta\gamma}\;
:F_\alpha \; F_\beta\; G_\gamma:
\ee
which is excaltly the negative of (\ref{3.116}).\\
\\
{\bf 4D.}\\
The fourth curly bracket is, see (\ref{a.7}) 
\ba \label{3.127}
%i/2 (-1/4)$
&&-\frac{i}{8} \sum_{
\alpha\le M_1\beta\le M_2,\gamma\le M_3,\delta\le M_4,
\mu\le M_5,\rho\le M_6}\; 
r_{\rho\mu}\; f_{\alpha\beta\gamma\delta} \;
[:F_\mu G_\rho:,:F_\alpha F_\beta G_\gamma G_\delta:]
\nonumber\\
&& [:F_\mu G_\rho:,:F_\alpha F_\beta G_\gamma G\delta:]
=
4\;:F_\alpha F_\beta G_\rho G_{(\gamma}:\;\delta_{\delta)\mu}
+4\delta_{\mu(\gamma}\;\delta_{\delta)\rho}\; :F_\alpha F_\beta:
\nonumber\\
&& -4\;\delta_{\rho(\alpha}\;:F_{\beta)} F_\mu G_\gamma G_\delta:
+4\delta_{\rho(\alpha}\;\delta_{\beta)\mu}\; :G_\gamma G_\delta:
\ea
\\
{\bf 4D.i}\\
Consider the contributions from two Kronecker symbols. Consider
the pattern 
$M_5\ge M_1,M_2\ge M_6\ge M_3,M_4$. 
For the term $\propto :F_\alpha F_\beta:$
we obtain using the sandwiching argumment $\alpha,\beta\le M_0$ and after 
solving for $\mu,\rho$ a 
coefficient proportional to
\be \label{3.129}  
\sum_{\gamma\le M_3,\delta\le M_4}\; [r_{\gamma\delta}+r_{\delta\gamma}]
f_{\alpha\beta\gamma\delta}=0
\ee
because $r_{\gamma\delta}=-r_{\delta\gamma}$. 
For the term $\propto :G_\gamma G_\delta:$
we obtain using the sandwiching argumment $\gamma,\delta\le M_0$ and after 
solving for $\mu$ and $\alpha$ or $\beta$ a   
coefficient proportional to
\be \label{3.130}
\sum_{\rho\le M_6}\;[
\sum_{\beta\le M_2} r_{\rho\beta} \;f_{\rho\beta\gamma\delta}
+\sum_{\alpha\le M_1} r_{\rho\alpha} \;f_{\alpha\rho\gamma\delta}]
\ee
Using $\alpha=(Imi),\beta=(Jnj),\gamma=(Kpk), \delta=(Lql),\mu=(Mur),
\rho=(Nvs)$ the term in square brackets becomes
\ba \label{3.131}
&&\sum_{J\le M_2}\; r^t_{NJ}\; \delta_{vn}\; \epsilon_{tsj}\; f_{NJKL}
\; \delta_{vp} 
\;\delta_{nq}\; \epsilon_{rsj} \;\epsilon_{rkl}
+\sum_{I\le M_1}\; r^t_{NI}\; \delta_{vm}\; \epsilon_{tsi}\; f_{INKL} 
\;\delta_{mp} \;\delta_{vq}\; \epsilon_{ris} \;\epsilon_{rkl}
\nonumber\\
&=& 2\epsilon_{tkl}\;\delta_{pq}\;
[\sum_{J\le M_2} r^t_{NJ} f_{JNKL}
-\sum_{I\le M_1} r^t_{NI} f_{INKL}]
\ea
which vanishes if $M_1=M_2$ due to total symmetry 
of $f_{IJKL}$. Thus we will adopt $M_1=M_2$ from now on.\\
\\
{\bf 4D.ii}\\
The single Kronecker symbol contributions to (\ref{3.127})
are given by 
\be \label{3.132}
-\frac{i}{4} \sum_{
\alpha\le M_1\beta\le M_2,\gamma\le M_3,\delta\le M_4,
\mu\le M_5,\rho\le M_6}\; 
r_{\rho\mu}\; f_{\alpha\beta\gamma\delta} \;
[2\;:F_\alpha F_\beta G_\rho G_{(\gamma}:\;\delta_{\delta)\mu}
-2\;\delta_{\rho(\alpha}\;:F_{\beta)} F_\mu G_\gamma G_\delta:]
\ee
We solve for $\mu$ and $\alpha$ or $\beta$ according to the chosen 
pattern and obtain with the sandwiching
argument and renaming indices carefully keeping track of the corresponding 
index range
\ba \label{3.133}
&&-\frac{i}{4} 
\sum_{\alpha,\beta,\rho\le M_0}\;
[\sum_{\gamma\le M_0,\delta\le M_4}\; f_{\alpha\beta\gamma\delta} \;
:F_\alpha F_\beta G_\rho G_\gamma:\;r_{\rho\delta}
+\sum_{\gamma\le M_3,\delta\le M_0}\; f_{\alpha\beta\gamma\delta} \;
:F_\alpha F_\beta G_\rho G_\delta:\;r_{\rho\gamma}]
\nonumber\\
&&+\frac{i}{4} 
\sum_{\mu,\gamma,\delta\le M_0,\rho\le M_6}\; 
r_{\rho\mu}\; 
[\sum_{\beta\le M_0}
f_{\rho\beta\gamma\delta}\;:F_\beta F_\mu G_\gamma G_\delta:
+\sum_{\alpha\le M_0}
f_{\alpha\rho\gamma\delta}\;:F_\alpha F_\mu G_\gamma G_\delta:]
\nonumber\\
&=&-\frac{i}{4} 
\sum_{\alpha,\beta,\gamma,\rho\le M_0}\;
[\sum_{\delta\le M_4}\; r_{\rho\delta}\; f_{\alpha\beta\gamma\delta} \;
+\sum_{\delta\le M_3}\; r_{\rho\delta}\; f_{\alpha\beta\delta\gamma}] \;
:F_\alpha F_\beta G_\gamma G_\rho:
\nonumber\\
&&+\frac{i}{4} 
\sum_{\beta,\mu,\gamma,\delta\le M_0;\rho\le M_6}\; 
r_{\rho\mu}\; 
[f_{\rho\beta\gamma\delta}\;
+f_{\beta\rho\gamma\delta}]
:F_\beta F_\mu G_\gamma G_\delta:
\nonumber\\
&=&-\frac{i}{4} 
\sum_{\alpha,\beta,\gamma,\delta\le M_0}\;
[\sum_{\rho\le M_4}\; r_{\delta\rho}\; f_{\alpha\beta\gamma\rho} \;
+\sum_{\rho\le M_3}\; r_{\delta\rho}\; f_{\alpha\beta\rho\gamma}] \;
:F_\alpha F_\beta G_\gamma G_\delta:
\nonumber\\
&&+\frac{i}{4} 
\sum_{\alpha,\beta,\gamma,\delta\le M_0;\rho\le M_6}\; 
r_{\rho\beta}\; 
[f_{\rho\alpha\gamma\delta}\;
+f_{\alpha\rho\gamma\delta}]
:F_\alpha F_\beta G_\gamma G_\delta:
\ea
We take the final cut-offs away and find 
\be \label{3.134}
-\frac{i}{4} 
\sum_{\alpha,\beta,\gamma,\delta\le M_0;\rho}\;
\{
r_{\delta\rho}\; [f_{\alpha\beta\gamma\rho}+f_{\alpha\beta\rho\gamma}]
-r_{\rho\beta}\; [f_{\rho\alpha\gamma\delta}+f_{\alpha\rho\gamma\delta}]\;
:F_\alpha F_\beta G_\gamma G_\delta:
\}
\ee
Note that $r_{\delta\rho}=-r_{\rho\delta}$ and 
$f_{\alpha\beta\gamma\delta}=
f_{\gamma\delta\alpha\beta}$. We have with
$\alpha=(Imi),\beta=(Jnj),\gamma=(Kpk),\delta=(Lql),\rho=(Mus)$
\ba \label{3.135}
&&  \sum_\rho \; r_{\rho\delta} \; f_{\alpha\beta\gamma\rho}
=\sum_M \;r^t_{ML}\; \delta_{uq}\; \epsilon_{tsl}\; f_{IJKM}\; \delta_{mp}\; 
\delta_{nu}\;\epsilon_{wij}\; \epsilon_{wks} 
\nonumber\\
&=& 
\epsilon_{tsl}\; <e_I e_J e_K e_L,f r^t>\; \delta_{mp}\; \delta_{nq}
\;\epsilon_{wij}\; \epsilon_{wks} 
=:[f_r]^t_{IJKL}\; \delta_{mp}\; \delta_{nq}\; 
\epsilon_{tsl}\; \epsilon_{wij}\; \epsilon{wks}
\ea
Thus 
\ba \label{3.136}
&&
\sum_\rho\; r_{\rho\delta}\; f_{\alpha\beta\rho\gamma} \;
=[f_r]^t_{IJKL}\; \delta_{mq}\delta_{np}\; 
\epsilon_{tsl}\epsilon_{wij}\epsilon_{wsk}
\nonumber\\
&& \sum_\rho r_{\rho\beta}\; f_{\rho\alpha\gamma\delta}\;
=[f_r]^t_{IJKL}\; \delta_{mq}\delta_{np}\; 
\epsilon_{tsj}\epsilon_{wsi}\epsilon_{wkl}
\nonumber\\
&& \sum_\rho r_{\rho\beta}\; f_{\alpha\rho\gamma\delta}\;
=[f_r]^t_{IJKL}\;  \delta_{mp}\delta_{nq}\; 
\epsilon_{tsj}\epsilon_{wis}\epsilon_{wkl}
\ea
Thus the coefficient of $\frac{i}{4}\;:F_\alpha F_\beta G_\gamma G_\delta:$
in (\ref{3.134}) 
is 
\be \label{3.137}
-[f_r]^t_{IJKL} \{
[\epsilon_{wij}\epsilon_{tsl}\epsilon_{wsk}
+\epsilon_{wkl}\epsilon_{tsj}\epsilon_{wsi}]\;
[\delta_{mp}\delta_{nq}-\delta_{np}\delta_{mq}]
\ee
As $:F_\alpha F_\beta G_\gamma G_\delta:$ is symmetric seperately under 
exchange of $\alpha,\beta$ and $\gamma,\delta$ we replace (\ref{3.137}) 
by half the sum of (\ref{3.137}) and (\ref{3.137}) with 
$\alpha\leftrightarrow\beta, \gamma\leftrightarrow\delta$.
We obtain
\ba \label{3.138}
&&-\frac{1}{2}[f_r]^t_{IJKL} \{
[\epsilon_{wij}
(\epsilon_{tsl}\epsilon_{wsk}-\epsilon_{tsk}\epsilon_{wsl})
+\epsilon_{wkl}
(\epsilon_{tsj}\epsilon_{wsi}-\epsilon_{tsi}\epsilon_{wsj})]\;
[\delta_{mp}\delta_{nq}-\delta_{np}\delta_{mq}]
\nonumber\\
&&-\frac{1}{2}[f_r]^t_{IJKL} \{
[\epsilon_{wij}\epsilon_{vkl}\epsilon_{vtw}
+\epsilon_{wkl}\epsilon_{vij}\epsilon_{vtw}]
\}
\nonumber\\
&&-\frac{1}{2}[f_r]^t_{IJKL}\;
\epsilon_{wij}\epsilon_{vkl}\;
[\epsilon_{vtw}+\epsilon_{wtv}=0
\ea
It follows e.g. with $M_5$ before $M_1=M_2$ before $M_6$ before 
$M_3=M_4$ 
\be \label{3.139}
[G_2(r),C_{0;4}(f)]=0
\ee
~\\
Altogether
\be \label{3.140}
[G(r),C_0(f)]=0
\ee
~\\
{\bf 5. Diffeomorphism - Hamiltonian}\\
\\
We have two curly brackets to consider:\\
\\
{\bf 5A.}\\
The first curly bracket is 
\be \label{3.141}
%i/2 (-1/2^{1/2})
-\frac{i}{2\;2^{1/2}} 
\sum_{\alpha\le M_1,\beta\le M_2,\gamma\le M_3,\mu\le M_4,\rho\le M_5}\;
u_{\mu\rho}\; f_{\alpha\beta\gamma}\;
[:F_\mu G_\rho:,:F_\alpha F_\beta G_\gamma:]
\ee
We can use the same results as in item 4C., in particular (\ref{3.117}).
We pick the same pattern $M_4\ge M_1,M_2\ge M_3,M_5$. \\
\\
{\bf 5A.i}\\
We consider first the contribution from two Kronecker symbols to (\ref{3.141})
which after solving for $\mu$ and 
$\alpha$ or $\beta$ is proportional to with 
$\alpha=(Imi),\beta=(Jnj),\gamma=(Kpk),\rho=(Lql)$
\be \label{3.142}
\sum_{\rho\le M_5}\;
[\sum_{\beta\le M_2} u_{\beta\rho} f_{\rho\beta\gamma}
+\sum_{\alpha\le M_1} u_{\alpha\rho} f_{\alpha\rho\gamma}]
=0
\ee
because $u_{\beta\rho}\propto \delta_{jl},\; f_{\rho\beta\gamma}\propto
\epsilon_{ljk}$ even without assuming $M_1=M_2$.\\
\\
{\bf 5A.ii}\\
The contribution from the single Kronecker symbol to (\ref{3.141})
is given by 
\be \label{3.143}
-\frac{i}{2^{1/2}} 
\sum_{\alpha\le M_1,\beta\le M_2,\gamma\le M_3,\mu M_4,\rho\le M_5}\;
u_{\mu\rho}\; f_{\alpha\beta\gamma}\;
[\delta_{\mu\gamma} \; :F_\alpha F_\beta G_\rho:
-2\delta_{\rho(\alpha}\; :F_{\beta)} F_\mu G_\gamma:]
\ee

We solve for $\mu$ or $\alpha,\beta$ and sandwich
\ba \label{3.143a}
&& -\frac{i}{2^{1/2}} 
\{
\sum_{\alpha,\beta,\rho\le M_0;\gamma\le M_3}\;
u_{\gamma\rho}\; f_{\alpha\beta\gamma}\; :F_\alpha F_\beta G_\rho:
-
\sum_{\beta,\mu,\gamma\le M_0,\rho\le M_5}\;
u_{\mu\rho}\; f_{\rho\beta\gamma}\;:F_\beta F_\mu G_\gamma:
\nonumber\\ &&
-
\sum_{\alpha,\mu,\gamma\le M_0;\rho\le M_5}\;
u_{\mu\rho}\; f_{\alpha\rho\gamma}\; :F_\alpha F_\mu G_\gamma:]
\}
\nonumber\\
&=&
-\frac{i}{2^{1/2}}\; \sum_{\alpha,\beta,\gamma\le M_0}
\;\{
\sum_{\rho\le M_3}\;
u_{\rho\gamma}\; f_{\alpha\beta\rho}\; 
-\sum_{\rho\le M_5}\;
(u_{\alpha\rho}\; f_{\rho\beta\gamma}
+u_{\beta\rho}\; f_{\alpha\rho\gamma}) 
\}
\;:F_\alpha F_\beta G_\gamma:
\ea
Taking the remaining cut-offs to infinity the sum over $\rho$ is 
unconstrained and we compute with  
$\alpha=(Imi),\beta=(Jnj),\gamma=(Kpk),\rho=(Lql)$
\ba \label{3.144}
&& \sum_\rho
\{u_{\rho\gamma}\; f_{\alpha\beta\rho} 
-u_{\alpha\rho}\; f_{\rho\beta\gamma}
-u_{\beta\rho}\; f_{\alpha\rho\gamma} 
\}
\nonumber\\
&=&\sum_L
\{
u^{qp}_{LK}\delta_{lk}\;f_{IJnL}\;\delta_{mq}\epsilon_{ijl}
-u^{mq}_{IL}\delta_{il}\;f_{LJnK}\;\delta_{qp}\epsilon_{ljk}
-u^{nq}_{JL}\delta_{jl}\;f_{ILqK}\;\delta_{mp}\epsilon_{ilk}
\}
\nonumber\\
&=&\sum_L
\{
u^{mp}_{LK}\;f_{IJnL}-u^{mp}_{IL}\;f_{LJnK}
-u^{nq}_{JL}\;f_{ILqK}\;\delta_{mp}\}
\;\epsilon_{ijk}
\ea
Consider the coefficient of $\delta_{mp}\epsilon_{ijk}$ from the first two 
terms together with the whole third term 
\ba \label{3.145}
&&\sum_L
\{
<e_L, u_r e_{K,r}>\; <(e_I e_J f)_{,n}, e_L>
-<e_I, u_r e_{L,r}>\; <(e_L e_J f)_{,n}, e_K>
\nonumber\\
&&-[<e_J, u_r e_{L,r}>\;\delta_{nq}+<e_J e_L, u_{q,n}>]\;<(e_I e_L f)_{,q},e_K>
\}
\nonumber\\
&=& \sum_L
\{
<e_L, u_r e_{K,r}>\; <(e_I e_J f)_{,n}, e_L>
-<(e_I u_r)_{,r}, e_L>\; <e_L, e_J f e_{K,n}>
\nonumber\\
&& -[<(e_J u_r)_{,r} e_L>\;\delta_{nq}-<e_L, e_J u_{q,n}>]\;<e_L, f e_I e_{K,q}>
\}
\nonumber\\
&=& 
<u_r e_{K,r},(e_I e_J f)_{,n}> -<(e_I u_r)_{,r}, e_J f e_{K,n}>
-<(e_J u_r)_{,r}, f e_I e_{K,n}>
+<e_J u_{r,n},f e_I e_{K,r}>
\nonumber\\
&=& 
<e_{K,r},u_r\;(e_I e_J f)_{,n}+(e_I e_J f)u_{r,n}>
-<e_{K,n}, f[ (e_I u_r)_{,r}, e_J +(e_J u_r)_{,r} e_I]>
\nonumber\\
&=& 
<e_{K,r},(u_r\;e_I e_J f)_{,n}>
-<e_{K,n}, f[ (e_I e_J u_r)_{,r} +u_{r,r} e_I e_J]>
\nonumber\\
&=& 
<e_{K,n},(u_r f_{,r}-u_{r,r} f)\;e_I e_J>
=-<(u[f] e_I e_J)_{,n} e_K>=-[u[f]]_{IJnK}
\ea
where we used 
$<(.)_{,m}\;,(.)'_{,n}>=<(.)_{,n}\;,(.)'_{,m}>$ in between.

The coefficient $\not\propto \delta_{mp}\epsilon_{ijk}$ of  
(\ref{3.144}) is 
\ba \label{3.146}
&& \sum_L \{
<e_L e_K, u_{p,m}>\;<(e_I e_J f)_{,n}, e_L>
-<e_I e_L, u_{p,m}>\;<(e_L e_J f)_{,n}, e_K>
\}
\nonumber\\
&=&\sum_L \{
<e_L, e_K u_{p,m}>\;<(e_I e_J f)_{,n}, e_L>
+<e_I e_L, u_{p,m}>\;<e_L, e_J f e_{K,n}>
\}
\nonumber\\
&=&
<e_K u_{p,m}, (e_I e_J f)_{,n}>
+<e_I, u_{p,m}\; e_J f e_{K,n}>
\nonumber\\
&=& -<e_I e_J e_K  f, u_{p,m,n}>
\ea
Together with the factor $\epsilon_{ijk}$ (\ref{3.146}) is a coefficient 
$[f_u]_{\alpha\beta\gamma}$ which is skew under exchange of $\alpha,\beta$
while $:F_\alpha F_\beta G_\gamma:$ is symmetric. It follows that the 
first curly bracket is 
\be \label{3.147}
i\;[-\frac{1}{2^{1/2}}\; \sum_{\alpha,\beta,\gamma\le M_0}
(-u[f])_{\alpha\beta\gamma}\;:F_\alpha F_\beta G_\gamma:]
\ee
Thus e.g. 
for the pattern $M_4$ before $M_1=M_2$ before $M_3,M_5$  
we get
\be \label{3.148}
[C_2(u),C_{0;3}(f)]=i C_{0;3}(-u[f])
\ee
~\\
{\bf 5B.}\\
The second curly bracket is given by
\be \label{3.149}
% (i/2) (-1/4)
-\frac{i}{8}
\sum_{\alpha\le M_1,\beta\le M_2,\gamma\le M_3,\delta\le M_4,
\mu\le M_5,\rho\le M_6}\;
u_{\mu\rho}\; f_{\alpha\beta\gamma\delta}\;
[:F_\mu G_\rho:,:F_\alpha F_\beta G_\gamma G_\delta:]
\ee
We can use the same results as in item 4D., in particular (\ref{3.127}).
We pick the same pattern $M_5\ge M_1,M_2\ge M_6\ge M_3,M_4$. \\
\\
{\bf 5B.i}\\
We consider first the contribution from two Kronecker symbols to (\ref{3.149})
%With 
%$\alpha=(Imi),\beta=(Jnj),\gamma=(Kpk), \delta=(Lql),\;\rho=(Mus)$
The coefficient of $:F_\alpha F_\beta:$ 
after solving for $\mu,\rho$ (which is consistent with the 
pattern as $\alpha,\beta$ get locked when sandwiching) is proportional to 
\be \label{3.150}
-2\sum_{\gamma\delta} \;[u_{\gamma\delta} +u_{\delta\gamma}] 
f_{\alpha\beta\gamma\delta}
\ee
which vanishes as $u_{\gamma\delta},u_{\delta\gamma}\propto \delta_{kl}$ while
$f_{\alpha\beta\gamma\delta}\propto \epsilon_{rkl}$. 
The coefficient of $:G_\gamma G_\delta:$ 
after solving for $\mu$ and $\alpha$ or $\beta$ is proportional to 
\be \label{3.151}
2\sum_\rho
[\sum_\beta \;u_{\beta\rho}\; f_{\rho\beta\gamma\delta}
+\sum_\alpha \; u_{\alpha\rho}\; f_{\alpha\rho\gamma\delta}]
\ee
which vanishes as $u_{\alpha\rho},\propto \delta_{il}$ while
$f_{\alpha\rho\gamma\delta},f_{\rho\alpha\gamma\delta}
\propto \epsilon_{ril}$. \\
\\
{\bf 5B.ii}\\
The single Kronecker symbol contribution to (\ref{3.149}) is given 
by 
\be \label{3.152}
-\frac{i}{4}
\sum_{\alpha\le M_1,\beta\le M_2,\gamma\le M_3,\delta\le M_4,
\mu\le M_5,\rho\le M_6}\;
u_{\mu\rho}\; f_{\alpha\beta\gamma\delta}\;
[
2\;:F_\alpha F_\beta G_\rho G_{(\gamma}:\;\delta_{\delta)\mu}
-2\;\delta_{\rho(\alpha}\;:F_{\beta)} F_\mu G_\gamma G_\delta:
]
\ee
We solve for $\mu$ and $\alpha$ or $\beta$ and get by sandwiching
\ba \label{3.153}
&& -\frac{i}{4}
\{
\sum_{\alpha,\beta,\rho,\gamma\le M_0;\delta\le M_4}\;
u_{\delta\rho}\; f_{\alpha\beta\gamma\delta}\;
:F_\alpha F_\beta G_\rho G_\gamma:
+
\sum_{\alpha,\beta\rho,\delta\le M_0;\gamma\le M_3}\;
u_{\gamma\rho}\; f_{\alpha\beta\gamma\delta}\;
:F_\alpha F_\beta G_\rho G_\delta:
\nonumber\\
&& -
\sum_{\beta,\mu,\gamma,\delta\le M_0;\rho\le M_6}\;
u_{\mu\rho}\; f_{\rho\beta\gamma\delta}
:F_\beta F_\mu G_\gamma G_\delta:
-
\sum_{\alpha,\mu,\gamma,\delta\le M_0;\rho\le M_6}\;
u_{\mu\rho}\; f_{\alpha\rho\gamma\delta}\;
:F_\alpha F_\mu G_\gamma G_\delta:
\}
\nonumber\\
&=& -\frac{i}{4}
\{
\sum_{\alpha,\beta,\gamma\delta\le M_0;\rho\le M_4}\;
u_{\rho\delta}\; f_{\alpha\beta\gamma\rho}\;
:F_\alpha F_\beta G_\gamma G_\delta:
+
\sum_{\alpha,\beta\gamma,\delta\le M_0;\rho\le M_3}\;
u_{\rho\gamma}\; f_{\alpha\beta\rho\delta}\;
:F_\alpha F_\beta G_\gamma G_\delta:
\nonumber\\
&& -
\sum_{\alpha,\beta,\gamma,\delta\le M_0;\rho\le M_6}\;
u_{\alpha\rho}\; f_{\rho\beta\gamma\delta}
:F_\alpha F_\beta G_\gamma G_\delta:
-
\sum_{\alpha,\beta,\gamma,\delta\le M_0;\rho\le M_6}\;
u_{\beta\rho}\; f_{\alpha\rho\gamma\delta}\;
:F_\alpha F_\beta G_\gamma G_\delta:
\}
\ea
Taking the remaining cut-offs to infinity the sum over $\rho$ becomes 
unconstrained and the coefficient  
of $-\frac{i}{4} :F_\alpha F_\beta G_\gamma G_\delta:$
in (\ref{3.153}) becomes
\ba \label{3.155}
&&\sum_\rho\;\{
u_{\rho\delta}\; f_{\alpha\beta\gamma\rho}\;
+u_{\rho\gamma}\; f_{\alpha\beta\rho\delta}\;
-u_{\alpha\rho}\; f_{\rho\beta\gamma\delta}
-u_{\beta\rho}\; f_{\alpha\rho\gamma\delta}\;
\}
\nonumber\\
&=& 
\sum_M\;\{
u^{uq}_{ML}\delta_{sl} f_{IJKM} \delta_{mp} \delta_{nu}
\epsilon_{tij} \epsilon_{tks} \;
+
u^{up}_{MK}\delta_{sk} f_{IJML} \delta_{mu} \delta_{nq}
\epsilon_{tij} \epsilon_{tsl} \;
\nonumber\\
&& -
u^{mu}_{IM}\delta_{is} f_{MJKL} \delta_{up} \delta_{nq}
\epsilon_{tsj} \epsilon_{tkl} \;
-
u^{nu}_{JM}\delta_{js} f_{IMKL} \delta_{pm} \delta_{uq} 
\epsilon_{tis} \epsilon_{tkl} \;
\}
\nonumber\\
&=& 
\epsilon_{tij} \epsilon_{tkl} \;
\sum_M\;\{
\delta_{mp}
[u^{nq}_{ML} f_{IJKM} - u^{nq}_{JM}  f_{IMKL}]
+\delta_{nq}
[u^{mp}_{MK}  f_{IJML} - u^{mp}_{IM}  f_{MJKL}]
\}
\ea
The contributions proportional to $\delta_{mp}\delta_{nq}$ to (\ref{3.155})
are 
\ba \label{3.156}
&& \sum_L \{
<e_M, u_r e_{L,r}>\;<e_I e_J e_K e_M,f>
-<e_J, u_r e_{M,r}>\;<e_I e_K e_L e_M,f>
\nonumber\\ &&
+<e_M, u_r e_{K,r}>\;<e_I e_J e_L e_M,f>
-<e_I, u_r e_{M,r}>\;<e_J e_K e_L e_M,f>
\}
\nonumber\\
&=&
<f u_r, e_{L,r} e_I e_J e_K>  
+<(e_J u_r)_{,r} e_I e_K e_L,f>
+<f u_r, e_{K,r} e_I e_J e_L>
+<(e_I u_r)_{,r} e_J e_K e_L,f>
\nonumber\\
&=&
<f, (e_I e_J e_K e_L u_r)_{,r}> + <f, e_I e_J e_K e_L u_{r,r}> +  
\nonumber\\
&=&
-<u[f],e_I e_J e_K e_L>=-(u[f])_{IJKL}
\ea
The contributions $\not\propto \delta_{mp}\delta_{nq}$ to (\ref{3.155})
are
\ba \label{3.157}
&& \sum_M\;\{
\delta_{mp}
[
<e_M e_L, u_{q,n}>\; <e_I e_J e_K e_M, f>
-
<e_J e_M, u_{q,n}>\; <e_I e_M e_K e_L, f>
]
\nonumber\\
&&+\delta_{nq}
[
<e_M e_K, u_{p,m}>\; <e_I e_J e_M e_L, f>
- 
<e_I e_M, u_{p,m}>\; <e_M e_J e_K e_L, f>
]
\}
\nonumber\\
&=& 
\delta_{mp}\;<f e_I e_J e_K e_L,[u_{q,n}-u_{q,n}]>
+\delta_{nq}\;<f e_I e_J e_K e_L,[u_{p,m}-u_{p,m}]>
=0
\ea
Since 
\be \label{3.158}
-(u[f])_{IJKL}\delta_{mp}\delta_{nq}
\epsilon_{tij} \epsilon_{tkl} \;=-(u[f])_{\alpha\beta\gamma\delta}
\ee
we find that 
(\ref{3.149}) reduces to 
\be \label{3.159}
i [-\frac{1}{4} 
\sum_{\alpha\beta\gamma\delta\le M_0}\; (-u[f])_{\alpha\beta\gamma\delta}
\;:F_\alpha F_\beta G_\gamma G_\delta:]
\ee
Thus, e.g. for the pattern $M_5$ before $M_1,M_2$ before $M_6$ before 
$M_3,M_4$
\be \label{3.160}
[C_2(u),C_{0;4}(f)]=i\; C_{0;4}(-u[f])
\ee
Altogether
\be \label{3.161}
[C(u),C_0(f)]=i\; C_0(-u[f])
\ee
~\\
{\bf 6. Hamiltonian - Hamiltonian}\\
\\
We have to compute three curly brackets:\\
\\
{\bf 6A.}\\
The first curly bracket is, see (\ref{a.8})
\ba \label{3.162}
% (1/2 from antisymmetrising times (-1/\sqrt{2})^2
&& \frac{1}{4}
\sum_{\alpha\le M_1,\beta\le M_2,\gamma\le M_3,
\mu\le M_4,\nu\le M_5,\rho\le M_6}\;
[f_{\alpha\beta\gamma}\;g_{\mu\nu\rho} -f\leftrightarrow g]\;
[:F_\alpha F_\beta G_\gamma:\;,\;:F_\mu F_\nu G_\rho:]
\\
&& [:F_\alpha F_\beta G_\gamma:\;,\;:F_\mu F_\nu G_\rho:]
\nonumber\\
&& 
=4\;
[\delta_{\rho(\alpha}\; :F_{\beta)}\; F_\mu\; F_\nu\; G_\gamma:
-\delta_{\gamma(\mu}\; :F_{\nu)}\; F_\alpha\; F_\beta\; G_\rho:]
+8\;
[\delta_{\gamma(\mu}\; \delta_{\nu)(\alpha} \; :F_{\beta)}\; G_\rho:
-\delta_{\rho(\alpha}\; \delta_{\beta)(\mu} \; :F_{\nu)}\; G_\gamma:]
\nonumber
\ea
~\\
{\bf 6A.i.}\\
We consider first the terms containing two Kronecker factors and 
pick the pattern $M_1,M_2,M_4,M_5$ before $M_3,M_6$. Using 
sandwiching, this allows 
to solve for $\mu,\nu$ at locked $\rho$ and $\alpha$ or 
$\beta$ and  
to solve for $\alpha,\beta$ at locked $\gamma$ and $\mu$ or 
$\nu$. Hence this term, without the $f\leftrightarrow g$ 
contribution is given by  
\ba \label{3.163}
&& \frac{1}{2}\;\sum_{\gamma\le M_3}
\{ 
\sum_{\beta,\rho\le M_0;\alpha\le M_1}\; 
f_{\alpha\beta\gamma}\; 
[g_{\gamma\alpha\rho}+g_{\alpha\gamma\rho}] :F_\beta\; G_\rho:
+
\sum_{\alpha,\rho\le M_0;\beta\le M_2}\; 
f_{\alpha\beta\gamma}\; 
[g_{\gamma\beta\rho}+g_{\beta\gamma\rho}] :F_\alpha\; G_\rho:
\}
\nonumber\\
&=& - \frac{1}{2}\;
\sum_{\rho\le M_6}
\{
\sum_{\nu,\gamma\le M_0;\mu\le M_4}\; 
g_{\mu\nu\rho}\; 
[f_{\rho\mu\gamma}+f_{\mu\rho\gamma}] :F_\nu\; G_\gamma:
+
\sum_{\mu,\gamma\le M_0;\nu\le M_5}\; 
g_{\mu\nu\rho}\; 
[f_{\rho\nu\gamma}+f_{\nu\rho\gamma}] :F_\mu\; G_\gamma:
\}
\nonumber\\
&=&
\frac{1}{2}\;\sum_{\gamma\le M_3;\beta,\rho\le M_0}
\{
[\sum_{\alpha\le M_1}\; f_{\alpha\beta\gamma}\; 
+\sum_{\alpha\le M_2}\; f_{\beta\alpha\gamma}]\; 
[g_{\gamma\alpha\rho}+g_{\alpha\gamma\rho}] :F_\beta\; G_\rho:
\}
\nonumber\\
&& -\frac{1}{2}\;\sum_{\rho\le M_6;\nu,\gamma\le M_0}
\{
[\sum_{\mu\le M_4}\; g_{\mu\nu\rho}\; 
+\sum_{\mu\le M_5}\; g_{\nu\mu\rho}]\; 
[f_{\rho\mu\gamma}+f_{\mu\rho\gamma}] :F_\nu\; G_\gamma:
\}
\ea
In contrast to 
previous considerations,
this does not vanish by inspection due to symmetry properties of 
the coefficients or when picking symmetric ranges 
$M_1=M_2=M_4=M_5$ and $M_3=M_6$ and not even 
when subtracting $f\leftrightarrow g$. To see that it vanishes 
nonetheless requires to perform the sums in the prescribed order.
We pick $M_1=M_2=M_4=M_5, M_3=M_6$ and relabel 
$\mu,\nu,\gamma$ by $\alpha,\beta,\rho$ in the second term to simplify
(\ref{3.162}) which becomes including the $f\leftrightarrow g$ term
\ba \label{3.164}
&&
\frac{1}{2}\;\sum_{\gamma\le M_3;\beta,\rho\le M_0,\alpha\le M_1}\;
:F_\beta\; G_\rho:\;
\left(
\{
[f_{\alpha\beta\gamma}\; +f_{\beta\alpha\gamma}]\; 
[g_{\gamma\alpha\rho}+g_{\alpha\gamma\rho}] 
-
[g_{\alpha\beta\gamma}+g_{\beta\alpha\gamma}]\; 
[f_{\gamma\alpha\rho}+f_{\alpha\gamma\rho}] 
\} -\; f\leftrightarrow g \right)
\nonumber\\
&=&
2\;\sum_{\gamma\le M_3;\beta,\rho\le M_0,\alpha\le M_1}\;
:F_\beta\; G_\rho:\;
\left(
[f_{\alpha\beta\gamma}\; +f_{\beta\alpha\gamma}]\; 
[g_{\gamma\alpha\rho}+g_{\alpha\gamma\rho}] 
\; -f \leftrightarrow g \right)
\ea
We remove the cut-off $M_1$, perform the sum over $\alpha$ and get with
$\alpha=(Imi),\;\beta=(Jnj), \gamma=(Kpk),\; \rho=(Lql)$ for the term
in the round bracket of (\ref{3.163}) without the $f\leftrightarrow g$ term
\ba \label{3.165}
&& 
\sum_I \; 
[
<(e_I e_J f)_{,n},e_K>\;\delta_{mp}\epsilon_{ijk}
-
<(e_J e_I f)_{,m},e_K>\;\delta_{np}\epsilon_{ijk}
]
\;
\times\nonumber\\
&& [
<(e_K e_I g)_{,m},e_L>\;\delta_{pq}\epsilon_{kil}
-
<(e_I e_K g)_{,p},e_L>\;\delta_{mq}\epsilon_{kil}
]
\nonumber\\
&=& 2\;\delta_{jl}\;\sum_I \; 
[
<e_I e_J f,e_{K,n}>\;\delta_{mp}
-
<e_J e_I f,e_{K,m}>\;\delta_{np}
]
\times \nonumber\\
&&
[
<e_K e_I g,e_{L,m}>\;\delta_{pq}
-
<e_I e_K g,e_{L,p}>\;\delta_{mq}
]
\nonumber\\
&=& 2\;\delta_{jl}\; \sum_I\;
[
<e_I e_J f,e_{K,n}>\;<e_K e_I g,e_{L,p}>\;\delta_{pq}
-<e_I e_J f,e_{K,n}>\;<e_I e_K g,e_{L,p}>\;\delta_{pq}
\nonumber\\
&& -<e_J e_I f,e_{K,m}>\;\delta_{np}\;<e_K e_I g,e_{L,m}>\;\delta_{pq}
+<e_J e_I f,e_{K,q}>\;\delta_{np}\;<e_I e_K g,e_{L,p}>
]
\nonumber\\
&=& -2\; \delta_{jl}\;
<e_J\;e_K\;[e_{K,m}\;e_{L,m}\;\delta_{pq}
-e_{K,q}\;e_{L,p}],\;fg>\; \delta_{np}
\ea
Since (\ref{3.165}) is symmetric under $f\leftrightarrow g$, 
(\ref{3.164}) vanishes.\\
\\
If instead we choose the pattern $M_3=M_6$ before $M_1=M_2=M_4=M_5$ 
then we must solve the two Kroneckers 
either for  $\gamma$ and e.g. one of $\alpha,\beta$ or for  
$\rho$ and e.g. one of $\mu,\nu$. With sandwiching and noticing that 
with the chosen synchronised ranges of the cut-offs in (\ref{3.163}) 
the first term
is minus the second term under exchange of $(\alpha,\beta,\gamma)$
with $(\mu,\nu,\rho)$ we find
\ba \label{3.171}
&& -
\sum_{\alpha,\beta,\mu,\nu\le M_1,\gamma,\rho\le M_3}\;
[f_{\alpha\beta\gamma}\;g_{\mu\nu\rho} -f\leftrightarrow g]\;
[
\delta_{\rho\alpha} \delta_{\beta\mu} :F_{\nu} G_\gamma:
+\delta_{\rho\beta} \delta_{\alpha\mu} :F_{\nu} G_\gamma:
+\delta_{\rho\alpha} \delta_{\beta\nu} :F_{\mu} G_\gamma:
+\delta_{\rho\beta} \delta_{\alpha\nu} :F_{\mu} G_\gamma:
] 
\nonumber\\
&=& - \{
\sum_{\nu\gamma\le M_0;\alpha,\beta\le M_1}\;
[f_{\alpha\beta\gamma}\;g_{\beta\nu\alpha} -f\leftrightarrow g]\;
:F_{\nu} G_\gamma:
+
\sum_{\nu\gamma\le M_0;\alpha,\beta\le M_1,\gamma}\;
[f_{\alpha\beta\gamma}\;g_{\alpha\nu\beta} -f\leftrightarrow g]\;
:F_{\nu} G_\gamma:
\nonumber\\ &&
+
\sum_{\mu\gamma\le M_0;\alpha,\beta\le M_1}\;
[f_{\alpha\beta\gamma}\;g_{\mu\beta\alpha} -f\leftrightarrow g]\;
:F_{\mu} G_\gamma:
+
\sum_{\mu,\alpha\le M_0;\alpha,\beta\le M_1}\;
[f_{\alpha\beta\gamma}\;g_{\mu\alpha\beta} -f\leftrightarrow g]\;
:F_{\mu} G_\gamma:
\} 
\nonumber\\
&=& -\sum_{\mu,\gamma\le M_0;\alpha,\beta\le M_1}
f_{\alpha\beta\gamma}\;
[g_{\beta\mu\alpha}+g_{\alpha\mu\beta}
+g_{\mu\beta\alpha}+g_{\mu\alpha\beta} - f \leftrightarrow g]
:F_{\mu} G_\gamma:
\ea
In this expression we no longer have the freedom to carry out one limit 
before the other. We have with $\alpha=(Imi),\beta=(Jnj),\gamma=(Kpk),
\mu=(Lql)$
\ba \label{3.172}
&& 
\sum_{\alpha,\beta\le M_1}\;
f_{\alpha\beta\gamma}\;
[g_{\beta\nu\alpha}+g_{\alpha\nu\beta}
+g_{\nu\beta\alpha}+g_{\nu\alpha\beta} - f \leftrightarrow g]
\nonumber\\
&=&
\sum_{I,J\le M_1}\;
f_{IJnK}\delta_{mp}\epsilon_{ijk}\;
[
g_{JLqI} \delta_{nm} \epsilon_{jli}
+
g_{ILqJ} \delta_{mn} \epsilon_{ilj}
+
g_{LJnI} \delta_{qm} \epsilon_{lji}
+
g_{LImJ} \delta_{qn} \epsilon_{lij}
- f \leftrightarrow g]
\nonumber\\
&=& 2\delta_{kl}
\sum_{I,J\le M_1}\;
f_{IJnK}\;
[
g_{JLqI} \delta_{np}
-
g_{ILqJ} \delta_{np}
-
g_{LJnI} \delta_{qp} 
+
g_{LIpJ} \delta_{qn} 
- f \leftrightarrow g]
\nonumber\\
&=& -2\delta_{kl}
\sum_{I,J\le M_1}\;
<e_I e_J f, e_{K,n}>
[-
<(e_J e_L g)_{,n},e_I>\delta_{qp} 
+<(e_I e_L g)_{,p},e_J>\delta_{qn} 
- f \leftrightarrow g]
\ea
where we used that $f_{IJnK}=f_{JInK}$ to interchange $I,J$ in the second term
at finite $M_1$ to see that it cancels against the first term
which again does not vanish at any finite $M_1$. Now we remove $M_1$ and 
can perform the sum over one of $I$ or $J$ using the completeness relation
and obtain after relabelling the remaining of $J,I$ into $I$
\ba \label{3.173}
&& -2\delta_{kl}
\sum_I\;
[-
<(e_I e_L g)_{,r}, f e_I e_{K,r}>\delta_{qp} 
+<(e_I e_L g)_{,p}, f e_I e_{K,q}> 
- f \leftrightarrow g]
\nonumber\\
&& 2\delta_{kl}\;
[
<\Lambda e_L, [g_{,r} f-f_{,r} g] e_{K,r}>\delta_{qp} 
-<\Lambda e_L [g_{,p}, f -f_{,p} g] e_{K,q}>]
\ea
where we introduced the quantity
\be \label{3.174}
\Lambda(x):=\lim_{M\to \infty} \Lambda_M(x),\;
\Lambda_M(x)=\sum_{I\le M} e_I(x)^2
\ee
As already mentioned, $\Lambda$ is formally given by $\delta(x,x)$ which   
is divergent. For both $\sigma=T^3,\;\mathbb{R}^3$ we have 
\be \label{3.175}
\Lambda_M=[\frac{\zeta_M(0)}{\pi}]^3,\;
\zeta_M(z):=\sum_{n=1}^M\; n^{-z}
\ee
where $\zeta_M(z)$ is the truncated Riemann zeta function \cite{30} whose 
limit converges for $\Re(z)>1$. Thus 
$\Lambda_M(x)$ is in this case 
independet of $x$. To regularise $\zeta_M(0)$ we 
may use the analytic extension of the zeta funcion which gives 
$\zeta(0)=-\frac{1}{2}$. Given any such regularisation of $\Lambda$
we would conclude that in the pattern $M_3=M_6$ before  
$M_1=M_2=M_4=M_5$ the contribution to the curly bracket from the double 
Kronecker symbol is given by 
\be \label{3.176}
-2\sum_{K,L\le M_0}\delta_{kl}\;
[<\Lambda e_L, \omega_r e_{K,r}>\delta_{qp} 
-<\Lambda e_L \omega_p e_{K,q}>]
:F_{Lql;} G_{Kpk}:,\;\; \omega_r:=f g_{,r}-g f_{,r} 
\ee
To interpret this result consider 
\ba \label{3.177}
&&\int\;d^3x\; \delta^{ab}\omega_a\; 
:(\partial_b A_c^k-\partial_c A_b^k) E^c_k:
=
\frac{i}{2}
\sum_{K,L} \delta^{rp}\;\delta_{kl} 
[<\omega_r, e_{K,p} e_L> :G_{Kqk} F_{Lql}: 
-<\omega_r, e_{K,q} e_L> :G_{Kpk} F_{Lql}:] 
\nonumber\\
&=&
\frac{i}{2}
\sum_{K,L} \delta^{rp}\;\delta_{kl} 
[<e_L, \omega_r e_{K,r}> \delta_{pg} -<e_L, \omega_p e_{K,q}>] 
:G_{Kpk} F_{Lql}:] 
\ea
It follows that (\ref{3.176}) is given by 
\be \label{3.178}
4i \int\; d^Dx\; [\Lambda \omega_a]\; \delta^{ab} \; 
:(\partial_b A_c^k-\partial_c A_b^k) E^c_k:
\ee
~\\
{\bf 6A.ii.}\\
We now focus on the single Kronecker symbol contribution to (\ref{3.162}).
In both schemes that we discuss for the first curly bracket we see that 
the second contribution to the single Kronecker term is the same 
as the first with $\alpha,\beta,\gamma$ interchanged with 
$\mu,\nu,\rho$ with a minus sign. Relabelling indices produces another minus
sign from the antisymmetric combination of the $f,g$ dependence. Thus 
we may drop the second contribution and instead multiply by a factor of 
two. Then $\mu,\nu,\gamma$ are locked by sandwiching and we solve 
for $\alpha$ or $\beta$ using the $M_1=M_2$ before $M_3$ pattern while 
we solve for $\rho$ in the $M_3$ before $M_1=M_2$ pattern. 

The first pattern gives
\ba \label{3.179}
&&
\sum_{\alpha,\mu,\beta,\nu\le M_1;\gamma,\rho\le M_3}
[f_{\alpha\beta\gamma}\;g_{\mu\nu\rho} -f\leftrightarrow g]\;
:F_\mu F_\nu [\delta_{\rho\alpha} F_{\beta} +\delta_{\rho\beta} F_{\alpha}] 
G_\gamma:
\nonumber\\
&=&
\sum_{\mu,\nu,\beta,\gamma\le M_0;\rho\le M_3}
[f_{\rho\beta\gamma}\;g_{\mu\nu\rho} -f\leftrightarrow g]\;
:F_\mu F_\nu F_{\beta} G_\gamma:
+
\sum_{\mu,\nu\alpha,\gamma\le M_0;\rho\le M_3}
[f_{\alpha\rho\gamma}\;g_{\mu\nu\rho} -f\leftrightarrow g]\;
:F_\mu F_\nu F_{\alpha} G_\gamma:
\nonumber\\
&=&
\sum_{\mu,\nu,\alpha,\gamma\le M_0;\rho\le M_3}
[(f_{\rho\alpha\gamma}+f_{\alpha\rho\gamma})\;g_{\mu\nu\rho} 
-f\leftrightarrow g]\;
:F_\mu F_\nu F_{\alpha} G_\gamma:
\ea
and the second 
\ba \label{3.179a}
&&
\sum_{\alpha,\mu,\beta,\nu\le M_1;\gamma,\rho\le M_3}
[f_{\alpha\beta\gamma}\;g_{\mu\nu\rho} -f\leftrightarrow g]\;
:F_\mu F_\nu [\delta_{\rho\alpha} F_{\beta} +\delta_{\rho\beta} F_{\alpha}] 
G_\gamma:
\nonumber\\
&=&
\sum_{\mu,\nu,\beta,\gamma\le M_0;\alpha\le M_1}
[f_{\alpha\beta\gamma}\;g_{\mu\nu\alpha} -f\leftrightarrow g]\;
:F_\mu F_\nu \;F_{\beta} \;G_\gamma:
+
\sum_{\mu,\nu,\alpha,\gamma\le M_0;\beta\le M_1}
[f_{\alpha\beta\gamma}\;g_{\mu\nu\beta} -f\leftrightarrow g]\;
:F_\mu F_\nu F_{\alpha}\; G_\gamma:
\nonumber\\
&=&
\sum_{\mu,\nu,\alpha,\gamma\le M_0;\beta\le M_1}\;
[(f_{\alpha\beta\gamma}+f_{\beta\alpha\gamma})
\;g_{\mu\nu\beta} -f\leftrightarrow g]\;
:F_\mu F_\nu F_{\alpha}\; G_\gamma:
\nonumber\\
&=&
\sum_{\mu,\nu,\alpha,\gamma\le M_0;\rho\le M_1}\;
[(f_{\alpha\rho\gamma}+f_{\rho\alpha\gamma})
\;g_{\mu\nu\rho} -f\leftrightarrow g]\;
:F_\mu F_\nu F_{\alpha}\; G_\gamma:
\ea
where we have relabelled $\beta$ by $\rho$ in the last step.
Thus except for the range, (\ref{3.179}) and (\ref{3.179a}) are identical.
This is in general true for terms with a single Kronecker symbol.  

The remaining cut-off $M_3$ or $M_1$ can now be removed and we have with 
$\alpha=(Imi), \gamma=(Jnj),\mu=(Mur),\nu=(Nvs),\rho=(Kpk)$
\ba \label{3.180}
&& \sum_\rho 
[(f_{\rho\alpha\gamma}+f_{\alpha\rho\gamma})\;g_{\mu\nu\rho} 
-f\leftrightarrow g]\;
\nonumber\\
&=& \sum_K\;
[f_{KImJ} \delta_{pn}\epsilon_{kij}+f_{IKpJ}\delta_{mn}\epsilon_{ikj}
-f\leftrightarrow g]\;
g_{MNvK}\delta_{up}\epsilon_{rsk}
\nonumber\\
&=&
\sum_K
[f_{KImJ} \delta_{un}-f_{IKuJ}\delta_{mn} -f\leftrightarrow g]\;
\epsilon_{kij} g_{MNvK}\epsilon_{krs}
\nonumber\\
&=&-
\sum_K
[ <e_K e_I f, e_{J,m}>\delta_{un}
-<e_K e_I f, e_{J,u}>\delta_{mn}-
-f\leftrightarrow g]\;<(e_M e_N g)_{,v},e_K> 
\epsilon_{kij}\epsilon_{krs}
\nonumber\\
&=&-
[<(e_M e_N g)_{,v} e_I f, e_{J,m}>\delta_{un}
-<(e_M e_N g)_{,v} e_I f, e_{J,u}>\delta_{mn}
-f\leftrightarrow g]\; 
\epsilon_{kij}\epsilon_{krs}
\nonumber\\
&=&-
[<e_M e_N e_I \omega_v, e_{J,m}>\delta_{un}
-<e_M e_N e_I \omega_v, e_{J,u}>\delta_{mn}]\; 
\epsilon_{kij}\epsilon_{krs}
\ea
Thus the single Kronecker symbol contribution to the curly bracket is
\ba \label{3.181}
&& -
\sum_{I,J,M,N\le M_0}
[<e_M e_N e_I \omega_v, e_{J,m}>\;
:F_{Mnr} F_{Nvs} F_{Imi} G_{Jnj}:
-<e_M e_N e_I \omega_v, e_{J,u}>\; 
:F_{Mur} F_{Nvs} F_{Ini} G_{Jnj}:]
\epsilon_{kij}\epsilon_{krs}
\nonumber\\
&& -
\sum_{I,J,M,N\le M_0}\; <e_M e_N e_I \omega_v, e_{J,m}>\;
:(F_{Mnr}  F_{Imi}-F_{Mmr} F_{Ini}) F_{Nvs} G_{Jnj}:\;
(\delta_{ir}\delta_{js}-\delta_{is}\delta_{jr})
\nonumber\\
&=& -
\sum_{I,J,M,N\le M_0}\; <e_M e_N e_I \omega_v, e_{J,m}>\;
\{
:(F_{Mnr}  F_{Imr}- F_{Inr} F_{Mmr}) F_{Nvj} G_{Jnj}:
-
\nonumber\\ &&
:(F_{Mnj}  F_{Imi}-F_{Mmj} F_{Ini}) F_{Nvi} G_{Jnj}:
\}
\nonumber\\
&=& 
\sum_{I,J,M,N\le M_0}\; <e_M e_N e_I \omega_v, e_{J,m}>\;
:[(F_{Mnj}  F_{Imi}-F_{Mmj} F_{Ini}) F_{Nvi} G_{Jnj}:
\ea
where in the step before the last we used that the 
normal ordered factor in the first term is antisymmetric 
in $I,M$ while the prefactor coefficient is symmetric.

To interpret this result consider
\ba \label{3.182}
&& \int\; d^3x\; \omega_a \; \delta^{ij} \; 
:E^a_i E^b_j (\partial_b A_c^k-\partial_c A_b^k)\; E^c_k:
\nonumber\\
&=& -\frac{i}{4}\sum_{I,J,M,N}\;
<\omega_v e_N e_M e_I, e_{J,m}>\; 
\{
:(F_{Mmi} F_{Inj}-F_{Mni} F_{Imj}) G_{Jnj} F_{Nvi} :
\nonumber\\
&=& -\frac{i}{4}\sum_{I,J,M,N}\;
<\omega_v e_N e_M e_I, e_{J,m}>\; 
\{
:(F_{Imi} F_{Mnj}-F_{Ini} F_{Mmj}) G_{Jnj} F_{Nvi} :
\ea
which is precisely $-i/4$ times (\ref{3.181}).

We summarise
\be \label{3.183}
[C_{0;3}(f),C_{0;3}(g)]= 4i\; 
\int\; d^3x\; [f\; g_{,a}-g f_{,a}]\;\;:[E^a_i e^b_j \delta^{ij}
+\Lambda \delta^{ab}]\; (2\partial_{[b} A_{c]}^k) \;E^c_k
\ee
where $\Lambda=0$ e.g. for the limiting pattern with 
$M_1=M_2=M_4=M_5$ before $M_3=M_6$ and $\Lambda\not=0$ is a regularisation 
of $\delta(0,0)$ e.g. for the pattern
$M_3=M_6$ before $M_1=M_2=M_4=M_5$.\\ 
\\
{\bf 6B.}\\
The second curly bracket is given by 
\be \label{3.184}
%(-1/2^{1/2})(-1/4)
\frac{1}{4\; 2^{1/2}}\; 
\sum_{\alpha\le M_1,\beta\le M_2,\gamma\le M_3,
\mu\le M_4,\nu\le M_5,\rho\le M_6, \sigma\le M_7}\;
[f_{\alpha\beta\gamma}\; g_{\mu\nu\rho\sigma}- f\leftrightarrow g]\;
[\;:F_\alpha F_\beta G_\gamma:\;,\;
:F_\mu\; F_\nu\; G_\rho\; G_\sigma:\;]
\ee
where the evaluation of the commutator is given in (\ref{a.9}). \\
\\
{\bf 6B.i}\\
This is the first time that we have to deal with a triple product 
of Kronecker symbols whose contribution is 
\be \label{3.184a}
[\;:F_\alpha F_\beta G_\gamma:\;,\;
:F_\mu\; F_\nu\; G_\rho\; G_\sigma:\;]_{\delta^3}
=-16\; F_{(\mu}\;
\delta_{\nu)(\alpha}\; \delta_{\beta)(\rho}\; \delta_{\sigma)\gamma}
-8\; F_{(\mu}\;
\delta_{\nu)\gamma}\; \delta_{\alpha(\rho}\; \delta_{\sigma)\beta}
\ee
Whatever scheme we choose, the Kronecker symbols identify the 
same finite range sub-indices $m=n,i=j$ in $\delta_{\alpha,\beta}$ 
if $\alpha=(Imi),\;\beta=(Jnj)$. Thus if we can show that 
(\ref{3.184}) resticted to (\ref{3.184a}) vanishes due to symmetries
among the finite range sub-indices in one scheme at finite 
remaining cut-offs, it does so 
in any other because different schemes differ only in the 
order in which the surviving ininite range indices are summed to $\infty$.   
For this purpose we consider the pattern 
$M_4,M_5, M_6, M_7$ before $M_1, M_2, M_3$ and 
solve for $\mu,\nu,\rho,\sigma$. 
After relabelling, sandwiching this results in a term proportional 
to 
\ba \label{3.184b}
&& \sum_{\mu\le M_0}\; F_\mu\;
\sum_{\alpha\le M_1,\beta\le M_2,\; \gamma\le M_3}\;
\times \\
&& 
\{f_{\alpha\beta\gamma}\;
[
(g_{\mu\alpha\beta\gamma}+
g_{\mu\alpha\gamma\beta}+
g_{\alpha\mu\beta\gamma}+
g_{\alpha\mu\gamma\beta})+
(g_{\mu\beta\alpha\gamma}+
g_{\mu\beta\gamma\alpha}+
g_{\beta\mu\alpha\gamma}+
g_{\beta\mu\gamma\alpha})+
(g_{\mu\gamma\alpha\beta}+
g_{\mu\gamma\beta\alpha}+
g_{\gamma\mu\alpha\beta}+
g_{\gamma\mu\beta\alpha})]\;
-\;f\leftrightarrow g\}
\nonumber
\ea
With $\alpha=(Imi),\; \beta=(Jnj),\; \gamma=(Kpk),\;\mu=(Lql)$ we note 
that the three groups of four $g$'s is proprtional 
to 
$\epsilon_{tli}\; \epsilon_{tjk},\;
\epsilon_{tlj}\; \epsilon_{tik},\;
\epsilon_{tlk}\; \epsilon_{tij}$ respectively while $f_{\alpha\beta\gamma}$ 
is proportional to $\epsilon_{ijk}$. Summing over $i,j,k=1,2,3$ gives 
respectively 
$2\;\epsilon_{tli}\; \delta_{ti},\;
-2\epsilon_{tlj}\; \delta_{tj},\;
2\epsilon_{tlk}\; \delta_{tk}$
which all vanish when summing over $t$.\\
\\
{\bf 6B.ii}\\
The double Kronecker contribution to the commutator in 
(\ref{3.184}) is given by
\ba \label{3.184c}
&& [\;:F_\alpha F_\beta G_\gamma:\;,\;
:F_\mu\; F_\nu\; G_\rho\; G_\sigma:\;]_{\delta^2}
=
8\; : F_\mu\; F_\nu \;F_{(\alpha}:\;
\delta_{\beta)(\rho}\;\delta_{\sigma)\gamma} 
+8\; : G_\rho\; G_\sigma \;F_{(\alpha}:\;
\delta_{\beta)(\mu}\;\delta_{\nu)\gamma} 
\nonumber\\
&& -16\; : G_\gamma\; G_{(\rho}\;
\delta_{\sigma)(\alpha}\;\delta_{\beta)(\mu}\; F_{\nu)}
\ea
We consistently keep the pattern 
$M_4,M_5, M_6, M_7$ before $M_1, M_2,M_3$.\\
\\ 
Then the first term in (\ref{3.184c}) gives using sandwiching 
and relabelling a contribution proportional to, when 
solving for $\rho,\sigma$ 
\be \label{3.184d}
2\sum_{\mu,\nu,\alpha\le M_0}\;
:F_\mu\; F_\nu\; F_\alpha;\;\sum_{\gamma\le M_3}\;
([\sum_{\beta\le M_2}\;f_{\alpha\beta\gamma}\;
+\sum_{\beta\le M_1}\;f_{\beta\alpha\gamma}]\;
[g_{\mu\nu\beta\gamma}+g_{\mu\nu\gamma\beta}]\;-f\leftrightarrow g)
\ee
With $\alpha=(Imi),\beta=(Jnj),\gamma=(Kpk),\; 
\mu=(Mur),\; \nu=(Nvs)$ we have explicitly
\ba \label{3.184e}
&& f_{\alpha\beta\gamma}=
-<e_I e_J f,e_{K,n}>\delta_{mp}\epsilon_{ijk},\;  
f_{\beta\alpha\gamma}=
<e_I e_J f,e_{K,m}>\delta_{np}\epsilon_{ijk},\;  
\nonumber\\
&& g_{\mu\nu\beta\gamma}
=<e_M e_N e_J e_K,g>\delta_{un}\delta_{vp}\epsilon_{trs}\epsilon_{tjk},\;
g_{\mu\nu\gamma\beta}
=-<e_M e_N e_J e_K,g>\delta_{up}\delta_{vn}\epsilon_{trs}\epsilon_{tjk},\;
\ea
If we remove $M_1,M_2$ at finite $M_3$ we can perform
the sum over $J$ in (\ref{3.184d}) and obtain an expression which depends 
only on $fg$ 
and from which we subtract the same expression with 
$f\leftrightarrow g$. Thus (\ref{3.184d})
vanishes in this pattern. If on the other hand we remove 
$M_3$ at finite $M_1,M_2$ then we possibly obtain a contribution which 
{\it renormalises the cosmological constant} because that term 
depends on the derivatives of $f,g$.

To see whether this is the case we substitute (\ref{3.184e}) into 
(\ref{3.184d}) and obtain
\ba \label{3.184e1}
&& 4\sum_{\mu,\nu,\alpha\le M_0}\;
:F_\mu\; F_\nu\; F_\alpha:\;\epsilon_{irs}\;
\sum_{K\le M_3}\;
([\sum_{J\le M_2}\;[f_{IJnK}\delta_{mp}
-\sum_{J\le M_1}\;f_{JImK}\delta_{np}]\;
g_{MNJK}\;[\delta_{un}\delta_{vp}-\delta_{up}\delta_{vn}]
\;-f\leftrightarrow g)
\nonumber\\
&&= 4\sum_{\mu,\nu,\alpha\le M_0}
:F_\mu F_\nu F_\alpha: \epsilon_{irs}
\sum_{K\le M_3}
(\{\sum_{J\le M_2} [\delta_{un}\delta_{vm}-\delta_{um}\delta_{vn}]
-\sum_{J\le M_1}
[\delta_{un}\delta_{vn}-\delta_{un}\delta_{vn}]\}
f_{IJnK} g_{MNJK}
-f\leftrightarrow g)
\nonumber\\
&&= 4\sum_{\mu,\nu,\alpha\le M_0}\;
:F_\mu\; F_\nu\; F_\alpha:\;\epsilon_{irs}\;
\sum_{K\le M_3}\;\sum_{J\le M_2}\; (g_{MNJK}\;
[f_{IJuK}\;\delta_{vm}-f_{IJvK}\;\delta_{um}]
\;-f\leftrightarrow g)
\ea
We remove $M_3$ and perform the sum over $K$ unconstrained which 
gives  
\ba \label{3.184e2}
&& -4\sum_{M,N,I\le M_0}\;
:F_{Mur}\; F_{Nvs}\; F_{Imi}:\;\epsilon_{irs}\;
\sum_{J\le M_2}\; 
\times\nonumber\\
&& (<e_J^2 e_M e_N e_I,[f_{,u}\; g-f\; g_{,u}]>\;\delta_{vm}
-<e_J^2 e_M e_N e_I,[f_{,v}\; g-f\; g_{,v}]>\;\delta_{um})
\nonumber\\
&=& -8\sum_{M,N,I\le M_0}\;
:F_{Mur}\; F_{Nvs}\; F_{Imi}:\;\epsilon_{irs}\;
\sum_{J\le M_2}\; 
<e_J^2 e_M e_N e_I,[f_{,u}\; g-f\; g_{,u}]>\;\delta_{vm}
\nonumber\\
&=& 8\sum_{M,N,I\le M_0}\; \sum_{J\le M_2}\; 
:F_{Mur}\; F_{Nvs}\; F_{Ivi}:\;\epsilon_{irs}\;
<e_J^2 e_M e_N e_I,\omega_u>
\nonumber\\
&=& 0
\ea
with $\omega_u=f \; g_{,u}-f_{,u}\;g$. 
Here we have relabelled $M,u,r$ with $N,v,s$ in the second step 
and used commutativity under the normal ordering symbol while in the 
last step we used that $\epsilon_{irs}<e_J^2 e_M e_N e_I, \omega_u>$
is antisymmetric under $(N,s)\leftrightarrow (I,i)$ while 
$:F_{Mur}\; F_{Nvs}\; F_{Ivi}:$ is symmetric and that $I,N\le M_0$ 
have the same range.  Accordingly, there 
is no renormalisation of the cosmological constant term in this 
pattern.\\
\\
The second term in (\ref{3.184c}) gives after sandwiching, relabelling
and solving for $\mu,\nu$
\be \label{3.184f}
2\sum_{\rho,\sigma,\alpha\le M_0}\;
:G_\rho\; G_\sigma\; F_\alpha;\;\sum_{\gamma\le M_3}\;
([\sum_{\beta\le M_2}\;f_{\alpha\beta\gamma}\;
+\sum_{\beta\le M_1}\;f_{\beta\alpha\gamma}]\;
[g_{\beta\gamma\rho\sigma}+g_{\gamma\beta\rho\sigma}])\;-f\leftrightarrow g)
\ee
Using the expressions (\ref{3.184e}) we see by the same argument 
that in the pattern $M_1,M_2$ before $M_3$ the expression 
(\ref{3.184f}) vanishes. In the pattern $M_3$ before $M_1,M_2$
the contribution including the prefactors to the curly bracket is 
with $\alpha=(Imi),\beta=(Jnj),\gamma=(Kpk),\rho=(Mur),\sigma=(Nvs)$
\ba \label{3.184g}
&& 
\frac{1}{2\; 2^{1/2}}
\sum_{\rho,\sigma,\alpha\le M_0}\;
:G_\rho\; G_\sigma\; F_\alpha;\;\sum_{\gamma\le M_3}\;
([\sum_{\beta\le M_2}\;f_{\alpha\beta\gamma}\;
+\sum_{\beta\le M_1}\;f_{\beta\alpha\gamma}]\;
[g_{\beta\gamma\rho\sigma}+g_{\gamma\beta\rho\sigma}])\;-f\leftrightarrow g)
\nonumber\\
&=& 
\frac{1}{2\; 2^{1/2}}
\sum_{\rho,\sigma,\alpha\le M_0}\;
:G_\rho\; G_\sigma\; F_\alpha;\;\sum_{K\le M_3}\;
\{
[\sum_{J\le M_2}\;  f_{IJnK} \delta_{mp} \epsilon_{ijk}
-\sum_{J\le M_1}\; f_{JImK} \delta_{np} \epsilon_{ijk}]\;
\times\nonumber\\
&&
[g_{KJMN} \delta_{pu}\delta_{nv}\epsilon_{tkj}\epsilon_{trs}
-g_{JKMN} \delta_{nu}\delta_{pv}\epsilon_{tkj}\epsilon_{trs}] 
- f \leftrightarrow g
\}
\nonumber\\
&=& 
-\frac{1}{2^{1/2}}
\sum_{\rho,\sigma,\alpha\le M_0}\;
:G_\rho\; G_\sigma\; F_\alpha;\;\sum_{K\le M_3}\;
\epsilon_{irs}\;
\{
[\sum_{J\le M_2}\;  f_{IJnK} \delta_{mp}
-\sum_{J\le M_1}\; f_{JImK} \delta_{np}]\;
\times\nonumber\\
&& [g_{KJMN} \delta_{pu}\delta_{nv}
-g_{JKMN} \delta_{nu}\delta_{pv}] 
- f \leftrightarrow g
\}
\nonumber\\
&=& 
-\frac{1}{2^{1/2}}
\sum_{\rho,\sigma,\alpha\le M_0}\;
:G_\rho\; G_\sigma\; F_\alpha;\;\sum_{K\le M_3}\;
\epsilon_{irs}\;
\times\nonumber\\
&& \{
\sum_{J\le M_2}\;  f_{IJnK}\;g_{JKMN}\;
[\delta_{mu}\delta_{nv}-\delta_{nu}\delta_{mv}] 
-\sum_{J\le M_1}\; f_{JImK}\;g_{JKMN}\;
[\delta_{nu}\delta_{nv}-\delta_{nu}\delta_{nv}] 
- f \leftrightarrow g
\}
\nonumber\\
&=& 
-\frac{1}{2^{1/2}}
\sum_{\rho,\sigma,\alpha\le M_0}\;
:G_\rho\; G_\sigma\; F_\alpha;\;\sum_{K\le M_3}\;
\epsilon_{irs}\;
\{\sum_{J\le M_2}\; \;g_{JKMN}\;
[\delta_{mu}\;f_{IJvK}-\delta_{mv}\;f_{IJuK}] 
- f \leftrightarrow g
\}
\nonumber\\
&=& 
-2^{1/2}
\sum_{\rho,\sigma,\alpha\le M_0}\;
:G_\rho\; G_\sigma\; F_\alpha;\;\sum_{K\le M_3}\;
\epsilon_{irs}\;
\{
\sum_{J\le M_2}\; \;g_{JKMN}\;
[\delta_{mu}\;f_{IJvK} - f \leftrightarrow g
\}
\ea
where in the last step we interchanged $Mur$ and $Nvs$ and used 
commutativity under the normal ordering symbol.

We remove $M_3$ and sum unconstrained over $K$ which gives 
\ba \label{3.184h}
&&
-2^{1/2}
\sum_{M,N,I\le M_0}\;
:G_{Mur}\; G_{Nvs}\; F_{Imi};\;\sum_{J\le M_2}\; 
\sum_{K\le M_3}\;
\epsilon_{irs}\;
\{
g_{JKMN}\;f_{IJvK}\;\delta_{mu}
- f \leftrightarrow g
\}
\nonumber\\
&=&
-2^{1/2}
\sum_{M,N,I\le M_0}\;
:G_{Mur}\; G_{Nvs}\; F_{Iui};\;\sum_{J\le M_2}\; 
\epsilon_{irs}\;
\{<e_J^2\; e_M e_N e_I,[f_{v} g-f\; g_{,v}]>
\nonumber\\
&=&
2^{1/2}
\sum_{M,N,I\le M_0}\;
:G_{Mur}\; G_{Nvs}\; F_{Iui};\; 
\epsilon_{irs}\;
<\Lambda_{M_2}\; e_M e_N e_I \omega_v>
\ea
where $\Lambda_M=\sum_{I\le M} e_I^2,\; \omega_v=f g_{,v}-f_{,v} g$.

To interpret this term we consider 
\be \label{3.209}
\int d^3x\; \Lambda \omega_a \delta^{ab}\;\epsilon_{ijk} 
\;:A_b^j \; A_c^k\; E^c_i:
=\frac{i}{2\; 2^{1/2}}\;\epsilon^{irs}
\sum_{I,M,N} <\Lambda \omega_v e_I e_M e_N> G_{Nvr} G_{Mus} F_{Jui} 
\ee
Thus (\ref{3.184h}) is $4i$ times (\ref{3.209}).\\ 
\\
The third term in (\ref{3.184c}) gives after sandwiching, relabelling
and solving 
for $\mu$ or $\nu$ and $\rho$ or $\sigma$ 
\be \label{3.184g1}
-2\sum_{\mu,\sigma,\gamma\le M_0}\;
:F_\mu\; G_\sigma\; G_\gamma:\;
\sum_{\alpha\le M_1,\beta\le M_2}\;
(f_{\alpha\beta\gamma}\;
[g_{\mu\alpha\beta\sigma}+
g_{\mu\beta\alpha\sigma}+
g_{\alpha\mu\beta\sigma}+
g_{beta\mu\alpha\sigma}+
g_{\mu\alpha\sigma\beta}+
g_{\mu\beta\sigma\alpha}+
g_{\alpha\mu\sigma\beta}+
g_{beta\mu\sigma\alpha}]
-f\leftrightarrow g)
\ee
This term vanishes, in both patterns $M_1,M_2$ before $M_3$ or 
$M_3$ before $M_1,M_2$ because $M_3$ is locked, by the same 
argument as before because we only get contributions depending 
on $fg$ when e.g. removing $M_1$ before $M_2$ and then 
we subtract the same expression with $f,g$ interchanged.\\  
\\
{\bf 6B.iii}\\
We now consider the single Kronecker symbol contribution to (\ref{3.184})
which is 
\be \label{3.184h1}
[\;:F_\alpha F_\beta G_\gamma:\;,\;
:F_\mu\; F_\nu\; G_\rho\; G_\sigma:\;]_{\delta^1}
=
-4\delta_{\gamma(\mu}\;:F_{\nu)}\; F_\alpha\; F_\beta \; 
G_\rho\; G_\sigma:
+8\; 
:F_{(\alpha}\; \delta_{\beta)(\rho}\; G_{\sigma)}\; F_\mu\; 
F_\nu\; G_\gamma:
\ee
Its contribution to (\ref{3.184}) is obtained by sandwiching and solving 
the Kronecker for $\mu$ or $\nu$ in the first term and for $\rho$ or 
$\sigma$ in the second
\ba \label{3.184i}
&&-\frac{1}{2\; 2^{1/2}}\; 
\sum_{\alpha,\beta,\rho,\sigma,\nu\le M_0} \;
:F_\nu\; F_\alpha\; F_\beta \; G_\rho\; G_\sigma:\;
\{
\sum_{\gamma\le M_3}
f_{\alpha\beta\gamma} \; [g_{\gamma\nu\rho\sigma}+g_{\nu\gamma\rho\sigma}]
- f\leftrightarrow g]
\nonumber\\
&&
+\sum_{\alpha,\sigma,\mu,\nu,\gamma\le M_0}\; 
:F_\alpha\; G_\sigma\; F_\mu\; F_\nu\; G_\gamma:\;
\{
[\sum_{\beta\le M_2} f_{\alpha\beta\gamma}
+\sum_{\beta\le M_1} f_{\beta\alpha\gamma}]\;
[g_{\mu\nu\beta\sigma}+g_{\mu\nu\sigma\beta}]
- f\leftrightarrow g
\}
\ea
We can now remove the remaining cut-offs. Then we see that the terms 
with the remaining sums over $\beta$ vanish because they produce just terms 
depending on $fg$ and therefore subtraction of the same term with $f,g$ 
interchanged leads to cancellation. We are thus left with with the first term 
in (\ref{3.184i}) given by
\be \label{3.211}
-\frac{1}{2^{1/2}}\; 
\sum_{\alpha,\beta,\rho,\sigma,\mu\le M_0;\gamma\le M_3}
\{
f_{\alpha\beta\gamma}\; 
[g_{\gamma\mu\rho\sigma} + g_{\mu\gamma\rho\sigma}- f\leftrightarrow g\}\;
:F_{\mu} F_\alpha F_\beta G_\rho G_\sigma:
\ee
Using $\alpha=(Imi),\beta=(Jnj),\gamma=(Kpk),\mu=(Mqr),\rho=(Sus),\sigma=(Tvt)$ 
we can perform the sum over $\gamma$ unconstrained removing 
$M_3$. We have
\ba \label{3.212}
&& \sum_\gamma
\{f_{\alpha\beta\gamma}\; 
[g_{\gamma\mu\rho\sigma} + g_{\mu\gamma\rho\sigma}
- f\leftrightarrow g\}\;
\nonumber\\
&=&
\sum_K\; 
\{f_{IJnK} \delta_{mp}\epsilon_{ijk}
[g_{KMST} \delta_{pu}\delta_{qv} \epsilon_{lkr}\epsilon_{lst}
+g_{MKST} \delta_{qu}\delta_{pv} \epsilon_{lrk}\epsilon_{lst}]
- f\leftrightarrow g\}\;
\nonumber\\
&=&
\sum_K\; 
\{f_{IJnK} g_{KMST} \epsilon_{ijk}\epsilon_{lkr}\epsilon_{lst}
[\delta_{mu}\delta_{qv}-\delta_{qu}\delta_{mv}]
- f\leftrightarrow g\}\;
\nonumber\\
&=&
\{<(e_I e_J f)_{,n}, e_M e_S e_T g>- f\leftrightarrow g\}\;
\epsilon_{ijk}\epsilon_{lkr}\epsilon_{lst}
[\delta_{mu}\delta_{qv}-\delta_{qu}\delta_{mv}]
\nonumber\\
&=&-
<e_I e_J e_M e_S e_T, \omega_n>\;
\epsilon_{ijk}\epsilon_{lkr}\epsilon_{lst}
[\delta_{mu}\delta_{qv}-\delta_{qu}\delta_{mv}]
\ea
Thus (\ref{3.211}) becomes
\ba \label{3.213}
&&\frac{1}{2^{1/2}}\; 
\sum_{I,J,M,S,T\le M_0}\;
<e_I e_J e_M e_S e_T, \omega_n>\;
\epsilon_{ijk}\epsilon_{lkr}\epsilon_{lst}
:F_{Mvr} F_{Imi} F_{Jnj} (G_{Sms} G_{Tvt}-G_{Svs} G_{Tmt}):
\nonumber\\
&=&\frac{1}{2^{1/2}}\; 
\sum_{I,J,M,S,T\le M_0}\;
<e_I e_J e_M e_S e_T, \omega_n>\;
\epsilon_{ijk}\epsilon_{lkr}\epsilon_{lst}
:F_{Mvr} F_{Imi}
 F_{Jnj} (G_{Sms} G_{Tvt}-G_{Tvs} G_{Smt}):
\nonumber\\
&=&\frac{1}{2^{1/2}}\; 
\sum_{I,J,M,S,T\le M_0}\;
<e_I e_J e_M e_S e_T, \omega_n>\;
\epsilon_{ijk}\epsilon_{lkr}\epsilon_{lst}
:F_{Mvr} F_{Imi} F_{Jnj} (G_{Sms} G_{Tvt}+G_{Tvt} G_{Sms}):
\nonumber\\
&=& 2^{1/2}\; 
\sum_{I,J,M,S,T\le M_0}\;
<e_I e_J e_M e_S e_T, \omega_n>\;
(\delta_{ir} \delta_{jl}-\delta_{il}\delta_{jr})
\epsilon_{lst}
:F_{Mvr} F_{Imi} F_{Jnj} G_{Sms} G_{Tvt}:
\nonumber\\
&=& 2^{1/2}\; 
\sum_{I,J,M,S,T\le M_0}\;
<e_I e_J e_M e_S e_T, \omega_n>\;
\epsilon_{lst}
:(F_{Mvr} F_{Imr} F_{Jnl} - F_{Mvr} F_{Iml} F_{Jnr}) 
G_{Sms} G_{Tvt}:
\nonumber\\
&=& 2^{1/2}\; 
\sum_{I,J,M,S,T\le M_0}\;
<e_I e_J e_M e_S e_T, \omega_n>\;
\epsilon_{lik}
:F_{Jnj} F_{Mvj} G_{Tvi}  G_{Smk} F_{Iml} :
\ea
where in the last relation we used that $F_{Mvr} F_{Imr}$ is symmetric 
under exchange of $m,v$ while $\epsilon_{lst} G_{Sms} G_{Tvt}$ is 
antisymmetric when both are contracted with the totally symmetric 
$<e_I e_J e_S e_T e_M, \omega_n>$.

To interpret this expression consider
\be \label{3.214}
\int d^3x\; \omega_a :E^a_j\; E^b_j \;\epsilon_{lik}\; A_b^i\; A_c^k\; E^c_l:
=-\frac{i}{2 2^{1/2}} 
\sum_{I,J,M,S,T}\;
<e_I e_J e_M e_S e_T, \omega_n>\;
\epsilon_{lik}
:F_{Jnj} F_{Mvj} G_{Tvi}  G_{Smk} F_{Iml} :
\ee
hence (\ref{3.213}) is $4i$ times (\ref{3.214}).

Altogether we obtain 
\be \label{3.215}
[C_{0;3}(f),C_{0;4}(g)]=4i
\int\; d^3x\; \omega_a \; 
:(E^a_i E^b_j\delta^{ij}+\Lambda \delta^{ab})
\epsilon_{ijk} A_b^i A_c^j E^c_k:
\ee
e.g. in the pattern with $M_4,M_5,M_6,M_7$ before $M_1,M_2,M_3$ and 
$M_1,M_2$ before $M_3$ and $M_1$ before $M_2$ 
for $\Lambda=0$ and $M_3$ before $M_1,M_2$ 
for $\Lambda\not=0$ a regulator dependent constant.\\
\\
{\bf 6C.}\\
The third curly bracket is given by 
\be \label{3.216}
%1/2 (-1/4)^2 
\frac{1}{32} 
\sum_{\alpha\le M_1,\beta\le M_2,\gamma\le M_3,\delta\le M_4,
\mu\le M_5,\nu\le M_6,\rho\le M_7,\sigma\le M_8}\;
[f_{\alpha\beta\gamma\sigma}\;g_{\mu\nu\rho\sigma}
- f\leftrightarrow g]\;
[\;:F_\alpha F_\beta G_\gamma G_\delta:\;,\;
F_\mu F_\nu G_\rho G_\sigma:\;]
\ee
where the commutator is worked out in (\ref{a.10}).

Since $f_{\alpha\beta\gamma\delta},\;
g_{\mu\nu\rho\sigma}$ 
do not depend on derivatives of $f,g$,
(\ref{3.216}) vanishes for 
any pattern such that after solving the Kronecker symbols we 
obtain a sum that depends on cut-offs which are to be taken 
away sequentially. Consider e.g. the pattern $M_n$ before $M_{n-1}$ with 
$n=2,..,8$. Then we can solve 
all Kronecker symbols for $\mu,\nu\,\rho,\sigma$. In the terms with 
$1,2,3$ Kronecker factors the normal ordered monomial carries 
$6,4,2$ indices that are locked by sandwiching and 
$1,2,3$ of the indices 
$\mu,\nu,\rho,\sigma$ are equated with one of $\alpha,\beta,\gamma,\delta$
so that $1,2,3$ indices are left over for summing. Since these are 
summed sequentially according to the chosen pattern, removal of the 
top cut-off leads to an expression that depends just on $fg$ minus 
the same expression depending on $gf$. Thus for 
the pattern $M_8>M_7>..>M_1$ we find  
\be \label{3.217}
[C_{0;4}(f),C_{0;4}(g)]=0
\ee
~\\
Altogether
\be \label{3.218}
[C_0(f), C_0(g)]=
i\;\int\; d^3x\; [4(f g_{,a}-f_{,a} g)]\; 
:\; (E^a_i E^b_j \delta^{ij}+\Lambda \delta^{ab}) \; F_{bc}^k\; E^c_k: 
\ee
\end{proof}  
~\\
The course of the proof has revealed the following result:
\begin{Corollary} \label{col3.1} ~\\
There exist limiting patterns for the commutator of two
normal ordered 
Hamiltonian constraints such that with $\omega=f\;dg-g\;df$
\be \label{3.220}
[C_0(f),C_0(g)]=4i\; \int\; d^3x\; \omega_a\;
:(\delta^{mn}\; E^a_m\; E^b_n+\Lambda\; \delta^{ab})\;
F_{bc}^j\; E^c_j:
\ee
where $\Lambda$ is a regularised constant depending on the limiting pattern. 

\end{Corollary}
~\\
Remarks:\\
0.\\
The investigation illustrates the fact that different limiting patterns
lead to different results. We did not consider the most general 
pattern but focussed on those that tend to avoid derivatives of 
test functions in potentially divergent contributions.\\
1.\\
Note that we work in Planck units, restoring SI units would reveal 
that the $\Lambda$ term is multiplied by $\ell_P^2$.\\   
2.\\
The constant $\Lambda=\lim_M \Lambda_M,\;
\Lambda_M(x)=\sum_I e_I(x)^2$ is naively infinite, 
but can be given a finite 
value using $\zeta$ function regularisation. At finite $M$, the function 
$\Lambda_M$ is in fact a smooth function of rapid decrease for 
$\sigma=\mathbb{R}^3$ while for $\sigma=T^3$ it is just 
$\propto (M+1)^3$ without $x$ dependence. 
It would be interesting 
to study this in more detail to see whether alternative limiting patterns
exist such that $\Lambda$ comes out finite without additional 
regularisation.\\
3.\\
As expected, the $\Lambda$ term carries explicit dependence
on the chosen background metric $g_{ab}=\delta_{ab}$. It is 
interesting to see that all the other terms carry no background dependence
despite the fact that the chosen Fock representation also depends 
on that background.\\
4.\\
We have performed the calculation in the simplest possible Fock 
representation corresponding to a white noise covariance in order 
to minimise the already substantial computational effort. It is conceivable 
that both proposition and corollary also hold in more general Fock 
reprsentations. This has been confirmed for the spatial diffeomorphism
subalgebra in any dimension and for all tensor types in \cite{29}.\\
5.\\
It would be easy to extend the analysis by a cosmological constant term 
without increasing the computational effort too much and using the 
same tools. This just requires to extend the appendix by five 
more commutators namely $[:F^3:,:F^3:]$ and $[:F^3;,:X:],\;
X=F,\; FG, FFG,FFGG$. The fact that we compute commutators 
of linear combinations of quadratic forms using the 
distributive law allows to reuse the above analysis, one just has to supplement
it with the additional commutators and limiting patterns.\\
6.\\
The ``anomalous'' $\Lambda$ term is in fact again proportional to a normal 
ordered constraint, so one may call this a ``soft anomaly'' in the sense 
that the anomaly just changes the structure functions but not the set 
of constraints. This is similar to what happens in LQG \cite{16} but 
while here we find a ``quantum correction'' to the ``classically expected''
structure function, in LQG the expected part is missing and one has just 
a quantum correction.\\
7.\\
As already mentioned, while the computation is therefore as non-anomalous 
as it can possibly be since only normal ordered expressions are 
well defined quadratic forms on the Fock space, 
the presence of the structure functions
$\omega_a\; E^a_k E^b_k$ for $C_b$ and 
$\omega_a\; E^a_k E^b_k A_b^j$ for $G_j$ 
implies that $l([C_0(f),C_0(g)]\;\psi)\not=0$ even if $l$ solves 
all three types of constraints. Once again, there is no contradiction 
because only the object $[C_0(f),C_0(g)]$ is a well defined quadratic 
form since $C_0(g)\psi,C_0(f)\psi$ respectively is not in the form domain 
of $C_0(f),C_0(g)$ respectively for any Fock state $\psi$
so that there is no reason to expect that 
this quantity vanishes. 
In that sense the commutator calculation merely 
confirms that one has a quantum represention of the hypersurface deformation
algebra and thus has chosen a valid quantisation of the constraints.\\
8.\\  
As a byproduct we have also confirmed that the constraint algebra 
of Euclidian $U(1)^3$ gravity (e.g. \cite{32} and references therein)
is in this sense properly represented in the same Fock representation
because
going through the details of the proof one sees that one just has to 
drop the $FG$ term from the Gauss constraint and the $FFGG$ from the 
Hamiltonian constraint and thus can drop the commutators corresponding 
to 1B., 1C., 3B., 4B., 4C., 4D., 5B, 6B., 6C. Both the proposition and the 
corollary remain literally intact.\\
9.\\
Inspection of the individual parts of the computation also reveals that 
the respective quadratic form monomials have the expected commutators among 
themselves: For example the $FG$ part of the Gauss constraint has 
vanishing commutator seperately with the $FFGG$ part of 
the Hamiltonian constraint 
which is the quantum statement of the fact that this term of the 
Gauss constraint generates SU(2) rotations and this part of the
Hamiltonian constraint is invariant under those. In general it is 
reassuring to see that all parts of the Poisson bracket calculation are 
refelected in the quantum commutator calculation, most importantly 
that ultra-local terms drop out and that one can use the 
crucial commutativity 
of $F,G$, which in the classical calculation is granted, also 
in the quantum computation under the normal ordering symbol.\\
10.\\
In LQG one avoids the complicated ``Lorentzian correction'' to 
the Euclidian Hamiltonian constraint $C_0(f)$, which depends on the 
Ricci scalar of spin connection $\Gamma$ of $E$, 
by using either a ``Wick rotation'' 
\cite{16a} with generator $\int\; d^3x \; \Gamma_a^j \; E^a_j$ or 
by using the commutator equivalent of the Poisson bracket\\
 $\{\{V,C_0(1)\},A_a^j\}\;\{\{V,C_0(1)\},A_b^k\}\{V,A_c^k\}
\epsilon^{abc}\epsilon_{jkl}$ where $V=\int d^3x \sqrt{|\det(E)|}$
is the total volume. Both $\Gamma,V$ depend non-polynomially 
on $E$ so that this approach leaves the framework of the present section
which relies on polynomial expressions. Thus within the present section 
one would need to deal with the polynomial version of the constraint 
mentioned at the beginning of this section.\\
\\
The task to deal with non-polynomial expressions naturally leads us to 
the next section.

\section{Geometric measures, master constraint 
and reduced Hamiltonian as quadratic forms 
on Fock space}
\label{s4}

As already mentioned, the reduced phase space approach unavoidably leads 
to rather non-polynomial expressions typically involving integrals 
over square roots of density weight two scalars which themselves depend 
on polynomial expressions of all the fields except for the spatial 
metric $q_{ab}$ which enters also with integer powers of $\sqrt{\det(q)}$.
This happens because wehn solving the constraints explicitly 
one is forced to 
solve them for canonical momenta because these are the only variables 
which enter 
the constraints polynomially and (at least sufficient subset of them)
without spatial derivatives so that an explicit, algebraic
and local solution is possible. The gauge fixing conditions are then 
imposed on the conjugate configuration variables. Then the reduced 
Hamiltonian is a linear combination of those momenta that one solved 
the constraints for which therefore is a density of weight one. 
Thus the power of $\sqrt{\det(q)}$ that enters various terms of the 
reduced Hamiltonian under the square root makes sure that the overall
density weight is always two.\\
\\
It is therefore of considerable interest to explore whether 
$[\sqrt{\det(q(x))}]^N$ can be defined as a densely defined quadratic form 
on the Fock space for any interger $N\in \mathbb{Z}$. If it can, then
it will be a normal ordered expression of the gravitational annihilation 
and creation operators. Then more complicated expressions that involve 
also polynomials of the other fields can also be defined as quadratic forms 
again by normal ordering. Even more, we are not only interested in 
the reduced Hamiltonian but also other quadratic forms that involve 
$\sqrt{\det(q)}$ such as length, area and volume functions in three spatial 
dimensions and the various forms of the non-polynomial 
but spatially diffeomorphism invariant master constraint. 

In the next subsection we consider these geoemtric measures in the 
$q,p$ formulation (ADM variables)
in any spatial dimension $D$ where we obtain the most systematic 
results. Then we consider the geometric measures in the $e,P$ formulation 
also in any $D$ and the $A,E$ formulation in $D=3$ respectively which 
require a case by case analysis. Finally we construct the reduced 
Hamiltonian as a quadratic form using these tools.

\subsection{Metric densities as quadratic forms in $(q,p)$ variables}
\label{s4.1}

Let $S$ be a submanifold of $\sigma$ of dimension $1\le d\le D$ and 
$Y:\;U\subset \mathbb{R}^d\;\to S\subset\sigma;\; y\mapsto x=Y(y)$ 
be the corresponding embedding where $y,x$ are coordinates on $S,\sigma$ 
respectively. For $d=D$ we can pick $Y=$id and $y=x$. 
The volume of $S$ is then 
given by 
\be \label{4.1}
{\sf Vol}[S]=\int_U\; d^dy\; \sqrt{\det(([Y^\ast q](y))}
\ee
On the other hand for the reduced 
dynamics or the master 
constraint we are interested in $[\sqrt{\det(q(x))}]^N,\; N\in\mathbb{Z}$
which formally also concerns the case $Y=$id. Thus we capture all 
cases of interest if we consider 
\be \label{4.2}
Q^Y_N(y):=\sqrt{\det([Y^\ast q](y))}]^N 
\ee
We pick a smooth background metric $g$ on $\sigma$ with Euclidian signature. 
We consider a Fock representation using the following (background scalar) 
annihilation 
operator 
\be \label{4.3}
A_{jk}:=2^{-1/2}[\kappa\cdot q_{jk}-i\;\kappa^{-1}\cdot (\omega^{-1}\cdot 
\; p_{jk})]
\ee
where $h_a^j$ is an adapted background $D-$Bein 
$g_{ab}=\delta_{jk} h^j_a h^k_b$ with inverse $h^a_j$ and 
we defined the background scalars $q_{jk}=h^a_j h^b_k q_{ab},\;
p^{jk}=h_a^j h_b^k p^{ab}$ of background density weight zero and 
one respectively which enjoy canonical brackets 
\be \label{4.4a}
\{p^{jk}(x),q_{mn}(y)\}=\delta^j_{(m} \delta^k_{n)} \delta(x,y)
\ee
The D-Bein indices $j,k,l..$ are moved with the Kronecker symbol.
Here $\kappa$ is an invertible operator function of the 
Laplacian $\Delta=g^{ab}\nabla^g_a \nabla^g_b$ 
mapping scalars to scalars 
and $\omega=\sqrt{\det(g)}$ is the canonical background scalar 
density. 
The one particle Hilbert space is $\mathfrak{h}=L_2(\sigma,\omega\; d^Dx)$. 
The operator $-\Delta$ is positive symmetric on $\mathfrak{h}$ and has 
self-adjoint extensions with respect to which $\kappa$ is defined 
using the spectral theorem. 
The operator $\kappa$ has a symmetric bi-scalar integral kernel 
$\kappa(x,y)=\kappa(y,x)$ which acts on scalars $f$ via 
\be \label{4.4b}
[\kappa\cdot f](x):=\int\; d^Dy\;\omega(y)\;\kappa(x,y)\; f(y)
\ee
Denoting the kernel of its inverse by $\kappa^{-1}(x,y)$ defined 
analogous to (\ref{4.4b}) we find 
\be \label{4.4c}
\int\; d^Dz\; \omega(z)\;\kappa(x,z)\;\kappa^{-1}(z,y)=\delta(x,y)\;
\omega(y)^{-1},\;
\int\; d^Dy\; \delta(x,y)\; f(y)=f(x)
\ee
The annihilation operator smeared with scalar background test functions 
is (summation over repeated indices is understood)
\be \label{4.4d}
<f,A>:=<f^{jk},A_{jk}>_{\mathfrak{h}}
=2^{-1/2}\; [
<\kappa\cdot f^{jk},q_{jk}>_{\mathfrak{h}}
-<\omega^{-1}\;\kappa^{-1}\cdot f^{jk},p_{jk}>_{\mathfrak{h}}]
\ee
Note the canonical commutation relations 
\be \label{4.6a}
[<f,A>,\; <f',A>]=0,\;  
[<f,A>,\; (<h,A>)^\ast]=<f,h>\cdot 1,\;
<f,h>:=\int\; d^Dx\; \omega(x)\; \overline{f^{jk}(x)}\;
f'_{jk}(x)
\ee  
A coherent state on the corresponding Fock space is defined by 
$D(D+1)/2$ complex valued scalar fields $Z_{jk}=Z_{kj}$ such that 
\be \label{4.4}
||Z||^2:=<Z,Z>,\;
<Z,Z'>:=
\int\; d^Dx\; \omega\; 
\overline{Z_{jk}}\;\delta^{jm}\delta^{kn}\; Z'_{kn}<\infty
\ee
and 
\be \label{4.5}
\Omega_Z:=e^{-\frac{||Z||^2}{2}}\; e^{[<Z,A>]^\ast}\; \Omega
\ee
where $\Omega$ is the Fock vacuum defined by (\ref{4.3}). 
From 
$Z$ one reconstructs classical $q,p$ by 
\be \label{4.6}
q_{ab}=2^{-1/2}\; h_a^j\; h^k_b\; 
\kappa^{-1}\cdot [Z_{jk}+\overline{Z_{jk}}],\;
p^{ab}=2^{-1/2}\;i\; h^{aj}\; h^{bk}\;
\omega\; \kappa\cdot [Z_{jk}-\overline{Z_{jk}}],\;
\ee
Note that the classical metric  $q$ from (\ref{4.6}) on which the 
coherent state is concentrated is completely unrelated 
to the fixed background metric $g_{ab}$ that is used to define the 
Fock representation. 

Given any $\Omega_Z$ the span of its excitations  
\be \label{4.7}
[<f_1,A>]^\ast\;..[<f_N,A>]^\ast\; \Omega_Z
=e^{<Z,A>^\ast}\; [<f_1,A>]^\ast\;..[<f_N,A>]^\ast\; \Omega
\ee
defines a dense subspace of the Fock space as $f_1,..,f_N$ varies. This 
follows from the fact that the algebraic Fock state 
$<\Omega,.\Omega>$ is pure
\cite{21}, therefore its GNS representation of 
the Weyl algebra or quivalently the annihilation and creation algebra 
is irreducible and thus every vector state such as $\Omega_Z$ is 
cyclic.
\begin{Proposition} \label{prop4.1}
In any Fock representation,
(\ref{4.2}) can be defined non-perturbatively 
as a quadratic form on a dense form domain which correponds to the 
excitations of any coherent state of that Fock representation 
which is concentrated on
a positive definite spatial metric.
\end{Proposition}
\begin{proof}:\\
If $N$ is a positive even integer then $Q^Y_N$ is a polynomial in $q_{ab}$, 
therefore a well defined quadratic form is obtained by normal ordering 
it. Thus we can restrict attention to the case that $N$ is a positive 
odd integer or a negative integer. If $N$ is a positive odd integer
we define $Q^Y_N:=:Q^Y_{N+1} \; Q^Y_{-1}:$ so we are left with the case 
of negative integers. We can define $Q^Y_{-N},\; N>0$ in various 
ways, for example by $:\;[Q^Y_{-1}]^N\;:$ or by independent expressions
which we explain below. To define $Q^Y_{-1}(y)$ as a quadratic form 
we note the identity
\be \label{4.8} 
Q^Y_{-1}(y)=
\int_{\mathbb{R}^d}\; \frac{d^dz}{[2\pi]^{d/2}}\;
e^{-\frac{1}{2}\;\sum_{I,J=1}^d\; z^I\; [Y^\ast q]_{IJ}(y)\; z^J}
\ee 
which holds in the classical theory for any positive definite metric 
$q_{ab}$. We now use (\ref{4.6}) to express the Gaussian 
exponent in terms of 
$A_{jk}, A_{jk}^\ast$, specifically
\ba \label{4.9}
&& \frac{1}{2}\;z^T\; [Y^\ast q](y)\;z=
<f_Y(y,z),A>+<f_Y(y,z),A>^\ast,\;\;
f^{jk}_Y(x;y,z)=\hat{f}^{jk}_Y(y,z)\; \kappa^{-1}(Y(y),x),\;
\nonumber\\
&& \hat{f}^{jk}_Y(y,z)=
2^{-3/2}\; z^I\; z^J\; Y^a_{,I}(y)\;Y^b_{,J}(y)
h_a^j(Y(y))\; h_b^k(Y(y))\;
\ea
and normal order the exponential 
to arrive at the quadratic form equivalent of 
(\ref{4.8})
\be \label{4.10}
Q^Y_{-1}(y):=\int_{\mathbb{R}^d}\; \frac{d^dz}{[2\pi]^{d/2}}\;
[E_Y(z,y)]^\ast\;[E_Y(z,y)],\;\;
E_Y(y,z)=e^{-<f_Y(y,z),A>}
\ee 
One possibility to define $Q^Y_{-N}(y)$ is then to take the $N-th$ power of 
(\ref{4.10}) and to normal order
\be \label{4.11}
Q^Y_{-N}(y):=
\int_{\mathbb{R}^{ND}}\; \frac{d^{dN}z}{[2\pi]^{dN/2}}
(\prod_{k=1}^N\; [E_Y(z_k,y)]^\ast)\;
(\prod_{k=1}^N\; [E_Y(z_k,y)])\;
\ee
Another possibility is to make use of the classical idenity, for 
$N\ge d$ 
\be \label{4.12}
Q^Y_{-N}(y)=\frac{1}{c_d}\;
\int_{\mathbb{R}^{d^2}}\; d^{d^2}z\;
|\det(z_1,..,z_d)|^{N-d}\;
e^{-\frac{1}{2}\;\sum_{k=1}^d\;
\sum_{I,J=1}^d\; z_k^I\; [Y^\ast q]_{IJ}(y)\; z_k^J}
\ee
where $c_d$ is the value of the integral in (\ref{4.12}) obtained for 
$[Y^\ast q]_{IJ}=\delta_{IJ}$ and $(z_1,..,z_d)$ is the d x d matrix
whose $k-th$ column is given by the vector $z_k$. We can now 
define $Q^Y_{-N}$ by normal ordering the exponentials in (\ref{4.12})
after decomposing them via (\ref{4.9}). The advantage of (\ref{4.12})
over (\ref{4.11}) is that for large $N$ there are fewer integrals 
to compute. 

Yet another option is to take combinations of both possibilities 
or even to multiply by $1=Q^Y_{2M}(y)\; Q^Y_{-2M}$ until the desired overall 
negative power 
of $Q^Y_1(y)$ is reached, followed by normal ordering. 

It remains to show that these quadratic forms indeed have (\ref{4.7})
as dense form domain. For this we just need to use the identity
\ba \label{4.13}
&& E_Y(z,y)\; 
<f_1,A>^\ast\;..<f_N,A>^\ast\;\Omega_Z
\nonumber\\
&=&[<f_1,A>^\ast+<f_Y(z,y),f_1>]\;..[<f_N,A>^\ast\;+<f_Y(z,y),f_N>]\;
E_Y(z,y)\; \Omega_Z
\nonumber\\
&& <f_Y(y,z),f>=\hat{f}_Y^{jk}(y,z)\; [\kappa^{-1}\cdot f_{jk}](Y(y))
\nonumber\\
&& E_Y(z,y)\; \Omega_Z
= e^{-<f_Y(z,y),Z>}\; \Omega_Z
\ea
where we made use of the fact that coherent states diagonalise the 
annihilation operator.  

Since
\be \label{4.14}
<f_Y(y,z),Z>+<f_Y(y,z),Z>^\ast
=\hat{f}^{jk}(y,z)\;\kappa^{-1}\cdot[Z_{jk}+\overline{Z_{jk}}](Y(y))
=2^{1/2}\;\hat{f}^{jk}(y,z)\;q_{jk}(Y(y))
=\frac{1}{2} z^I z^J [Y^\ast q]_{IJ}(y)
\ee
it follows that the matrix elements of (\ref{4.10}) and its 
descendants are Gaussian integrals in $z^I\; z^J$ with Gaussian factor 
$\exp(-z^T \; [Y^\ast q](y) z/2)$, $q$ the classical metric determined 
by $Z$ via (\ref{4.6}), and coefficients 
$Y^\ast[(g\otimes g)[\kappa^{-1} f_k])_{IJ}(y),\; k=1,..,N$. The result of the 
Gaussian integral is a polynomial in the inverse $m^{IJ}$ of the pull back 
metric 
$m_{IJ}=[Y^\ast q]_{IJ}$ times $\sqrt{\det(m)}^{-1}$. To see this 
we introduce the new integration variables $z^A=m^A_I z^I$ where $m^A_I$
is a $d-$Bein $m_{IJ}=m_I^A h_J^B \delta_{AB}$, express 
$z^I=h^I_A z^B$ where $h_I^A h^I_B=\delta^A_B$ and finally use that
Gaussian integrals of $z^{A_1} ... z^{A_{2N}}$ are Wick polynomials 
in the Kronecker symbols $\delta^{A_k A_l},\; k,l=1,..,2N$. \\
\end{proof}
It is interesting to note that the dense subspace consisting of the 
excitations of the Fock vacuum {\it does not} provide a form domain 
for (\ref{4.2}) because the Fock vacuum is a coherent state 
concentrated on a classical metric of signature $(0,..,0)$. The same 
is true for any other signature, $s=(\sigma_1,..,\sigma_D),\; 
\sigma_k\in \{0,\pm 1\},\; k=1,..,D$ different from $(1,..,1)$ because 
the integral of $e^{-\sigma z^2/2}$ diverges unless $\sigma=1$. This 
is the non-degeneracy footprint that the classical theory leaves in 
the quantum theory. 
\begin{Corollary} \label{col4.1} ~\\
The volume functionals (\ref{4.1}) are 
quadratic forms with the same dense form domain as in proposition 
\ref{prop4.1}. In particular, their expectation value with respect 
to coherent states $\Omega_Z$ in the form domain equals the value 
of the classical functional evaluated on the metric determined by $Z$.
\end{Corollary} 
\begin{proof}:\\
We define the quadratic form by 
\ba \label{4.15}
&&{\sf Vol}[S]
:= \int_U\; d^dy \; :\;Q^Y_2(y)\;Q^Y_{-1}(y):
=\int\; d^dy\; \;\int\; \frac{d^d z}{[2\pi]^{d/2}}\;
E_Y(z,y)^\ast :\det([Y^\ast q](y): \; E_Y(y,z) 
\nonumber\\
&&\det([Y^\ast q](y)\;=\;
\frac{1}{d!}
\epsilon^{I_1..I_d}\;\epsilon^{J_1..J_d}
[Y^\ast q]_{I_1 J_1}\;..\; [Y^\ast q]_{I_d J_d}
\nonumber\\ &&
[Y^\ast q]_{IJ}(y)=
<F_{Y,IJ}(y),A>+<F_{Y,IJ}(y),A>^\ast,\;
F_{Y,IJ}^{ij}(x;y)=\hat{F}_{Y,IJ}^{ij}(y)\; 
\kappa^{-1}(Y(y),x),\;
\nonumber\\
&& \hat{F}_{Y,IJ}^{ij}(y)=h_a^i(Y(y))\; h_b^j(Y(y))\; 
Y^a_{,I}(y)\;Y^b_{,J}(y)\;   
\ea
In computing matrix elements of (\ref{4.15}) with respect to the 
states (\ref{4.7}) the prescription of (\ref{4.15}) is to perform the 
$z$ integral before the $y$ integral. As all annihilators and creators
act on the ket and bra respectively the action of $E_Y(y,z)$ can 
be taken over from proposition \ref{prop4.1} while the action of 
the normal ordered determinant returns a polynomial in the 
$<F_{Y,IJ}(y),f_k>,\; <F_{Y,IJ},Z>$. The $z$ integral can be performed
explicily and returns and overall factor of $Q^Y_{-1}(y)$ for $q=q(Z)$ the 
classical metric determined by $Z$ via (\ref{4.6}). 
For the state $\Omega_Z$ itself
(no excitation) the expectation value followed by the $z$ integral 
just returns $Q^Y_1(y)$ for $q=q(Z)$ and the $y$ integral therefore 
${\sf Vol}[S]$ at $q=q[Z]$. 
\end{proof}

\subsection{Metric densities as quadratic forms in other variables}
\label{s4.2}
 
Our results will be confined to defining integer powers of 
$Q=\sqrt{\det(q)}$ in the $(e,p)$ formulation in any $D$ and the 
$(A,E)$ formulation in $D=3$. The reason for why it is 
more complicated to define ${\sf Vol}[S]$ in these variables will become 
clear shortly.

We have the relations   
\be \label{4.21}
Q^2=\det(q)=\det(e)^2=|\det(E)|
\ee
therefore 
\be \label{4.22}
Q^{-1}=c\;\int_{\mathbb{R}^D}\; d^Dz\; F(e\cdot z)
=[\det(E)]^4\;c_3 \int_{\mathbb{R}^9}\; d^9z\; 
|\det((z_1,z_2,z_3))|^{1/2}\; F_3(E\cdot z_1,E\cdot 2, E\cdot z_3),\;
\ee
where $[e\cdot z]^j=e_a^j z^a,\; [E\cdot z_I]^a=E^a_j z^j_I$,  
$F,F_3$ are functions on $\mathbb{R}^D, \mathbb{R}^9$ respectively such 
that the integrals evaluated at $e_a^j=\delta_a^j$ and $E^a_j=\delta_a^j$
equal the constants $c^{-1}, c_3^{-1}$ respectively. We now 
write $F, F_3$ in terms of their Fourier transforms
\be \label{4.23}
F(e\cdot z)=\int\; \frac{d^Dk}{[2\pi]^D}\; \hat{F}(k)\;
e^{i\;k_j e^j_a z_a},\;
F_3(\{E\cdot z\})=\int\; \frac{d^9k}{[2\pi]^9}\; \hat{F}_3(k)
e^{i\;k_a^I E^a_j z^j_I},\;
\ee
Then we pick, similar as for $(q,P)$ variables
a Fock quantisation of $(e,p)$ and $(A,E)$ respectively 
using the background metric $g$ based on the annihlators
\ba \label{4.24}
A_{jk} &=& 2^{-1/2}[\kappa\cdot h^a_j e_{ak}-i\kappa^{-1}\cdot \omega^{-1} 
h_{aj} p^a_k]  
\nonumber\\
A_{jk} &=& 2^{-1/2}[\kappa\cdot h^a_j A_{ak}-i\kappa^{-1}\cdot \omega^{-1} 
h_{aj} E^a_k]  
\ea
(for $j\le a$ in upper triangular gauge when the gravitational Gauss 
constraint is gauge fixed) which enjoy canonical commutation 
relations. Finally one normal orders the exponentials (Weyl elements) in
(\ref{4.23}). 

To see why it is more dfiicult to define ${\sf Vol}[S]$ in these variables 
unless $S$ is $D$ dimensional and 3-dimensional respectively 
(in which case we just use $Q=Q^2 Q^{-1}=\det(e)^2 Q^{-1}$ for 
$(e,p)$ and $Q=Q^4 Q^{-3}=\det(E)^2 Q^{-3}$ for $(A,E)$, 
use (\ref{4.22}) and normal 
order) we note that for other dimensions e.g. 
\be \label{4.25}
[Y^\ast q]_{IJ}=
\delta_{ij} 
[Y^\ast e]^i_I [Y^\ast e]^j_J
\ee
cannot be written as the determinant of a submatrix of $e_a^j$ and thus 
it cannot be the result of a $z-$integral which returns such a determinant as 
a result of the Jacobean when changing integration variables. Thus we must 
use different techniques. For instance in $D=3$ we can use for $d=1$ and 
$d=2$ respectively
\ba \label{4.26}
&& \sqrt{\delta_{ij}[Y^\ast e]^i [Y^\ast e]^j}\;^{-1}=
\int_{\mathbb{R}^3}\; \frac{d^3z}{2\pi^2\; ||z||^2}\; e^{i [Y^\ast e]^j z_j}
\nonumber\\
&& \sqrt{\delta^{ij}[Y^\ast E]_i [Y^\ast E]_j}\;^{-1}=
\int_{\mathbb{R}^3}\; \frac{d^3z}{2\pi^2\; ||z||^2}\; e^{i [Y^\ast E]_j z^j}
\ea
with $[Y^\ast E]_j=\frac{1}{2}\epsilon^{IJ}\epsilon_{abc} Y^a_{,I} 
Y^b_{,J} E^c_j$ which is the integral that expresses the Green function 
of the Laplace operator as the Fourier of $||z||^{-2}$ up to a factor. 
Then we can use the relation $E^a_j=Q e^a_j,\; e^a_j e_a^k=\delta_j^k$ 
to relate $E, e$ and then use (\ref{4.22}), (\ref{4.26}). Then 
the length of a curve or the area of a surface can be obtained 
by multiplying (\ref{4.26}) with the polynomials
$||Y^\ast e||^2, ||Y^\ast E||^2$ respectively, normal ordering and integrating 
over the domain of the embedding $Y$. Other cases
can be treated similarly.

\subsection{Master constraint as a quadratic form on Fock space}
\label{s4.3}

The polynomioal version of the master constraint is simply obtained 
by normal ordering of its integrand. The integrand of the non-polynomial 
version can be written as a polynomial times $Q^{-N}$ for sufficiently 
large $N$. We apply proposition \ref{prop4.1} and normal order 
the resulting expression.

\subsection{Reduced Hamiltonian as a quadratic form on Fock space}
\label{s4.4}

We now apply this theory to the reduced Hamiltonian which as we have seen 
in section \ref{s2} typically takes the form 
\be \label{4.16}
H=-\int\; d^Dx\; \sqrt{-2\;Q\; \tilde{C}}
\ee
where $C=\frac{\pi_0^2}{2Q}+\tilde{C},\; Q=\sqrt{\det(q)}$ 
is the full Hamiltonian constraint and $\tilde{C}$ is 
the contribution to $C$ from geometry and matter fields other than 
$\phi^0,\pi_0$ with gauge fixing $\phi^0=t$. We pick a dense 
form domain based on a coherent state $\Omega_Z$ as above and define 
the quantity
\be \label{4.17}
\Lambda_0:=-<\Omega_Z, \tilde{c}\; \Omega_Z>,\; 
\tilde{c}:=\frac{\tilde{C}}{Q}
\ee
where the spatial scalar 
$\frac{\tilde{C}}{Q}$ is quantised as a quadratic form in Fock 
representations for all geometry and matter fields by the above methods using 
the fact that $\tilde{C}/Q=[\tilde{C}\; Q^N]/Q^{N+1}$ and that 
$\tilde{C} Q^N$ is a polynomial for sufficiently large $N$ (pick the 
minimal such $N$) and using normal ordering. Assuming that $\Lambda_0>0$
is a positive function on $\sigma$ and that 
$\Delta=1+\frac{\tilde{c}}{\Lambda_0}$ is a small quantity 
we use the approximant derived in section (\ref{s2}) depending 
on the interpolation parameter $k$ or the truncation parameter $N$
and normal order, e.g. 
\be \label{4.18}
H_{k,Z}:=\int\; d^Dx\; \sqrt{2\Lambda_0}\;
: Q[-1+\Delta + k\; \Delta^2] :
\ee
with $0\le k\le 1$ 
where $:Q\;\Delta^n:, n=1,2$ is defined by decomposing into the terms 
$:\tilde{C}^l\; Q^{1-l}:,\; l=0,1,2$ which are treated by the above methods. 
The terms independent of $k$ are 
\be \label{4.19}
H_{1,Z}=\int\;d^Dx\; \frac{1}{\sqrt{2\Lambda_0}}\;
[:\tilde{C}:-\Lambda_0\;:Q:]
\ee
We have made the dependence of the approximant on $Z$ via
$\Lambda_0=\Lambda_0(Z)$ explicit. Similar remarks concern the 
treatment of the higher order 
powers $\Delta^n$ in the N-th order truncation of the Hamiltonian. 

Note that $<\Omega_Z,\Delta\Omega_Z>=0$ by construction and thus one 
expects the fluctuation term $<\Omega_Z,[H_{k,Z}-H_{1,Z}]\Omega_Z>$ to be 
subleading.
If $\tilde{C}=\Lambda Q+\hat{C}$ where $\Lambda>0$ is a ``bare'' cosmological 
constant term and $\hat{C}$ the remainder then 
\be \label{4.20}
:\tilde{C}:-\Lambda_0\;:Q:=
:\hat{C}:+[\Lambda-\Lambda_0]\;:Q:
\ee
where $\epsilon:=\Lambda-\Lambda_0$ is the dressed ``cosmological function'' 
which is a 
constant on sufficiciently large scales for suitable $\Omega_Z$. It can be 
positive since we have 
$<:\tilde{C}:>=<:\hat{C}:>+\Lambda <:Q:>= -\Lambda_0(1-\delta) <:Q:>$ 
with $\delta=<:\Delta Q:>/<:Q:>$, 
i.e.  for e.g. $\epsilon=\Lambda_0\delta$ we have  
$<:\hat{C}:>= -[(2-\delta)\Lambda_0+\epsilon] <:Q:>=-2\Lambda_0 
<:Q:>$  is negative 
due to the negative gravitational contribution.\\
\\   
This sketches how the above methods can be used in order to write 
%the N-th order truncation 
$k$ approximation 
of the exact square root Hamitonian 
as a quadratic form in Fock representations for both geoemtry and 
matter based 
on the background metric $g$ and using normal ordering and the 
expression for $Q^N$ provided in proposition \ref{prop4.1} for $Y=$id.  
One expects that the matrix elements are expecially simple if $q(Z)=g$ 
i.e. when the coherent state metric equals the background metric.

\subsection{Towards operators}
\label{s4.5}

As quadratic forms cannot be multiplied, it is not possible 
to evaluate e.g. ${\sf Vol}[S]^2$ on the form domain. However introducing 
modes $e_I$ as in section \ref{s3} we can define actual operators such 
as ${\sf Vol}_M[S]$ which are obtained by decomposing the part 
of ${\sf Vol}[S]$ which is polynomial in annihilation and creation operators,
i.e. the piece $:Q^Y_2:$ in (\ref{4.15}), 
into modes $I$ and to restrict the sum by $I<M$. 
Note that this is different from 
$P_{M,Z}\; {\sf Vol}[S]\; P_{M,Z}$ where $P_{M,Z}$ projects 
on the (closure of the) subspace ${\cal D}_{M,Z}$ given by the span
of excitations of $\Omega_Z$ by 
polynomials in the 
$<f,A>^\ast$ where $<e_I,f^{ij}>=0,\; I>M$, because 
$\Omega_Z$ is not annihilated by $<f,A>$ (rather it returns $<f,Z>$).
Then ${\sf Vol}_M[S]$ in fact preserves ${\cal D}_{M,Z}$ because the 
resticted polynomial dependence is now a polynomial in 
annihilation and creation operators $A_I^\ast, A_I,\;I<M$ and no longer 
an infinite series. This however does not 
allow to remove the cutoff in contrast to the commutator algebra 
where there are cancellations. \\
\\
To actually remove the regulator requires either the counterterm 
method familiar from the perturbative construction of the S matrix 
in usual QFT or methods from constructive QFT \cite{34}, see also 
\cite{29}. The latter 
idea can be sketched as follows (we drop the label $k,Z$ for simplicity): 
Consider the truncation $H_M$
%=P_M \; H\; P_M$ 
as defined above and try to find an invertible ``dressing operator'' 
$T_M: {\cal D}_0 \to {\cal D}_0$ where ${\cal D}_0$ is the 
span (\ref{4.7}) such that 
$H_M\; T_M = T_M h_M$ where $h_M$ has a strong limit $h$ in the topology 
of ${\cal H}_0$, the Fock space completion of 
${\cal D}_0$. One defines new scalar products
$<T_M \psi,T_M\psi'>_M:=<T_M \psi,T_M\psi'>/<T_M\Omega_Z,T_M\Omega_Z>$
and $\hat{{\cal D}}_0$ as the subspace of 
${\cal D}_0$ such that these 
inner producs converge as $M\to \infty$ thereby defining a new 
maps $T:\hat{{\cal D}}_0\to {\cal D}$ and a new 
Hilbert space $\cal H$ as the completion  of $\cal D$. Then 
$H T=T h$ if $h$ has range in $\hat{{\cal D}}_0$ which defines $H$ 
densely on $\cal D$ with range in $\cal H$.\\ 
\\
Another possibility is a variant of the Friedrichs extension technique
\cite{22}:\\
Suppose that one could show that the quadratic form
defining the physical Hamiltonian $H$, which is automatically symmetric
on the physical Hilbert space $\cal H$, 
is bounded from below on a dense subspace $V\subset {\cal H}$ of 
the physical Hilbert space and such 
that it is closable. That is, there exists $k\ge 0$ such that 
$||\psi||_k^2:=<\psi,H\psi>_{{\cal H}}+k||\psi||_{{\cal H}}^2\ge 0$ 
for all $\psi\in V$ and such that $V$ is in fact complete in the 
norm $||.||_{k+1}$, that is, $(V,<.,.>_{k+1})$ is a Hilbert space. 
Then we can define an operator $H'$ on $\cal H$ with 
dense domain ${\cal D}'=R_k({\cal H})$ where $R_k(\psi)\in V,\;\psi\in 
{\cal H}$ is the Riesz representative of the $||.||_{k+1}$ continuous
linear form $f\in V\mapsto <\psi,f>_{{\cal H}}$. Then for 
$\hat{\psi}_1',\hat{\psi}_2'\in {\cal D}'$ the formula 
$<\psi_1', H'\psi_2'>_{{\cal H}}:=
<\psi_1', H\psi_2'>_{{\cal H}}$ defines the matrix elements of a 
self-adjoint operator bounded from below by the same bound 
\cite{22} and we 
can use the spectral theorem to construct the scattering matrix
corresponding to $H'$.

\section{Conclusion}
\label{s5}

The main result of the present manuscript is to demonstrate that, 
although counter intuitive,
background dependent Fock representations can be used in a non perturbative 
fashion in order to construct a non-perturbative theory of quantum gravity.
In non-perturbative quantum gravity there is no natural split of the 
Hamiltonian (constraint) into a ``free'' and an ``interacting'' part.
Even in the polynomial version of the Hamiltonian constraint there 
is no term which is of quadratic order in the fields and the reduced 
Hamiltonian depends highly non-polynomially on the fields anyway.  

In perturbative quantum gravity 
one artificially decomposes the full metric $q_{ab}$ into a classical 
background $g_{ab}$ and a ``graviton'' 
fluctuation $h_{ab}=q_{ab}-h_{ab}$ and one 
expands e.g. the reduced Hamiltonian in powers of $h_{ab}$ which is 
an infinite series. 
Then one resorts to the formalism of perturbative 
QFT of $h_{ab}$ on the background defined by $g_{ab}$ by taking 
the free part (quadratic part of the expansion) of the Hamiltonian 
as an input for the choice of the Fock representation of $h_{ab}$. This
is quite different from what we have done here because we never 
perform the split $q_{ab}=g_{ab}+h_{ab}$, we quantise the full metric 
$q_{ab}$ and just use the background $g_{ab}$ to define a Fock 
representation of the full $q_{ab}$ (e.g. in order to define a background 
Laplacian etc.). In our formalism the fields $g_{ab}, q_{ab}$ are 
totally independent. While we also have to use an expansion of 
the reduced Hamiltonian in order to deal with the square root, that 
expansion is of a different nature than the one used in perturbative 
quantum gravity because it does not use the auxiliary, classical 
background field $g_{ab}$ as zeroth order but rather the expectation 
value of the argument of the quare root with respect to a coherent 
state $\Omega_Z$ which is concentrated on a classical Euclidian signature 
metric $q^0_{ab}$ constructed from $Z$.
This metric is again is completely unrelated to $g_{ab}$. The excitations 
of $\Omega_Z$ by smeared polynomials of the quantum metric $q_{ab}$ 
define a dense form domain of the resulting Hamiltonian quadratic form. 
The expansion is in powers of the fluctutations of the full argument
of the square root which itself is not a polynomial but that non-polynomial
dependence on the quantum metric can be dealt with using suitable Fock 
space techniques developed in section \ref{s4} of the present paper.

The resulting reduced Hamiltonian is then defined as a densely defined 
quadratic form but is not an operator. 
To define an actual operator one must 
use non-pertuerbative constructive QFT methods. These use the notion
of mode cut-offs and were successfully employed to construct interacting 
scalar fields in 3d Minkowski space \cite{34}. We also used elements of 
this in order to construct the polynomial quantum constraints of a version of 
Euclidian quantum gravity as quadratic forms and even their commutator 
algebra which is possible because commutators contain differences of 
products of quadratic forms and while products of quadratic forms are 
ill-defined such differences of products, carefully defined using 
mode cut-offs, can define well defined new quadratic forms.

The results of this paper have many applications of which we mention a 
few.\\ 
A.\\
In quantum cosmology and black hole perturbation theory one 
performs a split $q_{ab}=q^0_{ab}+q^1_{ab}$ where $q^0_{ab}$ 
is the ``symmetric'' part of the metric with respect to the given 
symmetry (spatial homogeneity, spherical symmetry, axi-symmetry,..)
and $q^1_{ab}$ is the ``non-symmetric'' remainder. However, instead of 
treating $q^0,q^1$ as $g,h$ of the perturbative approach to quantum gravity 
one can quantise the background $q^0$ as well in order to study 
cumulative backreaction effects. Thus one wants to keep track of how 
the symmetric part interacts with the non-symmetric part. This viewpoint
is in between the totally perturbative approach in which $g=q^0$ is fixed 
and just $h=q^1$ is quantised and the totally non-perturbative approach 
in which the full $q=q^0+q^1$ is quantised without making any difference 
between $q^0, q^1$. For this reason it is sometimes called 
the ``hybrid'' approach \cite{35,36}. The point is now that the 
perturbative expansion of reduced 
Hamiltonian is by construction a polynomial in $q^1$ but retains a 
non-polynomial 
dependence on $q^0$. We can therefore apply the non-perturbative 
Fock quantisation developed in this paper to the $q^0$ sector while we 
apply the usual perturbative Fock quantisation to the $q^1$ sector.\\
B.\\
A natural regularisation of all the quadratic forms constructed in this 
paper such as ${\sf Vol}[S], H_{k,Z}$ is obtained by considering 
$H_{k,Z,M}$ which is obtained from $H_{k,Z}$ by decomposing the polynomial
dependence of $H_{k,Z}$ in terms of annihilation and creation operators 
and modes $I$ and to restrict the sum over modes by $I<M$ with mode cut-off 
$M$ on the ONB $e_I$ of $L_2(\sqrt{\det(g)} d^Dx,\sigma)$. Then 
$H_{k,Z,M}$ preserves the span
${\cal D}_{M,Z}$ 
of excitations of a coherent state $\Omega_Z$ where the test functions 
$f$ of excitations $<f,A>^\ast$ have the property $<f,e_I>=0, I>M$. 
Here $Z$ is a point in the classical phase space subject to the 
condition that the metric $q[Z]$ that is encoded by it has Euclidian signature.
This is a real space regularisation independent of Fourier transform 
techniques which in principle 
works on any $\sigma$. It provides a UV cut-off as 
$M\to \infty$ corresponds to infinite spatial resolution. Thus at finite 
$M$ this provides\\ 
1. a well defined theory of quantum gravity in interaction with 
matter\\
2. using ordinary Fock space techniques combined with new non-perturbative 
techniques, in particular all Hilbert spaces are separable\\
3. all gauge invariance has been removed.\\
Thus the powerful machinery of Fock spaces, Feynman diagrammes etc. 
can be applied e.g. in constructing a scattering matrix. We expect 
the theory to be especially simple when the coherent state label 
$q[Z]$ coincides with the background metric $g$ used to construct 
the Fock representation.

To take the $M\to\infty$ limit and define actual operators 
one needs to invoke renormalisation
similar to \cite{34} or using Hamiltonian/Wilsonian renormalisation 
(see \cite{19} and references therein) based on 
the family of theories $({\cal H}_{M,Z}, H_{k,Z,M})$.\\
C.\\
If one takes the quadratic form $H_{k,Z}$ as it is (no cut-off)
one can construct 
the scattering matrix from it by the usual methods of 
QFT (Gell-Mann Low formula) if one 
can usefully decompose $H_{k,Z}$ into free and interacting 
terms. However the free term cannot just be a quadratic term 
because it does not exist in this non perturbative approach.  
One expects however that for matrix elements between matter 
excitations of the vector $\Omega_Z$ 
the Hamiltonian is effectively a QFT on CST Hamiltonian for matter fields 
for which such a decomposition is possible. There will be corrections 
as compared to QFT in CST due to fluctuation effects of $\Omega_Z$ and 
due to its gravitational 
excitations.\\  
D.\\
The commutator calculation has been performed only for Euclidian 
vacuum gravity (however with cut-off removed!) in four dimensions in 
the $(A,E)$ polarisation. 
By the methods introduced in the present work, it is possible to repeat 
the calculation for both signatures and any matter content in 
all three polarisations discussed and even in any dimension 
if one works in the $(q,p)$ or $(e,P)$ polarisation. However the computational
effort is significantly higher.

\begin{appendix}

\section{Nomal ordered commutators of normal ordered monomials}
\label{sa}

In this section we display the results of the computation of commutators 
of normal ordered monomials of 
$F_\alpha=A_\alpha-A_\alpha^\ast,\; G_\alpha=A_\alpha+A_\alpha^\ast$
as a linear combination of such normal ordered monomials. That 
linear combination is unique given the annihilation and creation 
algebra $A_\alpha, A_\alpha^\ast,\;[A_\alpha, A_\beta^\ast]=
\delta_{\alpha\beta}\; E$ where $E$ is the unit algebra element
and $\alpha,\beta,..$ can be from any index set. Since 
normal ordered monomials of $F,\; G$ are separately 
totally symmetric in the labels of $F,G$ factors respectively, that 
symmetry must be displayed also y the r.h.s. of the calculation which is 
a good consistency check. A normal orderd 
monomial is defined by its order $m$ and the number $p=0,..,m$
of factors $F$. If one computes the commuator 
of monomials with data $(m,p),(n,q)$ respectively, the r.h.s. is a linear 
combination of normal ordered
monomials of order $m+n-2k,\;k=1,..[(m+n)/2]$ where 
$[.]$ is the Gauss bracket unless $m=n,p=q$ in which case $1\le k\le m-1$.
Each such monomial comes with $k$ factors of Kronecker $\delta$. 
The various powers of 2 can be attributed to first the commutator 
$[F_\alpha,G_\beta]=2\delta_{\alpha\beta}\;E$ and the symmetrisation 
operation 
$K_\alpha L_\beta+K_\alpha L_\beta=:2\; K_{(\alpha} L_{\beta)}$ for any 
objects $K,L$, i.e. there is one factor of 2 for each Kronecker $\delta$
and for each symmetrisation in 2 indices. Note that a factor of 2 Kronecker
$\delta$ symmetrised in one pair of indices is automatically symmetrised 
in the other pair which is why some of the terms apparently lack an expected 
symmetry. Clearly there are no Kronecker $\delta$ for pairs of indices 
carried by the same monomial on the left hand side.   

Thus it appears that one can almost uniquely guess the result 
of the right hand side by these simple rules if one simply writes 
down all possible normal ordered terms $F^{p-r} G^{m-p-s} F^{q-s} G^{n-q-r}$
following from 
$[:F^p G^{m-p}:,:F^q G^{n-q}:]$
for the terms involving $k=r+s$ Kronecker factors and then symmetrises 
and normal orders, the intuition being that one needs to take commutators 
$[F,G]$ with $F,G$ respectively taken from  $:F^p G^{m-p}:,
:F^q G^{n-q}:$ respectively (plus sign) or vice versa (minus sign).
However, this is not the case 
for the terms involving more than one Kronecker factor: 
For instance the double Kronecker
terms $+:G^2:,\; +:F^2:$ in the computation of $[:FG:,:F^2 G^2:]$ are 
unexpected, one 
would rather have expected a term $:FG:$ by this intuition which however
is absent. 

When doing the computation by hand, at least for  
low order monomials, the computational effort can 
be reduced by not using lemma \ref{la3.1} but rather succesively 
using identities such as 
\ba \label{a.0}
&& :F_\alpha\; G_\beta:=
A_\beta^\ast \; F_\alpha+F_\alpha\;A_\beta 
=\frac{1}{2}[
(G_\beta-F_\beta) \; F_\alpha+F_\alpha\;(G_\beta+F_\beta)]
=\frac{1}{2}[G_\beta\; F_\alpha+F_\alpha\; G_\beta]
\nonumber\\
&&
:F_\alpha\; G_\beta\; G_\gamma:=
A_\gamma^\ast\; :F_\alpha\; G_\beta: 
+:F_\alpha\; G_\beta:\; A_\gamma
\nonumber
\ea
etc. which has the advantage to express the normal ordered objects 
in terms of symmetrised objects just containing $F,G$ so that 
$[F,F]=[G,G]=[F,G]-2E=0$ can be exploited in the commutator calculation.
When the comutator has been calculated, which at this point 
contains only one Kronecker factor, one needs to rearrange the result 
into the terms in the above list that one has produced which requires
further commutations producing the additional powers of Kronecker symbols.

Although straightforward, this becomes quickly algebraically very involved,
the hand written computation involves an order of 20 pages. The 
end result still fits on a single page. For better readability we have 
written $E^k$ instead of $E$ for the term that involves $k$ Kronecker 
symbols.
   
\newpage

~\\
1. $[F,F]$
\begin{equation} \label{a.1}
[\;:F_\alpha:\;,\;:F_\mu:\;]=E\cdot\; 0
\end{equation}
2. $[F,FG]$
\begin{equation} \label{a.2}
[\;:F_\alpha:\;,\;:F_\mu\; G_\rho:\;]=2\; E\cdot\;\delta_{\alpha\rho}\;:F_\mu:
\end{equation}
3. $[F,FFG]$
\begin{equation} \label{a.3}
[\;:F_\alpha:\;,\;:F_\mu\; F_\nu \; G_\rho:\;]=
2\;E\cdot\; \delta_{\alpha\rho}\;:F_\mu\; F_\nu:
\end{equation}
4. $[F,FFGG]$
\begin{equation} \label{a.4}
[\;:F_\alpha:\;,\;:F_\mu\; F_\nu \; G_\rho\;G_\sigma:\;]=
4\;E\cdot\;\delta_{\alpha(\rho}\;:G_{\sigma)}\;F_\mu\; F_\nu:
\end{equation}
5. $[FG,FG]$
\begin{equation} \label{a.5}
[\;:F_\alpha\;G_\gamma:\;,\;:F_\mu\; G_\rho:\;]=2\;E\cdot\;
\left(
\delta_{\alpha\rho}\;:F_\mu G_\gamma:
-\;\delta_{\mu\beta}\;:F_\alpha G_\rho:
\right)
\end{equation}
6. $[FG,FFG]$
\begin{eqnarray} \label{a.6}
&& [\;:F_\alpha\;G_\gamma:\;,\;:F_\mu\;F_\nu\; G_\rho:\;]
\\
&=& 2\;
\{
E\cdot \left( \delta_{\alpha\rho}\;:F_\mu\;F_\nu\;G_\gamma: 
-2\;\delta_{\gamma(\mu}\;:F_{\nu)}\; F_\alpha\; G_\rho:
\right)
+E^2\cdot \left( 
2\;\delta_{\alpha(\mu}\;\delta_{\nu)\gamma} \; :G_\rho: \right)
\}
\nonumber
\end{eqnarray}
7. $[FG,FFGG]$
\begin{eqnarray} \label{a.7}
&& [\;:F_\alpha\;G_\gamma:\;,\;:F_\mu\;F_\nu\; G_\rho\;G_\sigma:\;]
\\
&=& 4\;
\{
E\cdot \left(
\delta_{\alpha(\rho}\;:G_{\sigma)}\; F_\mu\;F_\nu\;G_\gamma:
-\;\delta_{\gamma(\mu}\;:F_{\nu)}\; F_\alpha\; G_\rho\; G_\sigma:
\right)
\nonumber\\
&& +E^2\cdot \left(
\delta_{\alpha(\mu}\;\delta_{\nu)\gamma} \; :G_\rho\; G_\sigma:
+\;\delta_{\alpha(\rho}\;\delta_{\sigma)\gamma} \; :F_\mu\; F_\nu:
\right)
\}
\nonumber
\end{eqnarray}
8. $[FFG,FFG]$
\begin{eqnarray} \label{a.8}
&& [\;:F_\alpha\;F_\beta\;G_\gamma:\;,\;:F_\mu\;F_\nu\; G_\rho:\;]
\\
&=& 4\;
\{
E\cdot\left(
\delta_{\rho(\alpha}\;:F_{\beta)}\; F_\mu\;F_\nu\;G_\gamma:
-\;\delta_{\gamma(\mu}\;:F_{\nu)}\; F_\alpha\; F_\beta \;G_\rho:
\right)
\nonumber\\
&& + E^2 \cdot\left(
\delta_{\gamma(\mu}\;\delta_{\nu)(\alpha}\; :F_{\beta)}\; G_\rho:
-2\;\delta_{\rho(\alpha}\;\delta_{\beta)(\mu} \; :F_{\nu)}\; G_\gamma:
\right)
\}
\nonumber
\end{eqnarray}
9. $[FFG,FFGG]$
\begin{eqnarray} \label{a.9}
&& [\;:F_\alpha\;F_\beta\;G_\gamma:\;,\;:F_\mu\;F_\nu\; G_\rho\; G_\sigma:\;]
\\
&=& 4\;
\{
E \cdot \left(
- \; \delta_{\gamma(\mu}\; :F_{\nu)}\; F_{\alpha}\; F_{\beta}\;  
G_{\rho}\; G_{\sigma}: 
+ 2\; :F_{(\alpha}\; \delta_{\beta)(\rho}\; G_{\sigma)}\; 
F_{\mu}\; F_{\nu}\; G_{\gamma}: 
\right)
\nonumber\\
&& + E^{2} \cdot \left(
2\; :F_{\mu}\; F_{\nu}\; F_{(\alpha}:\; \delta_{\beta)(\rho}\; 
\delta_{\sigma)\gamma} 
+ 2\; :G_{\rho}\; G_{\sigma}\; F_{(\alpha}:\; \delta_{\beta)(\mu}\; 
\delta_{\nu)\gamma} 
- 4\; :G_{\gamma}\; G_{(\rho}\; \delta_{\sigma)(\alpha}\; 
\delta_{\beta)(\mu}\; F_{\nu)}: 
\right) 
\nonumber\\
&& + E^{3} \cdot \left(
- 4\; :F_{(\mu}:\; \delta_{\nu)(\alpha}\; \delta_{\beta)(\rho}\; 
\delta_{\sigma)\gamma} 
- 2\; :F_{(\mu}:\; \delta_{\nu)\gamma}\; \delta_{\alpha(\rho}\; 
\delta_{\sigma)\beta}  
\right) 
\}
\nonumber
\end{eqnarray}
10. $[FFGG,FFGG]$
\begin{eqnarray} \label{a.10}
&& [\;:F_\alpha\;F_\beta\;G_\gamma\; G_\delta:\;,
\;:F_\mu\;F_\nu\; G_\rho\; G_\sigma:\;]
\\
&=& 8\;
\{
E \cdot \left(
-\; :F_{\alpha}\; F_{\beta}\; G_{\rho}\; G_{\sigma}\; G_{(\gamma}\; \delta_{\delta)(\mu}\; F_{\nu)}: 
+\; :F_{\mu}\; F_{\nu}\; G_{\gamma}\; G_{\delta}\; G_{(\rho}\; \delta_{\sigma)(\alpha}\; F_{\beta)}: 
\right)
\nonumber\\
&& + E^{2} \cdot (
-\; 2\; :F_{\alpha}\; F_{\beta}\; F_{(\mu}\; \delta_{\nu)(\gamma}\; 
\delta_{\delta)(\rho}\; G_{\sigma)}: 
-\; 2\; :G_{\gamma}\; G_{\delta}\; F_{(\mu}\; \delta_{\nu)(\alpha}\; 
\delta_{\beta)(\rho}\; G_{\sigma)}: 
\nonumber\\
&& +\; 2\; :F_{\mu}\; F_{\nu}\; G_{(\gamma}\; \delta_{\delta)(\rho}\; 
\delta_{\sigma)(\alpha}\; F_{\beta)}: 
+\; 2\; :G_{\rho}\; G_{\sigma}\; G_{(\gamma}\; \delta_{\delta)(\mu}\; 
\delta_{\nu)(\alpha}\; F_{\beta)}: 
) 
\nonumber\\
&& + E^{3} \cdot (
4\; :F_{(\alpha}\; \delta_{\beta)(\mu}\; \delta_{\nu)(\gamma}\; 
\delta_{\delta)(\rho}\; G_{\sigma)}: 
+\; 2\; :F_{(\alpha}\; \delta_{\beta)(\rho}\; G_{\sigma)}:\; 
\delta_{\mu(\gamma}\; \delta_{\delta)\nu} 
\nonumber\\
&& -\; 4\; :F_{(\mu}\; \delta_{\nu)(\alpha}\; \delta_{\beta)(\rho}\; 
\delta_{\sigma)(\gamma}\; G_{\delta)}: 
-\; 2\; :F_{(\mu}\; \delta_{\nu)(\gamma}\; G_{\delta)}:\; 
\delta_{\alpha(\rho}\; \delta_{\sigma)\beta} 
) 
\}
\nonumber
\end{eqnarray}

\end{appendix}

\end{document}